\documentclass[a4paper,onecolumn, accepted=2023-09-03]{quantumarticle} 
\pdfoutput=1
\usepackage[english]{babel}
\usepackage[utf8]{inputenc}
\usepackage{fec}
\usepackage[a4paper, total={6.5in, 8.6in}]{geometry}
\usepackage[algo2e]{algorithm2e} 
\providecommand{\SetAlgoLined}{\SetLine}

\usepackage{mathtools}
\usepackage{amsmath}
\usepackage{graphicx}
\usepackage{booktabs}
\usepackage{todonotes}
\newtheorem{theorem}{Theorem}
\newtheorem{definition}{Definition}
\newtheorem{lemma}{Lemma}
\newtheorem{corollary}{Corollary}
\newtheorem{proposition}{Proposition}

\newtheorem{remark}{Remark}

\newcommand{\tran}[1]{{#1}^{\top}}
\newcommand{\itran}[1]{{#1}^{-\top}}

\usepackage{xcolor}

\usepackage{authblk}

\newcommand{\bra}[1]{\left\langle{#1}\right|}
\newcommand{\ket}[1]{\left|{#1}\right\rangle}

\newcommand{\svec}{\ensuremath{\textbf{svec}}}
\newcommand{\smat}{\ensuremath{\textbf{smat}}}
\renewcommand{\vec}{\ensuremath{\textbf{vec}}}

\newcommand{\tr}{\ensuremath{\widetilde{R}}}

\newcommand{\wtR}{\ensuremath{\widetilde{R}}}
\newcommand{\wtX}{\ensuremath{\widetilde{X}}}
\newcommand{\wtS}{\ensuremath{\widetilde{S}}}
\newcommand{\wtDX}{\ensuremath{\widetilde{\Delta X}}}
\newcommand{\wtDS}{\ensuremath{\widetilde{\Delta S}}}

\newcommand{\overbar}[1]{\mkern 1.5mu\overline{\mkern-1.5mu#1\mkern-1.5mu}\mkern 1.5mu}

\usepackage{fixltx2e}
\usepackage[numbers,sort&compress]{natbib}

\author[1]{Brandon Augustino}\email{bra216@lehigh.edu}
\author[2]{Giacomo Nannicini}\email{g.nannicini@usc.edu}
\author[1]{Tam\'as~Terlaky}\email{terlaky@lehigh.edu}
\author[1]{Luis Zuluaga}\email{luis.zuluaga@lehigh.edu}

\affil[1]{\textit{Department of Industrial and Systems Engineering, Quantum Computing and Optimization Lab, Lehigh University}}
\affil[2]{\textit{Department of Industrial and Systems Engineering, University of Southern California}}

\begin{document}
\title{Quantum Interior Point Methods for Semidefinite Optimization}
\maketitle

\begin{abstract}
    We present two quantum interior point methods for semidefinite optimization problems, building on recent advances in quantum linear system algorithms. The first scheme, more similar to a classical solution algorithm, computes an inexact search direction and is not guaranteed to explore only feasible points; the second scheme uses a nullspace representation of the Newton linear system to ensure feasibility even with inexact search directions. The second is a novel scheme that might seem impractical in the classical world, but it is well-suited for a hybrid quantum-classical setting. We show that both schemes converge to an optimal solution of the semidefinite optimization problem under standard assumptions. By comparing the theoretical performance of classical and quantum interior point methods with respect to various input parameters, we show that our second scheme obtains a speedup over classical algorithms in terms of the dimension of the problem $n$, but has worse dependence on other numerical parameters.
\end{abstract}

	\clearpage
\tableofcontents

	\clearpage

\section{Introduction}
We develop quantum interior point methods (QIPMs) for
the solution of semidefinite optimization problems (SDOPs). Letting $b
\in \R{m}$, matrices $A^{(1)} , \dots, A^{(m)}, C \in \mathcal{S}^n$,
where $\mathcal{S}^n$ is the space of $n \times n$ symmetric
matrices, we consider the primal SDOP given as
\begin{align}\label{e:SDO}
     z_P &= \inf_X \left\{ C \bullet X : A^{(i)} \bullet X  = b_i,~\forall i \in [m], X \succeq 0 \right\},
\end{align}
where $[m] = \{1, \dots, m \}$, $U \bullet V = \trace (UV)$ defines the trace inner product of two matrices $U$ and $V$, and $U \succeq V$ indicates that $U - V$
is a symmetric positive semidefinite matrix. We assume that
the matrices $A^{(1)} , \dots, A^{(m)}$ are linearly independent. The corresponding dual problem is given by
\begin{align}\label{e:SDD}
    z_D &=  \sup_{y, S} \left\{ b^{\top} y:\sum_{i=1}^m y_i A^{(i)} + S = C,~S\succeq 0, y \in \R{m} \right\}. 
\end{align}
 The so called Interior Point Condition (IPC) holds if there is a primal feasible $X$ with $X \succ 0$, and there exists a dual feasible $(y,S)$ with $S \succ 0$. Although duality results are generally weaker for SDO when compared to linear optimization (LO), under the IPC strong duality holds; optimal solutions exist for both the primal and dual SDOPs, and their optimal objective values are equal, i.e., $z_P = z_D$. 

Interior point methods (IPMs) are algorithms to solve convex optimization problems, including semidefinite optimization problems, in polynomial time. The first polynomial-time algorithm for linear optimization problems (LOPs), a special class of convex optimization problems, is the ellipsoid method \cite{khachiyan1980polynomial}. Karmarkar's IPM improved on the complexity of the ellipsoid method \cite{karmarkar1984new}. The first polynomial-time IPM for SDOPs is attributed to Nesterov and Nemirovskii \cite{nesterov1988general, nesterov1994interior}. Their framework employs efficiently-computable self-concordant barrier functions.  

The interest in SDO originates from the fact that SDO provides a framework that allows one to formulate fundamental problems in control \cite{boyd1994linear}, statistics, information theory \cite{rains2001semidefinite}, machine learning \cite{lanckriet2004learning, weinberger2006unsupervised}, finance \cite{d2007direct, wolkowicz2012handbook}, and quantum information science \cite{eldar2003semidefinite, harrow2017improved, watrous2009semidefinite}, to name a few. Additionally, SDO can provide tight approximations for various combinatorial optimization problems \cite{goemans1995improved, lovasz1979shannon}, and it can be used to study the properties of convex optimization problems \cite{boyd1994linear}. SDOPs are attractive since they exhibit practical efficiency \cite{andersen2000mosek, sturm1999using, toh1999sdpt3} and
tractability \cite{alizadeh1998primal, nesterov1988general}.  LOPs are a special case of SDOPs in which each of the input matrices $C, A^{(1)}, \dots, A^{(m)}$ are diagonal matrices. Considering the vast number of applications of SDO, it is natural to ask if quantum computers can accelerate their solution. Several works have investigated this possibility, which we discuss next. 

\subsection{Literature Review}
When $m \leq \sqrt{n}$, the Cutting Plane Methods of \cite{jiang2020improved, lee2015faster} are the best performing classical algorithms for solving SDOPs, and can do so in time
$$\Ocal \left( m ( mn^2 + m^2 + n^{\omega}) \cdot \textup{polylog} \left(m,n, R, \frac{1}{\epsilon} \right) \right),$$
where $\omega \in [2, 2.38]$ is the matrix multiplication exponent, $R$ is an upper bound on the trace of primal optimal solutions, $\epsilon$ is the precision parameter, and $mn^2$ is the input size of the SDO problem. However, in practical situations we typically have that $m$ is roughly between $n$ and $n^2$, in which case the Cutting Plane Methods given in \cite{jiang2020improved, lee2015faster} are outperformed by the IPM for SDO from Jiang~et~al.~\cite{jiang2020faster}. Their IPM exhibits a worst case running time of 
$$\Ocal \left( \sqrt{n} (mn^2 + m^{\omega} +  n^{\omega}) \cdot \textup{polylog} \left(m,n, \frac{1}{\epsilon} \right) \right),$$
where the term $m^{\omega} + n^{\omega}$ represents the per-iteration cost of inverting the Hessian and dual slack matrices. For IPMs, $\epsilon$-optimality implies that the primal and dual feasible solutions exhibit a \textit{normalized} duality gap bounded by $\epsilon$, i.e.:
    $$\frac{X \bullet S}{n} \leq \epsilon.$$ One can also note that when $m = \Ocal (n^2)$, the running time of the IPM in \cite{jiang2020faster} is shown to be $\widetilde{\Ocal}_{\frac{1}{\epsilon}} \left( n^{5.246} \right)$, providing an enhancement over the running time bound of $\widetilde{\Ocal}_{\frac{1}{\epsilon}} \left( n^{6.5} \right)$ found in \cite{monteiro1998polynomial, nesterov1997self, nesterov1998primal}. Here, the notation $\widetilde{\Ocal}_{\alpha, \beta} (f(x))$ suppresses polylogarithmic factors in $\alpha$ and $\beta$ in the running time, i.e., $\widetilde{\Ocal}_{\alpha, \beta} (f(x)) =\Ocal(f(x) \cdot \textup{polylog}(\alpha, \beta, f(x)))$.

Although IPMs are prevalent in practice, they are not the only
approach to solve SDOPs. The classical matrix multiplicative weights
update (MWU) algorithm of Arora and Kale
\cite{arora2006multiplicative} is one such algorithm that has been
quantized with notable gains. The classical MWU method has a running
time of
$$ \mathcal{O} \left(mn^2 \cdot \textup{poly} \left( \frac{Rr}{\epsilon_{\text{abs}}}
\right) \right),$$
where $r$ is an upper bound on the $\ell_1$-norm of dual optimal solutions, and $\epsilon_{\text{abs}}$ is an additive error to which the optimal objective value is approximated. That is, the objective value attained by the acquired solution is $\text{OPT} \in [\varsigma- {\cal O}(\epsilon_{\text{abs}}), \varsigma + {\cal O}(\epsilon_{\text{abs}})]$, where $\varsigma$ is a bound on the optimal objective value determined using binary search. The MWU algorithm alternates
between candidate solutions to the primal and dual SDO problems, and
does not involve the solution of linear systems.
A quantum SDO solver based on the Arora-Kale framework was constructed
by Brand{\~a}o and Svore \cite{brandao2017quantum}, and van Apeldoorn
et al.\ \cite{van2020quantum}. With successive improvements
\cite{gribling2019applications,van2018improvements}, the running time
of the quantum MWU algorithm for SDOPs has been brought down to:
$$ \mathcal{O} \left( \left(\sqrt{m} + \sqrt{n}
\frac{Rr}{\epsilon_{\text{abs}}}\right)s \left(\frac{Rr}{\epsilon_{\text{abs}}}\right)^4  \cdot \textup{polylog} \left(m,n, R, \frac{1}{\epsilon_{\text{abs}}} \right)
\right),$$
 where $s$ is a sparsity parameter (maximum number of non-zero
entries per row). While these methods achieve running times with a more
favorable dependence on $m$ and $n$ as compared to classical IPM
solvers, their polynomial dependence on the parameters $R$, $r$ and
$1/\epsilon_{\text{abs}}$ is exponentially slower than classical IPMs, which exhibit logarithmic dependence on the inverse precision. (If $\epsilon_{\text{abs}}$ is
considered to be part of the input, it is expressed in binary format and a running time of $1/\epsilon_{\text{abs}}$ is not polynomial in the usual sense.)

The QIPM presented by Kerenidis and Prakash
\cite{kerenidis2020quantum} (rather, the earlier arXiv
version of that paper) initiated the other major research direction
regarding the solution of SDOPs with quantum algorithms. By
constructing block-encodings of the Newton linear system matrix, which
can be efficiently constructed at runtime using quantum random access
memory (QRAM), the authors in \cite{kerenidis2020quantum} solve the
Newton linear system in time polylogarithmic in $n$ (but note that this does not include the time to obtain a classical description of the solution). As the solution
of the Newton linear system is the most expensive step in classical
IPMs, the hope is to obtain a quantum
advantage. The paper \cite{kerenidis2020quantum} posits a worst case running time
of
$$ \Ocal \left( \frac{n^{2.5}}{\xi^2} \mu \kappa^3 \cdot \textup{polylog} \left(n, \kappa, \frac{1}{\epsilon} \right) \right)$$ for SDOPs, where $\epsilon$
is the optimality gap tolerance to which we solve the SDOP, $\xi$ is a
feasibility tolerance, $\kappa$ is an upper bound for the condition
number of the Newton systems, and $\mu$ is upper bounded by the
largest Frobenius norm of the coefficient matrix of the Newton systems. Note, however, that the algorithm of \cite{kerenidis2020quantum} does not converge to an $\epsilon$-optimal solution in the usual sense, for reasons that will be discussed below; thus, this running time expression does not settle the complexity of QIPMs for SDOPs. The term $ \Ocal
\left( \frac{n^2 \kappa^2}{\xi^2}\right)$ comes from a tomography
subroutine that is used to map the quantum-mechanical encoding of the
solution of the Newton linear system, to a classical description used to
construct the subsequent iterate. This reconstruction is performed in
each step to $\ell_2$-norm error $\xi/\kappa$. When this algorithm is
applied to solve LOPs, the running time is
$$ \Ocal \left( \frac{n^{1.5}}{\xi^2} \mu \kappa^3 \cdot \textup{polylog} \left(n, \kappa, \frac{1}{\epsilon} \right)\right).$$ Similar ideas have been
used in other recent papers
\cite{casares2020quantum, kerenidis2019quantum}, all of which follow
the same structure: they apply a classical interior point method, and
substitute the classical solution of the Newton linear system (or
equivalent) with a quantum linear system algorithm, with the goal of
improving the running time. The solution of the linear system is
extracted from the quantum computer via tomography. The analysis of
the algorithms requires nontrivial machinery and has led to
interesting developments. One limitation of some of these existing QIPM works is the use of the feasibility tolerance $\xi$, that does not translate to the usual concept of primal feasibility. Another limitation is that the analysis of
\cite{kerenidis2020quantum} is heavily based on classical exact IPMs,
and does not account for some challenges that arise when translating
those algorithms to the quantum setting. Similar challenges are found
in \cite{casares2020quantum, kerenidis2019quantum}. These challenges
are discussed below, and addressed in this paper for the specific case
of semidefinite optimization. 

The QIPM of \cite{kerenidis2020quantum} directly quantizes the
classical IPM described in \cite{ben2001lectures}. This leads to two
major issues. The first is that this algorithm does not guarantee that
the solutions to the Newton linear system remain symmetric
\cite{mj1999study}. As a consequence, the algorithm described in
\cite{kerenidis2020quantum} may not converge when applied to SDOPs
(this is not an issue for linear optimization, because all the
involved matrices are diagonal). Additional safeguards must be
employed to guarantee that the resulting solution will be symmetric;
numerous techniques have been proposed for this purpose. These methods
are generally grouped into families of search directions such as those
introduced by Kojima, Shindoh, and Hara \cite{kojima1997interior}, Monteiro and Zhang \cite{monteiro1997primal, monteiro1998unified, monteiro1998polynomial, zhang1998extending}, Monteiro and Tsuchiya
\cite{monteiro1999polynomial}, and Tseng \cite{tseng1998search}. We
refer the reader to an overview of these (as well as many other) search directions by Todd
\cite{mj1999study}.

The second major issue is that the analysis of the IPM in
\cite{ben2001lectures} assumes that the Newton linear systems are
solved exactly. This cannot be done in the quantum setting, because
nonnegligible errors are introduced in the tomography step, even if
the quantum linear systems algorithm can solve the problem with high
precision. Notably, the primal and dual search directions resulting from
tomography are not orthogonal, therefore the convergence proof of
\cite{ben2001lectures} no longer applies\footnote{Lemma 6.3 in
\cite{kerenidis2020quantum} requires $dS \bullet dY = 0$, but this property does not hold when the Newton linear system is solved inexactly; it is not clear how to amend
the subsequent analysis using the same framework, since even convergence of purely classical algorithms requires rigorous assumptions on the linear system errors.}. To ensure convergence when the search directions are not exact, we must rely on
IPM frameworks that can deal with the errors introduced from tomography.

Infeasible IPMs were originally motivated by the need of having an initial point (and all subsequent iterates) in the interior of the semidefinite cone. It is easy to have an initial interior point which is infeasible \cite{potra1998superlinearly, zhang1994convergence}, hence
the use of infeasible IPMs. Further, excessive CPU time and memory requirements may prevent solving large linear systems directly. In such a situation, one can employ an iterative method to determine the solutions to the linear systems to compute search directions, and such solutions are not guaranteed to be exact (e.g., due to early stopping of the iterative method). The resulting algorithm is
called inexact-infeasible IPM. The first polynomial time inexact-infeasible IPMs were proposed by Korzak \cite{korzak2000convergence} and Mizuno and Jarre \cite{mizuno1999global}. Inexact-infeasible IPMs are natural candidates for quantization. Using this terminology, the ``prototypical'' classical IPM described in \cite{ben2001lectures} can be called an exact-feasible IPM.

The little-studied inexact-feasible IPM (IF-IPM) bridges the gap between exact-feasible and inexact-infeasible IPMs. In an inexact-feasible IPM, it is typically assumed that primal and dual feasibility can be satisfied exactly, but we allow for a margin of error in solving a reduced Newton system. We argue that this framework is the most well-suited to serve as a foundation for QIPMs, compared to inexact-infeasible IPMs. Indeed, inexact-feasible IPMs exhibit the same iteration complexity as exact-feasible IPMs, and they impose milder restrictions on the inexactness of the Newton linear system; because of this, quantizing an inexact-feasible IPM leads to better theoretical performance than an inexact-infeasible IPM. Gondzio \cite{gondzio2013convergence} presented an Inexact-Feasible IPM for convex quadratic optimization problems; his complexity analysis is the basis for the analysis of our novel IF-QIPM, although our QIPM bears little resemblance to the classical IF-IPMs. Indeed, when designing an inexact-feasible QIPM we make several choices that are very atypical for classical IPMs, yielding an algorithm with no direct classical analogue, regardless of the fact that its analysis relies on known classical tools. One key distinction being that the IPM in \cite{gondzio2013convergence} \textit{assumes} that primal and dual feasibility can be satisfied exactly, whereas our IF-QIPM \textit{guarantees} primal-dual feasibility. 

\subsection{Contributions of this paper}
This paper highlights several challenges in the literature on QIPMs that have been previously overlooked. To address some of these challenges, it develops two provably convergent QIPMs for SDO. One algorithm is a quantization of the classical inexact-infeasible IPM from \cite{zhou2004polynomiality}; the other algorithm is an inexact-feasible QIPM loosely based on \cite{gondzio2013convergence}, where we introduce significant differences for the quantum setting. The inexact-infeasible QIPM converges to an $\epsilon$-approximate solution to an SDOP in $\Ocal \left(n^2 \log \left( \frac{1}{\epsilon} \right) \right)$ iterations, the inexact-feasible QIPM converges in $\Ocal \left(\sqrt{n} \log \left( \frac{1}{\epsilon} \right) \right)$ iterations.


Our first algorithm can be seen as a direct extension of \cite{kerenidis2020quantum} to properly address the two main issues discussed above, i.e., ensuring the existence of symmetric solutions and dealing with inexact search directions. This leads to the inexact-infeasible QIPM (II-QIPM). The running time of this algorithm, when using the HKM
direction \cite{helmberg1996interior, kojima1997interior, monteiro1998polynomial}, is
\begin{equation*}
 \Ocal \left(\left( n^{5.5} \kappa \kappa_{\cal A} \rho
  \left(  \|{\cal A}\|_F + \rho n^{1.5} \right)  \right) \frac{1}{\epsilon} \cdot \operatorname{polylog} \left(n, \kappa, \frac{1}{\epsilon} \right) \right),
\end{equation*}
where $\kappa_{\Acal}$ is the condition number of the matrix $\Acal$ whose columns are the vectorized constraint matrices $A^{(1)}, \dots, A^{(m)}$, and $\rho > 0$ is used to define the size of the initial infeasible solution and is generally
considered a constant in many papers (i.e., it is often ignored in the
running time analysis). Comparing the running time of the inexact-infeasible QIPM to classical IPM, there does not seem to be an
advantage in any parameter: classical exact-feasible IPMs (with running time
$\Ocal(n^{2\omega + 0.5} \log 1/\epsilon)$ on dense problems) are superior to
the above quantum algorithm. While the II-QIPM does not achieve any advantage compared to classical algorithms, we discuss it because it shows one possible way of dealing with the issues of existing QIPMs in the open literature, highlighting the fact that quantum advantage seems unlikely in this setting: to improve upon classical algorithm at least in some parameters, we need to modify the basic classical IPM scheme.

Our second algorithm, the inexact-feasible QIPM (IF-QIPM), solves the primal-dual SDO pair \eqref{e:SDO}-\eqref{e:SDD} with $m = \Ocal (n^2)$ constraints to $\epsilon$-optimality using at most
$$ \Ocal \left( n^{3.5} \frac{\kappa^2}{\epsilon} \cdot \operatorname{polylog} \left(n, \kappa, \frac{1}{\epsilon} \right) \right)  $$
QRAM accesses and $ \Ocal \left( n^{4.5}   \operatorname{polylog} \left(n, \kappa, \frac{1}{\epsilon} \right) \right) $ arithmetic operations. When applied to LO problems, the IF-QIPM requires
$$ \Ocal \left( n^{2} \frac{\kappa^2}{\epsilon} \cdot \operatorname{polylog} \left(n, \kappa, \frac{1}{\epsilon} \right) \right)  $$
QRAM accesses and $ \Ocal \left( n^{2.5} \cdot  \operatorname{polylog} \left(n, \kappa, \frac{1}{\epsilon} \right) \right) $ arithmetic operations.

For SDO problems involving $n \times n$ matrices and $m = \Ocal (n^2)$ constraints, the IF-QIPM achieves speedups in terms of $n$ compared to both classical feasible IPMs, and II-(Q)IPMs existing in the literature, though its dependence on $\kappa$ and $1/\epsilon$ prohibits an overall speedup, see, Section \ref{s:comparingrt}. We also de-quantize our framework to obtain a classical IF-IPM, finding that our classical algorithm exhibits an overall running time of $$\Ocal \left( n^{4.5} \kappa \cdot \operatorname{polylog} \left(\kappa, \frac{1}{\epsilon} \right) \right)$$ when applied to SDO problems $m = \Ocal (n^2)$. Like the IF-QIPM, our IF-IPM is faster in $n$ than all existing classical IPMs, and its dependence on $\epsilon$ is exponentially faster than its quantum counterpart. Yet, we cannot decisively conclude that the IF-IPM outperforms the current state of the art classical IPMs for SDO; its running time still depends (linearly) on $\kappa$. When applied to LO problems, our IF-IPM has a running time of $\Ocal \left( n^{2.5} \kappa \cdot \operatorname{polylog} \left(\kappa, \frac{1}{\epsilon} \right)\right)$. Note that this does not outperform the best classical IPM for LO even with respect to the dimension $n$: classical randomized IPMs can solve LOPs with $m = n$ constraints in time $\widetilde{\Ocal}_{n, \frac{1}{\epsilon}} \left( n^{\omega} \right)$, see \cite{cohen2021solving}.

It should be noted that the running time dependence on the condition number $\kappa$ is due to use of a QLSA for the solution of the Newton linear system. It is known that $\kappa$ can grow very large over the course of IPMs (regardless of problem class, e.g., LO or SDO), and in fact $\kappa \to \infty$ as we approach optimality, see the discussion in Section~\ref{s:comparingrt}. The QIPM has better dependence on
$1/\epsilon$ than the quantum MWU, but this comes at a cost of 
worse dependence on other parameters, most notably $n$ and
$\kappa$. It is worth noting, however, that QIPMs (as they are presented in this work) classically output the primal-dual optimal solution $(X^*, y^*, S^*)$, whereas QMWU algorithms only output $y^*$ and a way to construct $X^*, S^*$. In addition, the solutions we obtain using the IF-QIPM will satisfy primal and dual feasibility \textit{exactly}, while (Q)MWU algorithms can only obtain solutions which satisfy the constraints to some additive error, e.g., in the primal case we have
$$  A^{(i)} \bullet X  = b_i + \epsilon,~\forall i \in [m].$$
It should also be noted that (Q)IPMs achieve a normalized duality gap $\frac{X^* \bullet S^*}{n}$ at most $\epsilon$, whereas with (Q)MWUs the guarantee is $\epsilon$-additive. For a more detailed comparison of (Q)IPMs and (Q)MWUs, see Section \ref{ss:compare}. The QIPMs developed here enjoy a lower complexity than classical IPMs only if $\epsilon$ is relatively large, and $\kappa$ remains polynomially small over the course of the algorithm, a condition that cannot be guaranteed in theory. In practice, it is possible that these algorithms exhibit good performance if we are interested in obtaining a low-precision solution.

Summarizing, based on the current state of the art for classical and
quantum algorithms, quantum algorithms for linear systems do not seem
sufficient to obtain a speedup over classical IPMs, due to the
challenges highlighted above. This was
also observed in \cite{brandao2019faster} comparing classical and quantum algorithms for a specific application of SDO
in combinatorial optimization. It is known that the dependence of QLSAs on
$\kappa$ cannot in general be reduced to $\kappa^{1-\delta}$ for any $\delta > 0$
unless BQP = PSPACE \cite{harrow2009quantum}. We therefore
conjecture that a quantum speedup would have to rely on different
techniques: the direct quantization of a classical IPM, with linear
algebra performed on a quantum computer, does not seem to be
sufficient for an end-to-end speedup in general.

Our paper identifies some open questions. First, can we reduce the dependence on $1/\epsilon$ from linear to polylogarithmic? A promising direction to attain this objective is that of iterative refinement techniques, which have been successfully applied in the classical optimization literature for LOPs \cite{gleixner2020linear, gleixner2012improving, gleixner2016iterative}. We conjecture that they can be adapted to the setting of SDO, and this may help us improve the $\epsilon$-dependence of the QIPM. Second, is there a way to improve the dependence on $\kappa$? This seems difficult within the QLSA framework, due to the lower bounds of \cite{harrow2009quantum,somma2021complexity}, but in the context of QIPM, it would still be advantageous to use techniques that allow us to trade $\kappa$ for a slightly higher dependence on $n$ -- which, crucially, is \emph{exponential} in the number of qubits anyway.

The rest of this paper is organized as follows. In Section
\ref{s:pre}, we introduce our notation, some necessary concepts in
linear algebra, and basic quantum procedures. In Section \ref{s:cip} we
discuss the main components of classical IPMs. Section \ref{s:error} discusses how to deal with tomography errors , as well as algorithmic variants that allow one to capture inexactness and infeasibility. Section \ref{s:tech} provides technical results that are crucial to the convergence analysis of the IF-QIPM. In Section \ref{s:IFQIPM} we discuss its quantization, and provide a theoretical analysis of the
running time for the Inexact-Feasible and Inexact-Infeasible Quantum Interior Point Methods in Section \ref{s:complexity}. Section
\ref{s:con} concludes the paper. Some technical results, and results from other papers are presented in Appendices A and B, respectively. 

\section{Preliminaries} \label{s:pre} \label{s:qls}
We write $A^{\dagger}$ to indicate the conjugate transpose of $A$. We denote by $G \otimes K$ the tensor product of two matrices. For a matrix $G \in \R{n \times n}$, $\vec(G)$ is the $n^2 \times 1$ vector given by
$$ \vec(G) = \left(g_{11}, g_{21}, \dots, g_{n1}, g_{12}, \dots, g_{nn}  \right)^{\top}.$$ We make extensive use of the functions \textbf{svec} and \textbf{smat}, formally defined as follows. 

\begin{definition}
For $G \in \Scal^n$, $\textbf{\textup{svec}} (G) \in \R{\frac{1}{2} n (n+1)}$ is given by
    $$ \textup{\textbf{svec}}(G) =\left(g_{11}, \sqrt{2} g_{21}, \dots, \sqrt{2} g_{n1}, g_{22}, \sqrt{2} g_{32} \dots, \sqrt{2} g_{n2}, \dots, g_{nn} \right)^{\top} .$$
Further, the operator $\textup{\textbf{smat}}$ is the inverse operator of $\textbf{\textup{svec}}$. That is, 
$$ \textbf{\textup{smat}} \left[ \textbf{\textup{svec}}  (G) \right] = G. $$
\end{definition}
We also use the symmetric Kronecker product, defined below.
\begin{definition}\label{d:SymKronProd}
Let $G, K \in \Scal^n$. Define the $\frac{n(n+1)}{2} \times n^2$ matrix $V$ as:
$$ V_{(i,j), (k,l)} = \begin{cases}
1 &\text{if}~i = j = k = l, \\
1/\sqrt{2} &\text{if}~i = j \neq k = l,~\text{or}~i = l \neq j = k \\
0 &\text{otherwise}.
\end{cases}.$$
Then, the symmetric Kronecker product $G \otimes_s K$ of $G$ and $K$ is defined as:
$$G \otimes_s K = \frac{1}{2} V \left( G \otimes K + K \otimes G \right) V^{\top}.$$
\end{definition}

We let $\Scal_+^n$ and $\Scal_{++}^n$ denote the cones of symmetric positive semidefinite, and symmetric positive definite matrices, respectively. Given
two vectors $x, y$, we denote by $x \circ y$ the concatenation of the
two vectors. For $x \in \R{n}$, we denote its amplitude encoding by
$\ket{x}$, defined as
$$ \ket{x} = \frac{1}{\|x\|} \sum_{j \in [n]} x_j \ket{j}.$$ Notice
that $\ket{x}$ is a $(\log n)$-qubit state;  for simplicity, we assume that the sizes
of all spaces are powers of 2. All logarithms are base
2. The smallest and largest singular values of a matrix $A$ are denoted $\sigma_{\min}(A), \sigma_{\max}(A)$, and the smallest and largest eigenvalues are denoted $\lambda_{\min}(A), \lambda_{\max}(A)$. The condition number of $A$ is $\kappa_A = \frac{\lambda_{\max}(A)}{\lambda_{\min} (A)}$; we simply write $\kappa$ when referring the upper bound on the condition number of the Newton linear system. For an element $x \in \Cmbb^{d}$, $\Re(x) \in \R{d}$ is its real part.  We write ``II'' for inexact-infeasible and ``IF'' for inexact-feasible.

\subsection{Basic concepts on block-encoded matrices}
The QIPMs described in this paper make extensive use of quantum linear
algebra subroutines; notably, matrix multiplication
and matrix-vector products. We perform all these operations in the
framework of block-encodings, as discussed in \cite{chakraborty2018power,gilyen2019quantum}. The block-encoding framework generalizes earlier work on high-accuracy QLSAs in the sparse-access input model \cite{childs2017quantum} by expanding not only the possible input models, but also the set of operations allowed on the input matrix. Before giving formal definitions and an
overview of key results from the literature, we give an informal
introduction of the work in \cite{chakraborty2018power} to convey intuition.

To get a basic understanding of how to perform matrix-vector and matrix-matrix multiplication using this framework, suppose we want to multiply a matrix $A$ by some vector $u$. We assume that $u$ is encoded in the quantum state $\ket{u}$, which we want to update as follows:
$$ \ket{u} \to \frac{A\ket{u}}{\| A \ket{u}  \|}.$$
A block-encoding $U$ of $A$ is a unitary such that:
\begin{align*}
U = \begin{pmatrix}
\frac{A}{\alpha}& \cdot \\
\cdot & \cdot
\end{pmatrix},
\end{align*}
where $\alpha$ is a normalization factor that is chosen to ensure that $U$ has operator norm at most 1, i.e., $\| U \| \leq 1$. In particular, if $A$ is a
$w$-qubit operator, an $(\alpha, a, \xi)$-block-encoding of $A$
uses $a$ extra qubits and implements $A$ up to total error
$\xi$. Hence, we have
\begin{align*}
U \ket{u}\ket{0}^{\otimes a} \approx \begin{pmatrix}
\frac{A}{\alpha}& \cdot \\
\cdot & \cdot
\end{pmatrix}
\begin{pmatrix}
\ket{u} \\
\ket{0}
\end{pmatrix}
= \frac{1}{\alpha} A \ket{u} \ket{0}^{\otimes a} + \ket{\cdot},
\end{align*}
and performing approximately $\frac{\alpha}{\| A \ket{u} \|}$ rounds of
amplitude amplification yields $ \frac{A \ket{u} }{\| A \ket{u}\|}$
with high probability, see e.g., \cite{chakraborty2018power}.

For matrix multiplication, let $U_A$ be an $(\alpha_1, a_1,
\xi_1)$-block-encoding of an $w$-qubit operator $A$, and $U_B$ be a
$(\alpha_2, a_2, \xi_2)$-block-encoding of a $w$-qubit operator $B$. Then we
have
$$(I_{a_2} \otimes U_A)(I_{a_1} \otimes U_B) = \begin{pmatrix}
\frac{B}{\alpha_2}& \cdot \\
\cdot & \cdot
\end{pmatrix}
\begin{pmatrix}
\frac{A}{\alpha_1}& \cdot \\
\cdot & \cdot
\end{pmatrix} = \begin{pmatrix}
\frac{AB}{\alpha_1 \alpha_2}& \cdot \\
\cdot & \cdot
\end{pmatrix},$$
which yields an $(\alpha_1 \alpha_2, a_1 + a_2, \alpha_1 \xi_2 + \alpha_2
\xi_1)$-block-encoding of $AB$.

For matrix inversion, let $U$ be a $(\alpha, a, \xi)$-block
encoding of $A$ that can be implemented in time $T_U$. Then,
Chakraborty et al.~\cite{chakraborty2018power} prove that we can
implement a block-encoding $V$ of the inverse of $A$:
$$ V = \begin{pmatrix}
\frac{A^{-1}}{2 \kappa}& \cdot \\
\cdot & \cdot
\end{pmatrix}$$
Hence, given $A$ as a block-encoding and $\ket{u}$, we can solve the linear system 
$$Ax = u$$
by outputting the state
$$ \frac{A^{-1} \ket{u}}{\|A^{-1} \ket{u} \|},$$
computed by applying amplitude amplification to the state: 
$$ V \ket{u} \ket{0}^{\otimes a} = \frac{1}{2 \kappa} A^{-1} \ket{u} \ket{0}^{\otimes a}  + \ket{\cdot}.$$

To construct block-encodings of the matrices used by our QIPMs, we use the following
construction, see \cite{chakraborty2018power} for details. Suppose we
want to construct a matrix $A$ and we have a way of constructing the
quantum states $\ket{\psi_i} = \sum_{j} \frac{A_{ij}}{\| A_i \|}
\ket{i,j}$, where $A_{i}$ is the $i$-th row of $A$, and $\ket{\phi_j}
= \sum_{i} \frac{\|A_{i} \|}{\| A \|_F} \ket{i,j}$. Then we construct:
$$ U_R^{\dagger} U_L = \begin{pmatrix} 
- & \bra{\psi_1} & - & \cdot\\
&  \vdots &  & \\
- & \bra{\psi_n} & - & \cdot \\
 & \cdot & & \cdot
\end{pmatrix} 
\begin{pmatrix} 
| & & |  &\\
\ket{\phi_1} & \cdots & \ket{\phi_n} & \cdot \\
| & & | & \\
\cdot &  & \cdot & \cdot
\end{pmatrix} = 
\begin{pmatrix} 
[ \bra{\psi_i}\ket{\phi_j} ]_{i,j \in [n]} & \cdot \\
\cdot & \cdot 
\end{pmatrix} = \begin{pmatrix} 
\frac{A}{\|A\|_F} & \cdot \\
\cdot & \cdot 
\end{pmatrix},$$ where we used the fact that
\begin{align*}
\bra{\psi_i}\ket{\phi_j} &=  \frac{A_{ij}}{\|A\|_F}.
\end{align*}
Thus, $U_R^{\dag} U_L$ is a block-encoding of $A$ with normalization
factor $\alpha = \|A\|_F$. For the above construction to work, we must have a procedure to implement $U_R, U_L$. Note that these unitaries are easy to obtain starting from controlled operations to prepare the states $\ket{\psi_i}, \ket{\phi_j}$, i.e., operations of the form
$\ket{i}\ket{0} \to \ket{i}\ket{\psi_i}$, $\ket{j}\ket{0} \to
\ket{j}\ket{\phi_j}$, and these controlled operations can in turn be constructed with a procedure similar to the state preparation
algorithm of Grover and Rudolph \cite{grover2002creating} for
efficiently integrable distributions. This can be made more efficient if we allow pre-processing to create certain data structures that can be stored in quantum-accessible storage, i.e., QRAM. A quantum RAM (QRAM) is a form of storage that allows for querying a superposition of addresses. Given a QRAM that stores the classical vector $v_j \in
\R{2^q}$, and a quantum state $\sum_{j=0}^{2^q-1} \beta_j \ket{j}$,
the QRAM assumption is that the following mapping can be performed in
time $\widetilde{\Ocal}_n (1)$, i.e., polylogarithmic in the size of the vector:
\begin{equation*}
  \sum_{j=0}^{2^q-1} \beta_j \ket{j} \otimes \ket{0} \to \sum_{j=0}^{2^q-1} (\beta_j \ket{j} \otimes \ket{v_j}).
\end{equation*}
Throughout this paper, we assume that we have access to a classical-write, quantum-read QRAM that is large enough to store all input matrices, as well as a number of vectors so that we can efficiently perform indexed SWAP operations in our tomography subroutine (see Section \ref{s:tomography}), i.e., $ \widetilde{\Ocal}_{n, \frac{1}{\delta}, \frac{1}{\epsilon}} \left( m n^2 + 2^{\frac{n^2}{\epsilon}} \right)$. The exponential term in the QRAM size comes from the tomography subroutine if we insist on a classical-write QRAM, but as we discuss in \ref{s:tomography}, there are several alternative models to reduce this cost: for example, if the QRAM can be written in superposition, or if we have access to indexed SWAP gates, then the necessary QRAM size is only polynomial. For more
details about the QRAM data structure to prepare the amplitude
encoding of a vector, or the matrices $U_R, U_L$ discussed above, we
refer the reader to \cite{chakraborty2018power, kerenidis2020quantum}. 

\subsection{Useful results on block-encoded matrices}
We now provide formal definitions for the concepts informally
discussed in the previous section, as well as other results that are
used in the remainder of this paper. The material in this section is
mostly taken from \cite{chakraborty2018power, gilyen2019quantum},
which provide improvements over the framework discussed in
\cite{kerenidis2016quantum,kerenidis2020quantum}.

\begin{definition}[block-encoding]
  \label{prop:qramblockenc}
Let $A \in \C{2^w \times 2^w}$ be a $w$-qubit operator. Then, a $(w + a)$-qubit unitary $U$ is an $(\alpha, a, \xi)$-block-encoding of $A$ if $ U = \begin{pmatrix}
\widetilde{A} & \cdot \\
\cdot & \cdot
\end{pmatrix}$,
such that $$\| \alpha \widetilde{A} - A \| \leq \xi$$ An $(\alpha, a, \xi)$-block-encoding of $A$ is said to be efficient if it can be implemented in time $T_U = \Ocal (\textup{poly} (w)).$
\end{definition}
Note that Definition \ref{prop:qramblockenc} is equivalent to the following property:
$$ \| A - \alpha ( \bra{0}^{\otimes a} \otimes I_{2^w}) U ( \ket{0}^{\otimes a} \otimes I_{2^w}) \| \leq \xi.$$

The next proposition formalizes an idea discussed in the previous
section. We do not provide details on the necessary data structure,
referring the reader to \cite{chakraborty2018power} for an extensive
discussion on this topic.
\begin{proposition}[Lemma 50 in \cite{gilyen2019quantum}]
  \label{prop:qramblockenc}
  Let $A \in \C{m \times m}$ with $m = 2^w$ and $\xi > 0$.
  \begin{enumerate}
  \item[(i)] Fix $q \in [0, 2]$ and define $\mu_q (A) = \sqrt{n_q (A) n_{(2-q)} (A^{\top})}$ where $n_q (A) = \max_i \| a_i \|_q^q$ is the $q$-th power of the maximum $q$-norm of the rows of $A$. If $A^{q}$ and $(A^{2-q})^{\dag}$ are
    both stored in QRAM data structures, then there exist unitaries $U_R$ and $
    U_L$ that can be implemented in time $\Ocal(\textup{poly}(w \log
    \frac{1}{\xi}))$ and such that $U^{\dag}_RU_L$ is a $(\mu_q(A),
    w + 2, \xi)$-block-encoding of $A$.
  \item[(ii)] If $A$ is stored in a QRAM data structure, then there exist
    unitaries $U_R$ and $U_L$ that can be implemented in time
    $\Ocal(\textup{poly}(w \log \frac{1}{\xi}))$ and such that
    $U^{\dag}_RU_L$ is an $(\|A\|_F, w + 2, \xi)$-block-encoding of
    $A$.
  \end{enumerate}
\end{proposition}

\begin{definition}[State preparation pair]
Let $y \in \mathbb{C}^m$ and $\| y \|_1 \leq \beta$. The pair of unitaries $(P_L, P_R)$ is called a $(\beta, p, \xi)$-state-preparation-pair if $P_L \ket{0}^{\otimes p} = \sum_{j = 0}^{{2^p} - 1} c_j \ket{j}$ and $P_R \ket{0}^{\otimes p} = \sum_{j = 1}^{{2^p} - 1} d_j \ket{j}$ such that $\sum_{j=0}^{m-1} | \beta (c_j^* d_j) - y_j | \leq \xi$ and for all $j \in m, \dots, 2^p -1$ we have $c_j^* d_j = 0$.
\end{definition}

\begin{proposition}[Lemma 52 in \cite{gilyen2019quantum}]
  \label{prop:lincombblock}
  (Linear combination of block-encoded matrices, with weights given by
  a state preparation pair) Let $A = \sum_{j=0}^{m-1} y_j A^{(j)}$ be
  a $w$-qubit operator, where $A^{(j)}$ are matrices. Suppose $P_L, P_R$
  is a $(\beta, p, \xi_1)$-state-preparation pair for $y$, $W =
  \sum_{j=0}^{m-1} \ket{j}\bra{j}\otimes U_j + ((I - \sum_{j=0}^{m-1}
  \ket{j}\bra{j}) \otimes I_a \otimes I_w)$ is an $(w + a + p)$-qubit
  unitary with the property that $U_j$ is an $(\alpha, a,
  \xi_2)$-block-encoding of $A^{(j)}$. Then we can implement a
  $(\alpha\beta, a+p, \alpha \xi_1 + \alpha \beta
  \xi_2)$-block-encoding of $A$ with a single use of $W, P_R$ and
  $P_L^{\dag}$.  
\end{proposition}

The following two propositions are critical to the efficiency of the
QIPMs. They state that one can construct a block-encoding as a
product of two block-encoded matrices, with overhead that is merely
polylogarithmic in the size of the matrices. The difference between
the two propositions is that in the second one, the normalization
factor of the block-encoding of the product is fixed, rather than
dependent on the input.
\begin{proposition}[Lemma 4 \cite{chakraborty2018power}] 
  \label{prop:product}
  (Product of block-encoded matrices)
If $U_A$ is an $(\alpha_1, a_1, \xi_1)$-block-encoding of an $s$-qubit operator $A$, and $U_B$ is an $(\alpha_2, a_2, \xi_2)$-block-encoding of an $s$-qubit operator $B$, then $(I_{a_2} \otimes U_A)(I_{a_1} \otimes U_B)$ is an $(\alpha_1 \alpha_2, a_1+a_2, \alpha_1 \xi_2 + \alpha_2 \xi_1)$-block-encoding of $AB$. 
\end{proposition}

\begin{proposition}[Lemma 5 in \cite{chakraborty2018power}]
  \label{prop:ampproduct}
  (Product of preamplified block-encoded matrices) Suppose $A$ and $B$ have sizes such that $AB$ is a a valid matrix product, and $\|A \|, \|B \| \leq 1$. Given an $(\alpha_1,
  a_1, \xi_1)$-block-encoding $U_A$ of an $s$-qubit operator $A$, and a
  $(\alpha_2, a_2, \xi_2)$-block-encoding $U_B$ of an $w$-qubit operator
  $B$, with $\alpha_1, \alpha_2 \ge 1$, we can implement a $(2, a_1 + a_2
  + 2, \sqrt{2}(\xi_1 + \xi_2 + \xi_3))$-block-encoding of $AB$ in
  time $\Ocal((\alpha_1 (T_{U_A} + a_1) + \alpha_2(T_{U_B} + a_2))\log \frac{1}{\xi_3})$,
  where $T_{U_A}$ and $T_{U_B}$ are the implementation times for $U_A$ and $U_B$,
  respectively.
\end{proposition}

If one has a block-encoded matrix, one can also implement a block-encoding of (positive or negative) powers of that matrix, see \cite[Lemma 9, 10]{chakraborty2018power}. We are now in a position to state how block-encodings can be used to solve quantum linear systems problems, which is formalized in the following result from \cite{chakraborty2018power}.
\begin{theorem}[Theorem 30 in \cite{chakraborty2018power}]
  \label{thm:qlsa}
  (Solution of linear system) Let $r \in (0, \infty)$, $\kappa \geq 2$
  and $H$ a Hermitian matrix such that its nonzero eigenvalues lie in
  $[-1, -1/\kappa]\cup[1/\kappa, 1]$. Suppose that
  $$ \xi = o \left( \frac{ \delta}{ \kappa^{2} \log^3
    \frac{\kappa^{2}}{\delta}} \right)$$ and $U$ is an $(\alpha, a,
  \xi)$-block-encoding of $H$, that can be implemented in time $T_U$. Suppose further that we can prepare a state
  $\ket{v}$ that is in the image of $H$ in time $T_v$.
  Then, for any $\delta$, we can output a state that is $\delta$-close to $H^{-1} \ket{v} / \|H^{-1} v\|$ in time
$$ \Ocal \left( \kappa\left(\alpha  (a + T_U) \log^2 \left(\frac{\kappa}{\delta}\right) + T_v\right) \log \kappa\right).$$
\end{theorem}

\begin{proposition}[Corollary 32 in \cite{chakraborty2018power}]
  \label{prop:blocknormest}
  (Norm estimation) Let $p \in (0, \infty)$, $\kappa \geq 2$ and $H$ a
  Hermitian matrix such that its nonzero eigenvalues lie in $[-1,
    -1/\kappa]\cup[1/\kappa, 1]$. Suppose that
  $$ \xi = o \left( \frac{ \delta}{ \kappa^{2} \log^3
    \frac{\kappa^{2}}{\delta}} \right)$$ and $U$ is an $(\alpha, a,
  \xi)$-block-encoding of $H$, that can be implemented in time $T_U$. Suppose further that we can prepare a state
  $\ket{v}$ that is in the image of $H$ in time $T_v$. Then we can output $\tilde{e}$ such that
  \begin{equation*}
    (1-\delta)\|H^{-1}\ket{v}\| \le \tilde{e} \le (1+\delta)\|H^{-1}\ket{v}\|
  \end{equation*}
  in time
$$ \Ocal \left( \frac{\kappa}{\delta} \left(\alpha(a + T_U) \log^2 \left(\frac{\kappa}{\delta}\right) + T_v \right) \log^3 \kappa \log \frac{\log\kappa}{\xi} \right).$$
\end{proposition}

\begin{proposition}
  \label{prop:blockencrule}
  Let $$A = \begin{pmatrix} M_{11} & \dots & M_{1c} \\ \vdots & \ddots
    & \vdots \\ M_{r1} & \dots & M_{rc} \end{pmatrix} \in \R{n \times
    n},$$ where for simplicity $r$, $c$ are powers of 2, and each $M_{ij}$ is a matrix, of appropriate dimension,
  that is stored in a QRAM data structure. Suppose further the row
  norms and column norms of each matrix $M_{ij}$ are known. Then we can
  construct a $(rc\|A\|_F, \Ocal(\log n), \xi)$-block-encoding of $A$ in
  time $\Ocal(\textup{poly}(rc, \log n, \log \frac{1}{\xi}))$.
\end{proposition}
\begin{proof}
Let $A^{(1)}, \dots, A^{(m)}$ denote the $m = rc$ matrices in $\R{n \times n}$ that each store one of the $M_{ij}$ matrices in the same position that it appears in $A$, but with all other entries being 0. Then, $A$ can be written as a linear combination of the $m$ many matrices $A^{(k)}$ as $A = \sum_{k=1}^{m} y_k A^{(k)}$,
where $y_k = 1$ for all $k$. We normalize all matrices using the Frobenius norm of $A$, to ensure that all spectral norms are $\le 1$.

Notice that if a matrix $M_{ij}$ is stored in QRAM, we can construct the block-encoding of the corresponding $A^{(k)}$ using Proposition~\ref{prop:qramblockenc}, with straightforward modifications of the underlying algorithm, because the elements and norms of $A^{(k)}$ can be efficiently computed from those of $M_{ij}$. Thus, a $(\|A\|_F, \log n + 2, \xi')$-block-encoding of $A^{(k)}$ can be efficiently constructed, for any $\xi' > 0$. 

Let $P_L, P_R$ be a $(m, \log m, 0)$-state-preparation pair for $y$, which can be constructed by simply taking $P_L = P_R = H^{\otimes \log m}$ where $H$ is the Hadamard gate. We then use Proposition \ref{prop:lincombblock}, where $W$ can be obtained by adding a control qubit to the circuits used to construct the block-encoding of each $A^{(k)}$, and the parameter $\xi_2$ is chosen to be $\xi/(m\|A\|_F)$. This yields a $(m \| A \|_F , \log n + \log m,  \xi)$-block-encoding of $A$ with a single use of $W, P_R$ and $P_L^{\dag}$; since $P_R$ and $P_L^{\dag}$ only require $\Ocal(\log m)$ gates each, and $W$ takes $\Ocal(\textup{poly}(m, \log n, \log \frac{1}{\xi}))$ time, the proof is complete.
\end{proof}

When symmetrizing the Newton linear system, some components of the coefficient matrix can be conveniently defined using Kronecker, and symmetric Kronecker products of the input matrices. We conclude this section by formalizing how one can block-encode $G \otimes_s K$ up to error $\xi$ in time $\widetilde{\Ocal}_{\frac{n}{\xi}}(1)$, provided that the matrices $G$ and $K$ are stored in a QRAM data structure.


 
To do so in an efficient manner, we construct a block-encoding using the sparse-access input model; in particular we use this input model to construct block-encodings of the matrix $V$ and its transpose $V^{\top}$ which are used to define the symmetric tensor product of two matrices as given in Definition \ref{d:SymKronProd}. In this setting, it is assumed that the matrix $\Mcal \in \C{M \times N}$ we wish to block-encode has $s_r$-sparse rows and $s_c$-sparse columns. Further, we can query the elements of $\Mcal$ using an oracle:
$$ O_{\Mcal} : \ket{i} \ket{j} \ket{0}^{\otimes p} \mapsto \ket{i} \ket{j} \ket{m_{ij}},\quad \forall i \in [M], j \in [N].$$
Likewise, it is assumed that we have access to oracles $O_r$ and $O_c$ that are capable of querying the indices of non-zero elements of each row and column. This input model is the most suitable to construct $V$, since we can easily implement an efficient algorithmic description of its sparsity pattern and element values. The following result from \cite{gilyen2019quantum} establishes how one can efficiently implement block-encodings of sparse-access matrices.  

\begin{lemma}[Lemma 48 in \cite{gilyen2019quantum}] \label{lem:sparseBE}
Let $\Mcal \in \C{2^w \times 2^w}$ be a matrix that is $s_r$-row-sparse and $s_c$-column-sparse, and each element of $\Mcal$ has value at most 1. Suppose that we have access to the following sparse-access oracles acting on two $(w+1)$ qubit registers: 
\begin{align*}
     O_{r} : \ket{i} \ket{k}  &\mapsto \ket{i} \ket{r_{ij}}\quad \forall i \in [2^w] -1, k \in [s_r],~\text{and} \\
     O_{c} : \ket{l} \ket{j}  &\mapsto \ket{c_{lj}} \ket{j}\quad \forall l \in [s_c], j \in [2^w] - 1,~\text{where}
\end{align*}
$r_{ij}$ is the index for the $j$-th non-zero entry of the $i$-th row of $\Mcal$, or if there are less than $i$ non-zero entries, then it is $j + 2^w$, and similarly $c_{ij}$ is the index for the $i$-th non-zero entry of the $j$-th column of $\Mcal$, or if there are less than $j$ non-zero entries, then it is $i + 2^w$. Additionally, assume that we have access to an oracle $O_{\Mcal}$ that returns the entries of $\Mcal$ in a binary description: 
$$ O_{\Mcal} : \ket{i} \ket{j} \ket{0}^{\otimes p} \mapsto \ket{i} \ket{j} \ket{m_{ij}},\quad \forall i,j \in [2^w]-1,$$
where $m_{ij}$ is a $p$-bit binary description of the $ij$-matrix element of $\Mcal$. Then, we can implement a $(\sqrt{s_r s_c}, w+3, \xi)$-block-encoding of $\Mcal$ with a single use of $O_r$, $O_c$ and two uses of $O_{\Mcal}$, and additionally using \\$\Ocal~\left(w + \log^{2.5} \left( \frac{s_r s_c}{\xi} \right) \right)$ one and two qubit gates while using $\Ocal \left(p, \log^{2.5} \left( \frac{s_r s_c}{\xi} \right) \right)$ ancilla qubits. 
\end{lemma}

The above result allows us to establish how to block-encode the Symmetric Kronecker product of two matrices.
\begin{proposition}
  \label{prop:blockencSymKronecker}(Block-encoding symmetric Kronecker products)
  Let $G, K \in \R{n \times n}$, be stored in a QRAM data structure. Then we can
  construct a $(\|G \otimes K\|^2_F, \Ocal(\log n), \xi)$-block-encoding of $G \otimes_s K$ in
  time $$\Ocal\left(\textup{poly}\left(\log n, \log \frac{1}{\xi}\right)\right).$$
\end{proposition}

\begin{proof} Recall that 
$$G \otimes_s K = \frac{1}{2} V \left( G \otimes K + K \otimes G \right) V^{\top},$$
where $V$ is the $\frac{n(n+1)}{2} \times n^2$ matrix defined in Definition~\ref{d:SymKronProd}. Observe that by definition, $V$ is 1-column-sparse and 2-row-sparse so $s_c s_r = 1 \cdot 2 = 2$.   

Using the description of $V$, we can construct the sparse-access oracles $O_r$ and $O_c$ as defined in Lemma \ref{lem:sparseBE} (which act on two $(\log n^2+1)$ qubit registers). Additionally, from the definition of $V$ we can construct an oracle $O_{V}$ that returns the entries of $V$ in a binary description: 
$$ O_{V} : \ket{i} \ket{j} \ket{0}^{\otimes p} \mapsto \ket{i} \ket{j} \ket{v_{ij}},\quad \forall i,j \in [2^{\log n^2}]-1,$$
where $v_{ij}$ is a $p$-bit binary description of the $ij$-matrix element of $V$. By Lemma \ref{lem:sparseBE} we can efficiently implement a $(\sqrt{2}, \Ocal(\log n), \xi_V)$-block-encoding of $V$ with a single use of $O_r$, $O_c$ and two uses of $O_{V}$.
We can construct a $(\sqrt{2}, \Ocal(\log n), \xi_{V})$-block-encoding of $V^{\top}$ in the same manner.

Given that $G$ and $K$ are stored in a QRAM data structure, we can construct a $(\|G \otimes K\|_F, \Ocal(\log n), \xi_{G \otimes K})$-block-encoding of $G \otimes K$ in time $\Ocal(\textup{poly}(\log n, \log \frac{1}{\xi_{G \otimes K}}))$, and similarly for $K \otimes G$, for some error parameters to be specified later. Indeed, note that a $(\|G \otimes K\|_F, \Ocal(\log n), \xi_{G \otimes K})$-block-encoding of $G \otimes K$ is trivial to construct: by definition of block-encoding, it suffices to take the tensor product of block-encodings for $G$ and $K$, constructed using Proposition~\ref{prop:qramblockenc}, keeping the ancilla bits separate and ensuring that the sum of errors is less than $\xi_{G \otimes K}$.
 
From here we use a linear combination of our block-encodings of $G \otimes K$ and $K \otimes G$ to block-encode $G \otimes K + K \otimes G$. Applying Proposition \ref{prop:lincombblock} with 
\begin{align*}
    \xi_{G \otimes K} = \frac{\xi_1}{2 \| G \otimes K \|_F} \qquad \text{and}  \qquad
    \xi_{K \otimes G} = \frac{\xi_1}{2 \| G \otimes K \|_F},
\end{align*}
yields a $(\|G \otimes K\|_F^2, \Ocal( \log n), \xi_1)$-block-encoding of $G \otimes K + K \otimes G$, and we can implement this block-encoding in time $\Ocal(\textup{poly}(\log n, \log \frac{1}{\xi_1}))$, i.e., $\widetilde{\Ocal}_{\frac{n}{\xi_1}} (1)$. 
  
Using our block-encodings of $V$ and $G \otimes K + K \otimes G$, we apply Proposition \ref{prop:product} with
\begin{align*}
    \xi_{V} = \frac{\xi_2}{2 \sqrt{2}} \qquad \text{and} \qquad \xi_{1} = \frac{\xi_2}{4  \| G \otimes K \|_F^2},
\end{align*}
to implement a $( \sqrt{2} \|G \otimes K\|_F^2, \Ocal(\log n), \xi_2)$-block-encoding of $V \left( G \otimes K + K \otimes G \right)$ in time $$\Ocal\left(\textup{poly}\left(\log n, \log \frac{1}{ \xi_2} \right)\right).$$ A final application of Proposition \ref{prop:product} with 
\begin{align*}
    \xi_{V} = \frac{\xi}{2 \sqrt{2}} \qquad \text{and} \qquad
    \xi_{2} = \frac{\xi}{2 \sqrt{2} \| G \otimes K \|_F^2},
\end{align*}
yields a $( 2 \|G \otimes K\|_F^2, \Ocal( \log n ), \xi)$ block-encoding of $$\frac{1}{2} V \left( G \otimes K + K \otimes G \right) V^{\top} = G \otimes_s K,$$ in time  $\Ocal(\textup{poly}(\log n, \log \frac{1}{\xi}))$, which completes the proof.  
\end{proof}

\subsection{Extracting the solution of a linear system: tomography algorithm}
\label{s:tomography}
In each iteration of our QIPMs, we solve a linear system of
equations (known as the Newton linear system) to obtain search directions to progress to the next iterate. Due to the fact that we solve this system using a QLSA, the solution we obtain is encoded in a quantum state, and a classical description is necessary to construct the linear
system that arises in the subsequent iteration. Namely, in order to update the current solutions to the primal and dual SDOPs, $X$ and $(y,S)$, we require a procedure to map the state proportional to the solution of the Newton linear system to a classical solution pair. There are many such procedures (e.g.,
\cite{keyl2006quantum, o2016efficient}); in this paper we can take advantage of the fact that we have access to a unitary preparing the quantum state, as well as its inverse. 

To obtain an $\ell_2$-norm estimate of the quantum state, we rely on the following result.
\begin{theorem}[\cite{van2021tomography}]
\label{thm:euclidean_norm_tomo}
Let $\ket{\psi} = \sum_{j=0}^{d-1} y_j \ket{j}$ be a quantum state, $y \in \mathbb{C}^d$ the vector with elements $y_j$, and $U\ket{0} = \ket{\psi}$. There is a quantum algorithm that, with probability at least $1-\delta$, outputs $\tilde{y} \in \mathbb{R}^{d}$ such that $\| \Re(y) - \tilde{y}\|_2 \le \varepsilon$ using $\Ocal(\frac{d}{\varepsilon} \log \frac{d}{\delta})$ applications of $U$, $\widetilde{\Ocal}_{d, \frac{1}{\delta}, \frac{1}{\varepsilon}}(\frac{d}{\varepsilon})$ indexed-SWAP gates acting on $d$ bits, and $\widetilde{\Ocal}_{d, \frac{1}{\delta}, \frac{1}{\epsilon}}(\frac{d}{\varepsilon})$ additional gates. If we have access to a classical-write, quantum-read QRAM of size $\widetilde{\Ocal}_{d, \frac{1}{\delta}, \frac{1}{\varepsilon}}(2^{d/\varepsilon})$, we do not need the indexed-SWAP gates.
\end{theorem}
\begin{proof}
This follows from \cite[Prop.~22]{van2021tomography}, setting the $\ell_\infty$-norm error to $\varepsilon/\sqrt{d}$. For the gate complexity, \cite{van2021tomography} uses a ``QRAM-like'' indexed-SWAP gate, which is used to efficiently implement the mapping from a binary description of a vector $x$ to its amplitude encoding $\ket{x}$. It is straightforward to note that this mapping can also be stored in a classical-write, quantum-read QRAM: as the vectors $x \in [-1,1]^d$ necessary for the construction are known in advance, and each coordinate requires precision $\Ocal(1/\varepsilon)$, we can employ the usual QRAM data structure for each of these $\Ocal(2^{d/\varepsilon})$ vectors in QRAM. This operation needs to be performed only once and it only depends on the dimension and the precision. Alternatively, the operation can be performed with a $\widetilde{\Ocal}_{d, \frac{1}{\delta}, \frac{1}{\varepsilon}}(\frac{d^2}{\varepsilon})$ quantum-writable QRAM (the data structure can be computed in superposition and written onto the QRAM).
\end{proof}

It is shown in \cite{van2021tomography} that the above sample complexity is optimal, i.e., the number of applications of $U$ cannot be reduced in general. For purposes of the QIPMs presented in this paper, we treat the algorithm from \cite{van2021tomography} as a black box oracle that performs our tomography steps. For simplicity, and since we are already assuming access to a classical-write, quantum-read QRAM, we use that input model, although as noted, other models are possible and potentially more efficient.

In light of the normalization of quantum states, in quantum state tomography it is assumed that the vector to be extracted (for our purposes, the solution to a linear system) has unit norm. The tomography error bound is
relative, and hence, in order to ensure that we satisfy the absolute error bound for the unscaled vector, we need to divide $\xi_k$ by the
maximum norm of a solution vector, which we denote by $\varrho$. Setting $\varepsilon = \xi/ \varrho$ and $T_U = T_{LS}$ (i.e, the time to prepare and solve our linear system using a quantum computer), Theorem \ref{thm:euclidean_norm_tomo} asserts that we can extract a classical estimate of a solution to the linear system with Euclidean norm error $\xi$ in time
$$T_{TO} (T_{LS}, \xi) = \widetilde{\Ocal}_{d, \varrho, \frac{1}{\xi}} \left( \frac{d}{\xi/ \varrho} T_{LS}\right) = \widetilde{\Ocal}_{d, \varrho, \frac{1}{\xi}}  \left(\frac{d}{\xi} \varrho  T_{LS}  \right).$$
Note that previous work on QIPMs used the tomography subroutine in \cite{kerenidis2020quantum}, which has quadratic dependence on the inverse precision, rather than linear; this improves our final running time, but it is not the only reason why we are able to give a faster QIPM than previously known. Indeed, as we describe in Section~\ref{s:IFQIPM}, the most impactful source of improvement is our development of an inexact-feasible QIPM, that overcomes the issues with feasibility present in other approaches.

\section{Primal and Dual SDOs and the Central Path}\label{s:cip}
Recall that the symmetric matrices $A^{(1)} , \dots, A^{(m)}$ are
linearly independent by assumption, and $S \in \Scal^n$ is the slack
matrix of the dual problem, as defined in \eqref{e:SDD}. We denote the feasible sets of \eqref{e:SDO} and \eqref{e:SDD} by
\begin{align*}
    \Pcal &= \left\{X \in \Scal^n :  A^{(i)} \bullet X = b_i, ~\forall i \in  [m], X \succeq 0 \right\}, \\
    \Dcal &=  \left\{(y, S) \in \R{m} \times \Scal^n : \sum_{i=1}^m y_i A^{(i)} + S = C, S \succeq 0 \right\}.
\end{align*}
For a feasible IPM we must assume that a strictly feasible pair
$X$ and $(y, S)$ with $X \succ 0$ and $S \succ 0$ exists, i.e., the
interior point condition (IPC) is satisfied \cite{de1997initialization}. It
is known that with the self-dual embedding model this condition may be
assumed without loss of generality \cite{de1997initialization}. We define the sets of \textit{interior feasible solutions} by
\begin{align*}
    \Pcal^0 &= \left\{ X \in \Pcal : X \succ 0 \right\},\\
    \Dcal^0 &= \left\{ (y, S) \in \Dcal : S \succ 0 \right\}.
\end{align*}
With the IPC satisfied, it is guaranteed that the
primal and dual optimal sets:
\begin{align*}
    \Pcal^* &= \left\{X \in \Pcal : C \bullet X = z_{P} \right\}, \\
    \Dcal^* &=  \left\{(y, S) \in \Dcal :  b^{\top} y = z_{D} \right\}, 
\end{align*}
are nonempty and bounded, with $z_P = z_D$, i.e., the duality gap is zero. In particular, this holds for all optimal solutions
$X^*$ and $(y^*, S^*)$ satisfying
\begin{equation}
    C \bullet X^* - b^{\top} y^*  = X^* \bullet S^*  = 0,
\end{equation}
which implies $X^* S^* = S^* X^* = 0$ as $X^*$ and $S^*$ are symmetric
positive semidefinite matrices. However, for primal-dual IPMs, the
complementarity condition $XS = 0$ is perturbed to:
\begin{equation}\label{e:CP}
   XS = \sigma \nu I,
\end{equation}
where $I$ is the $n \times n$ identity matrix, $\sigma \in [0,1]$ is the centering parameter, and $\nu > 0$ is the
so-called central path parameter, which is progressively reduced to
zero in the course of the optimization algorithm.

For all $\nu > 0$, assuming the IPC and linear
independence of the matrices $A^{(i)}$ for $i \in [m]$, the central path equation
system
\begin{align}\label{e:CP2}
    A^{(i)} \bullet X &= b_i ~\forall i \in [m], ~X \succ 0 \nonumber \\
    \sum_{i \in [m]} y_i A^{(i)} + S &= C,~ S \succ 0 \\
    XS &=  \sigma \nu I, \nonumber 
\end{align}
has a unique solution \cite{nesterov1988general}. The set of solutions
for all $\nu > 0$ gives the central path for the primal and dual
SDOPs. In IPMs we aim to follow the central path as $\nu \to 0$,
i.e., as we approach optimality. In fact, we simply seek to
stay in a certain neighborhood of the central path.

A classical IPM, as outlined in Algorithm \ref{alg:IPM}, begins with a strictly feasible primal-dual pair $(X_0 , S_0) \in \Pcal^0 \times \Dcal^0$ for the primal and dual SDO problems \eqref{e:SDO} and \eqref{e:SDD}. This pair has a duality gap of $X_0 \bullet S_0 = \nu_0 n$, and a distance to the central path of $d(X_0,S_0) \leq \gamma \nu_0$ for $\gamma \in (0,1)$, where $\nu_0 = \frac{X_0 \bullet S_0}{n}$, for some appropriate distance metric $d(X,S)$. Without loss of generality, one can assume that $\nu_0 = \Ocal (1)$. As we mentioned earlier, it is always possible obtain an interior feasible solution using the self-dual embedding model, and in particular, the analytic center of our space, i.e., $X_0 = S_0 = n^{-1} I$ is always strictly feasible in this setting. 

In our work on the IF-QIPM, we consider the narrow (or, Frobenius) neighborhood 
$$ \Ncal_F (\gamma) = \left\{ (X, y, S) \in \Pcal^0 \times \Dcal^0 : \left\| X^{1/2} S X^{1/2} - \nu I \right\|_F = \left[ \sum_{i=1}^n (\lambda_i (XS) - \nu)^2 \right]^{1/2} \leq \gamma \nu  \right\}.$$
Another popular alternative is the so called \textit{negative infinity neighborhood} that is a \textit{large neighborhood}, defined as
$$ \Ncal_{\infty}^- (1-\gamma) = \{ (X,y,S) \in \Pcal^0 \times \Dcal^0 : \lambda_{\min} (XS) \geq \gamma \nu \},$$
where $\gamma \in (0,1)$. It is well known that classical feasible primal-dual path following IPMs that use the former neighborhood exhibit an iteration complexity of the order $\Ocal(\sqrt{n} \log(1/\epsilon))$ whereas long-step IPMs, using the latter neighborhood, have iteration complexity $\Ocal(n \log(1/\epsilon))$. However, in practice the long step algorithms tend to converge faster than the short step counterpart; in this paper we focus on the theoretical running time and therefore use the narrow (or Frobenius) neighborhood $\Ncal_F (\gamma)$. Thus, we define our centrality measure at a point $(X, S) \in \Scal_{+}^n \times \Scal_{+}^n$, with $\nu = \frac{X \bullet S}{n}$, to be
\begin{equation}\label{e:7}
    d(X,S) = \left\| X^{1/2} S X^{1/2} - \nu I \right\|_F.
\end{equation}

Then, in each iteration, we solve the so-called Newton linear system in order to obtain $\Delta X$ and $\Delta S$ and update the solutions using the following rule:
\begin{align*}
    X &= X + \Delta X, \\
    S &= S + \Delta S.
\end{align*}
If we were to directly linearize the complementarity condition, the Newton linear system would then be written as:
\begin{equation}\label{e:KP}
\begin{aligned}
\Delta X S + \Delta S X &= \sigma \nu I - XS, \\
\Delta X &\in \Null (A^{(1)}, A^{(2)}, \dots, A^{(m)}), \\
\Delta S &\in \Rcal (A^{(1)}, A^{(2)}, \dots, A^{(m)}),
\end{aligned}
\end{equation}
where $\Null (\cdot), \Rcal (\cdot)$ are the nullspace (or, \textit{kernel}) and rowspace of the constraint matrices, respectively (we provide explicit definitions for these subspaces in the following sections). 
However, it is known that \eqref{e:KP} does not have a symmetric
matrix solution with respect to $\Delta X$, see, e.g.,
\cite{alizadeh1998primal}. This must be addressed, because both the
primal and dual solutions must be symmetric.
 
\begin{algorithm}[H]
\SetAlgoLined
\KwIn{$\epsilon, \delta > 0$;  $\sigma = 1 - \delta/\sqrt{n}$; $\gamma \in (0,1)$ \\
Choose $(X_0, y_0, S_0) \in \Ncal_F (\gamma)$\\
Choose a nonsingular matrix $P$ from Monteiro-Zhang family of search directions , and compute $P_0$ \\
Set $\nu_0 \gets \frac{X_0 \bullet S_0}{n}$, $k \gets 0$ \\}
    \While{$\nu > \epsilon$}{ 
    \begin{enumerate}
        \item  $ \nu_k \leftarrow \frac{X_k \bullet S_k}{n}$
        \item Compute the scaling matrix $P_k$
        \item Compute the solution $(\Delta X_k, \Delta y_k, \Delta S_k)$ to the Newton linear system \eqref{e:svecTodd2} defined by choice of $P_k$
        \begin{align*}
           X_{k+1} &\leftarrow X_k +  \Delta X_k,~ S_{k+1} \leftarrow S_k + \Delta S_k ~\text{and}~y_{k+1} \leftarrow y_k +  \Delta y_k \\
           k &\leftarrow k+1
        \end{align*}
    \end{enumerate}
  } 
 \caption{Classical interior point method}
\label{alg:IPM}
\end{algorithm}

\subsection{Symmetrizing the Newton System}
As has been extensively studied in the literature of IPMs (see, e.g.,
\cite{alizadeh1998primal, nesterov1997self, nesterov1998primal}), the
Newton steps $\Delta S$ and $\Delta X$ have to be symmetric
matrices. More specifically, even though $\Delta S$ is guaranteed to
be symmetric by the definition of a dual feasible solution, there is
no symmetric $\Delta X$ that solves system \eqref{e:KP}. Thus, we need
a different approach to guarantee symmetry.

In an effort to generalize the scaling methods required for
primal-dual symmetry, following Zhang \cite{zhang1998extending}, we
define the linear transformation $H_P (M)$ that symmetrizes a
matrix $M$ using a given invertible matrix $P$:
\begin{equation}
    H_P(M) = \frac{1}{2 } \left[ P M P^{-1} + \itran{P} \tran{M} \tran{P} \right].
\end{equation}
Observe that $H_P(M) = \sigma \nu I \iff M = \sigma \nu I$; we can then write the
central path equations in symmetric form as: $H_P (XS) = \sigma \nu I$. In
this light, we can also symmetrize the linearized complementarity condition of the Newton system $X \Delta S +
S\Delta X = \sigma \nu I - XS$ as
\begin{equation}\label{e:tranP}
    H_P (X \Delta S + S\Delta X) = \sigma \nu I - H_P (XS).
\end{equation}
Equation \eqref{e:tranP} can be explicitly expressed as: 
$$
    P (X \Delta S +  \Delta X S) P^{-1} 
    + \itran{P} (\Delta S X + S\Delta  X) \tran{P} 
    = 2 \sigma \nu I - P XS P^{-1} - \itran{P} SX \tran{P}.
$$

Defining $\Acal^{\top} = [\textbf{\textup{vec}}(A^{(1)}) \textbf{\textup{vec}}(A^{(2)}) \cdots \textbf{\textup{vec}}(A^{(m)})]$, we can write the symmetrized Newton linear system as
\begin{align*}
         \begin{pmatrix}
         0 & \tran{\cal A} & I_{n^2} \\
         {\cal A} & 0 & 0 \\
         \Ecal & 0 & \Fcal
         \end{pmatrix}
         \begin{pmatrix}
         \vec (\Delta X) \\
          \Delta y \\
         \vec (\Delta S)
         \end{pmatrix}
         &= \begin{pmatrix}
        0 \\
        0 \\
         \vec (R^c)
         \end{pmatrix},
\end{align*}
where $I_{n^2}$ is the identity matrix of order $n^2$ and 
\begin{align*}
         \Ecal &= P \otimes S P^{-1} + P^{-1} S \otimes P,\\
         \Fcal &= P X \otimes P^{-1} + P^{-1} \otimes X P,\\
         R^c&= \sigma \nu I - H_P (XS).
\end{align*} 

Now let us define $\Acal_s^{\top} = [\textbf{\textup{svec}}(A^{(1)}) \textbf{\textup{svec}}(A^{(2)}) \cdots \textbf{\textup{svec}}(A^{(m)})].$ Following Todd et al.~\cite{todd1998nesterov}, we can also write the Newton linear system using the $\svec$ notation, in which case we have:
\begin{equation}\label{e:svecTodd2}
    \begin{pmatrix}
         0 & {\cal A}_s & \Ical \\
        \tran {\cal A}_s & 0 & 0 \\
         \Ecal_s & 0 & \Fcal_s
         \end{pmatrix}
         \begin{pmatrix}
         \svec (\Delta X) \\
          \Delta y \\
         \svec (\Delta S)
         \end{pmatrix}
    = 
            \begin{pmatrix}
    0 \\ 0\\ R^c_s
    \end{pmatrix},
\end{equation}
where $\Ical$ is the identity matrix of order $\frac{n(n+1)}{2}$ and 
\begin{align*}
    \Ecal_s &= P \otimes_s P^{- \top} S,\\
    \Fcal_s &= P X \otimes_s P^{-\top}, \\
    R^c_s   &= \svec(\sigma \nu I - H_P (XS)) = \svec(R^c).
\end{align*}
One can note that matrices $\Ecal_s$ and $\Fcal_s$ are nonsingular whenever $X$ and $S$ are positive definite matrices \cite{todd1998nesterov}. We have:
\begin{subequations}
\begin{align}
    \Ecal_s = P \otimes_s P^{-\top} S &= (I \otimes_s P^{-\top} S P^{-1}) (P \otimes_s P), \label{e:EkInv} \\
    \Fcal_s = PX \otimes_s P^{-\top} &= (PXP^{\top} \otimes_s I) (P^{-\top} \otimes_s P^{-\top}), \label{e:FkInv}
\end{align}
\end{subequations}
where the equalities follow from the definition of the symmetric Kronecker product detailed earlier. 

We then need to choose $P$ to obtain specific instances of this system, and we refer to the class of search directions parameterized by $P$ as the Monteiro-Zhang family of search directions \cite{monteiro1998unified}. We present results involving three possible choice of $P$: $P = I$, the AHO direction \cite{alizadeh1998primal}; $P = S^{1/2}$, the HKM direction \cite{helmberg1996interior, kojima1997interior, monteiro1998polynomial}; and $P = W^{-1/2}$, where $W$ is the Nesterov-Todd \cite{nesterov1998primal} scaling matrix, defined as:
\begin{equation}\label{e:NTdir}
 W = S^{-1/2} (S^{1/2} X S^{1/2})^{1/2} S^{-1/2} 
= X^{1/2} (X^{1/2} S X^{1/2})^{-1/2} X^{1/2}. 
\end{equation}

To define the HKM Newton linear system for the $k$-th iteration, we use symmetry to drop all of the transpose terms involving $P$ and set:
\begin{subequations}\label{e:HKM}
\begin{align}
    \Ecal_{s_k} &=S_k^{1/2} \otimes_s S_k^{-1/2} S_k =S_k^{1/2} \otimes_s S_k^{1/2}, \label{e:HKM_E}\\
    \Fcal_{s_k} &= S_k^{1/2} X_k \otimes_s S_k^{-1/2}, \label{e:HKM_F}\\
    R^c_{s_k}   &= \svec(2 \sigma \nu_k I -  S_k^{1/2} X_k S_k^{1/2}).
\end{align}
\end{subequations}
In defining the AHO Newton Linear system, taking $P = I$ simplifies the above quantities to to
\begin{align*}
    \Ecal_{s_k} &= I \otimes_s S_k,\\
    \Fcal_{s_k} &= X_k \otimes_s I, \\
    R^c_{s_k}   &= \svec( 2 \sigma \nu I - X_k S_k - S_k X_k).
\end{align*}
Although this guarantees that our solutions to the Newton system are symmetric, the AHO
direction only guarantees a solution if we are within a small (local)
neighborhood of the central path: the $\infty$-norm neighborhood with opening less than 1/3. Finally, for
the Nesterov-Todd Newton system, we set $P=W^{-1/2}$ and we can write:
\begin{subequations}\label{e:NT2}
          \begin{align}
    \Ecal_{s_k} &= W_k^{-1/2} \otimes_s W_k^{1/2} S_k, \\
    \Fcal_{s_k} &= W_k^{-1/2} X_k \otimes_s W_k^{1/2},\\
         R_{s_k}^c &= \svec \left(2 \sigma \nu I -  W_k^{-1/2} X_k S_k W_k^{1/2}- W_k^{1/2} S_k X_k W_k^{-1/2} \right).
         \end{align}
\end{subequations}

Symmetry is not the only issue overlooked by previous works on QIPMs; we still need to account for the consequences of inexact tomography in the context of QIPMs. 

\section{Newton Linear Systems for QIPMs}\label{s:error}
We now turn our attention to solving the Newton linear system using quantum linear system solvers. We present two different Newton linear systems that account for the inexactness in the search directions resulting from the use of tomography. The first system, based on the infeasible central path and its neighborhood, allows us to quantize an Inexact-Infeasible IPM. However, in subsequent sections we will show that an algorithm based on this direct methodology leads to an algorithm with a poor overall running time. Thus, we subsequently introduce a new methodology that allows us to design an inexact, but feasible QIPM framework which yields a speedup in $n$ over the classical variants. 

\subsection{An Infeasible Central Path and its Neighborhood}
Our first effort to deal with the tomography errors is to quantize an Inexact-Infeasible IPM framework. That is, we work with a solution to a Newton linear system in which the right hand side includes residuals to capture \textit{(i)} infeasibility of the iterate sequence
and \textit{(ii)} inexactness of the search direction. In this section we describe the infeasible central path, and the resulting Newton linear system. 
 
To illustrate the issue at hand precisely, previous works on QIPMs \cite{casares2020quantum, kerenidis2020quantum, kerenidis2019quantum} (if we apply symmetrization to their Newton systems) assume we can exactly solve the Newton linear system \eqref{e:svecTodd2}. Yet, solving the system \eqref{e:svecTodd2} with a quantum computer introduces errors due to our use of QLSA and tomography. As a result, a classical estimate $(\Delta X, \Delta y, \Delta S)$ of the quantum state $\ket{\Delta X \circ \Delta y \circ \Delta S}$ will indeed only satisfy 
\begin{equation*} 
         \begin{pmatrix}
         0 & \tran{\cal A} & I_{n^2} \\
         {\cal A} & 0 & 0 \\
         \Ecal & 0 & \Fcal
         \end{pmatrix}
         \begin{pmatrix}
         \vec (\Delta X) \\
          \Delta y \\
         \vec (\Delta S)
         \end{pmatrix}
    = 
            \begin{pmatrix}
    \xi_d \\ \xi_p\\ R^c + \xi_c
    \end{pmatrix},
\end{equation*}
where $(\xi_d, \xi_p, \xi_c)$ are the errors to which we satisfy primal feasibility, dual feasibility and the complementarity condition, respectively. 

 The feasible IPM framework, upon which the QIPMs in \cite{kerenidis2020quantum, kerenidis2019quantum} are built, therefore cannot be applied (even if these works had properly accounted for symmetry). The analysis of feasible primal-dual IPMs relies on the fact that $\Delta X \bullet \Delta S = 0$, this property is used throughout the classical analyses. However, the presence of the errors $\xi_d$ and $\xi_p$ means the primal and dual search directions $\Delta X$ and $\Delta S$ are not guaranteed to be in orthogonal subspaces. While this alone implies that the QIPMs presently in the literature are not convergent, one also has to account for the inexactness to which we solve the complementarity condition. 

Let $(X_0, y_0, S_0$) be an initial point, not necessarily
feasible, where for $\rho >0$ we have
\begin{equation} \label{e5}
X_0 = S_0 = \rho I.
\end{equation}
Following \cite{zhou2004polynomiality}, it is assumed that $\rho$ is specified such that for a primal-dual optimal solution $(X^*, y^*, S^*)$, and for some $0 < \gamma_1 < 1$, we have
\begin{subequations}
\begin{align}
(1 - \gamma_1) X_0 &\succeq X^*,~~(1 - \gamma_1) S_0 \succeq S^* \label{e6}\\
\rho &\geq \frac{1}{n} \left( \trace{(X^*)} + \trace{(S^*)} \right). \label{e7}
\end{align}
\end{subequations}

Letting $b^{\top} = [b_1, b_2, \dots, b_m]$, consider the quantities: 
$$
\mu_0 = \frac{X_0 \bullet S_0}{n} = \rho^2, \quad 
R_0^p = \Acal \vec (X_0) - b,  \quad 
R_0^d = \tran{\Acal} y_0 + \vec (S_0) - \vec (C). 
$$
Then, for $\tau \in (0, 1]$, the \emph{infeasible} central path equations are given by:
  \begin{equation}
    \label{e:infcp}
    \begin{pmatrix}
      \tran{\cal A} y + \vec(S) - \vec(C) \\
      {\cal A} \vec(X) - b \\
      XS
    \end{pmatrix}
    = \begin{pmatrix}
      \tau R_0^d \\
      \tau R_0^p \\
      \tau \mu_0 I
    \end{pmatrix}, ~X, S \succ 0.
  \end{equation}
These equations define a set of infeasible centers, that approaches
the solution of \eqref{e:SDO}-\eqref{e:SDD} as $\tau \to 0$. Let
$\gamma_2 \in (0,1)$; we can define a neighborhood of the
infeasible central path, which we denote $\Ncal_I$, as:
\begin{equation}
  \label{e:neighborhood}
\Ncal_I = \left\{ \begin{split}&(\tau, \theta, X, y, S) \in \R{}_+ \times \R{}_+  \times \Scal_+^n \times \R{m} \times \Scal_+^n : \tau \leq \theta \\
&\Acal^{\top} y +  \vec(S) - \vec(C) = \tau( R_0^d + \zeta^d), \|\zeta^d \| \leq \gamma_1 \rho \\
&\Acal \vec(X) - b  = \tau( R_0^p + \zeta^p), \| \Acal^+ \zeta^p \| \leq \gamma_1 \rho \\
&\| H_p (XS) - \theta \mu_0 I \| \leq \gamma_2 \theta \mu_0
\end{split} \right\}
\end{equation}
where $\Acal^+ = \tran{\Acal} (\Acal \tran{\Acal})^{-1}$ is the right pseudoinverse of $\Acal$.

The following results from Zhou and Toh \cite{zhou2004polynomiality}
are crucial for the complexity analysis of the II-QIPM later in the paper; they
correspond to when the HKM search direction is chosen, which is the
one for which we give a full analysis of the method.\footnote{For the other
search directions, we will limit ourselves to discussing the running
time, without giving all the details.}

\begin{lemma}[Lemma 1 in \cite{zhou2004polynomiality}] \label{lemma1}
For any $r^p$ and $r^d$ satisfying $\|\Acal^+ r^p\| \leq \gamma_1 \rho$ and $\|  r^d\| \leq \gamma_1 \rho$, there exists $(\tilde{X}, \tilde{y}, \tilde{S})$ that satisfies the following conditions:
\begin{subequations}
\begin{align}
\Acal^{\top} \tilde{y} + \vec( \tilde{S}) - \vec(C) &= R_0^d + r^d \\
 \tran{\Acal} \vec( \tilde{X}) - b   &= R_0^p + r^p\\
 (1-\gamma_1) \rho I &\preceq \tilde{X} \preceq (1+ \gamma_1) \rho I \\
  (1-\gamma_1) \rho I &\preceq \tilde{S} \preceq (1+ \gamma_1) \rho I.
\end{align}
\end{subequations}
\end{lemma}
\begin{lemma}[Lemma 2 in \cite{zhou2004polynomiality}] \label{lemma2}
 If the conditions \eqref{e5}, \eqref{e6} and \eqref{e7} hold, then for any $(\tau, \theta, X, y, S) \in \Ncal_I$ with $\theta \in (0, 1]$, we have 
$$ \tau \trace{(X)} = \Ocal(\theta \rho n)~\textup{and}~\tau \trace{(S)} = \Ocal(\theta \rho n).$$
\end{lemma}
Implications of the above lemma are summarized below.
\begin{remark}[\cite{zhou2004polynomiality}] \label{r:Rem2}
Let $\Bcal = \{ (X, y, S) : (\tau, \theta, X, y, S) \in \Ncal_I, \theta
\in (0,1], \tau = \theta\}$.  Then, by Lemma \ref{lemma2}, the set
  $\Bcal$ is bounded, because for $(X, y, S) \in \Bcal$ we have
$\| X \| \leq \trace{(X)} \leq \Ocal(\rho n)~\textup{and}~\| S \| \leq \trace{(S)} \leq \Ocal(\rho n)$.
\end{remark}
From the above, one can derive Frobenius norm bounds on the solutions $X$ and $S$ by noting the relationship $\| G \|_F \leq \sqrt{n} \| G \|$ for $G \in \R{n \times n}$, i.e., 
\begin{equation}\label{e:II_norm_bound}
    \| X \|_F, \|S \|_F = \Ocal(\rho n^{1.5}).
\end{equation}
We do however note that the bounds on the operator norms and trace of the solutions $$\| X \| \leq \trace{(X)} \leq \Ocal(\rho n)~\textup{and}~\| S \| \leq \trace{(S)} \leq \Ocal(\rho n)$$ given in the above remarks are \emph{only relevant} to II-IPMs. They are properties associated with the infeasible central path, and our discussion for the IF-QIPM will not use these bounds. 

\begin{remark}[\cite{zhou2004polynomiality}]
  Suppose we generate a sequence $\{(\tau_k, \theta_k, X_k, y_k,
  S_k)\}$ in the neighborhood $\Ncal_I$ with the properties that
  $\theta_k \geq \tau_k \; \forall k$ and $1 = \theta_0 \geq \theta_k
  \geq \theta_{k+1} \geq 0$. If $\theta_k \to 0$ as $k \to \infty$,
  then any limit point of the sequence $\{(X_k, y_k, S_k)\}$ is a
  solution of
    \begin{equation*} 
    \begin{pmatrix}
      \tran{\cal A} y + \vec(S) - \vec(C) \\
      {\cal A} \vec(X) - b \\
      XS
    \end{pmatrix}
    = 0, \quad X, S \succeq 0.
  \end{equation*}
  If $\tau_k = \theta_k$, then the sequence
  $\{(X_k, S_k)\}$ is also bounded.
\end{remark}

\subsubsection{Inexact Search Directions}
Let $\eta_1 \in (0, 1]$, $\eta_2 \in (0,1)$ with $\eta_1 \geq \eta_2$
and $\tau_0, \theta_0 = 1.$ Given a point $(\tau_k, \theta_k, X_k,
y_k, S_k) \in \Ncal_I$, we try to generate a new point
$$(\tau_{k+1}, \theta_{k+1}, X_{k+1}, y_{k+1}, S_{k+1}) \in \Ncal_I$$
by computing an update along the search direction $(\Delta X_k,
\Delta y_k, \Delta S_k)$ for the solution $X, y, S$. This naturally leads to an iterative algorithm. We choose the search direction to be a solution of the following system of equations:
\begin{align}\label{e:inNewt}
         \begin{pmatrix}
         0 & \tran{\cal A} & I_{n^2} \\
         {\cal A} & 0 & 0 \\
         \Ecal_k & 0 & \Fcal_k
         \end{pmatrix}
         \begin{pmatrix}
         \vec (\Delta X_k) \\
          \Delta y_k \\
         \vec (\Delta S_k)
         \end{pmatrix}
         &= \begin{pmatrix}
         - \eta_1 R_k^d \\
         - \eta_1 R_k^p \\
         \vec (R^c_k)
         \end{pmatrix},
\end{align}
where 
$$ R_{k}^c = \vec \left((1 - \eta_2) \tau_k \mu_0 I -  H_P (X_k S_k) \right).$$

As previously discussed, in the quantum setting both the QLSA and
tomography necessarily introduce errors in the solution of
\eqref{e:inNewt}. We therefore consider an inexact search direction
$(\Delta X_k, \Delta y_k, \Delta S_k)$, defined next, following
\cite{zhou2004polynomiality}. Let $\{ \vartheta_i \}_{i =0}^{\infty}$
be a sequence of reals in $(0, 1]$, such that
$$ \bar{\vartheta} = \sum_{i = 0}^{\infty} \vartheta_i < \infty.$$
We say that $(\Delta X_k, \Delta y_k, \Delta S_k)$ is an \emph{inexact search direction} at the $k$-th iteration if it solves the linear system:
\begin{align}\label{e:inexactNewtII}
         \begin{pmatrix}
         0 & \tran{\cal A} & I_{n^2} \\
         {\cal A} & 0 & 0 \\
         \Ecal_k & 0 & \Fcal_k
         \end{pmatrix}
         \begin{pmatrix}
         \vec (\Delta X_k) \\
          \Delta y_k \\
         \vec (\Delta S_k)
         \end{pmatrix}
         &= \begin{pmatrix}
         - \eta_1 (R_k^d + r^d_k)\\
         - \eta_1 (R_k^p + r^p_k) \\
         \vec (R_k^c) + \vec(r_k^c)
         \end{pmatrix}, 
\end{align}
where $r_k^c \in \Scal^n$ and the \emph{residual terms} satisfy:
\begin{equation}\label{e:resbound}
\| r_k^p \| \leq \hat{\vartheta} \gamma_1 \rho \tau_k \vartheta_k,~~\| r_k^d \| \leq \gamma_1 \rho \tau_k \vartheta_k,~~\|r_k^c\| \leq \frac{1}{2}(1 - \eta_2) \gamma_2 \theta_k \mu_0,
\end{equation}
with $0< \hat{\vartheta} \leq \sigma_{\min}(\Acal)$. 
Upon obtaining a classical estimate of the solution to \eqref{e:inexactNewtII} using QLSA and tomography, it is possible that the complementarity error $r^c$ is not symmetric (as we have no safeguards to ensure that the tomography errors are symmetric). Therefore, in order to ensure that we always have $r^c \in \Scal^n$, we take one further precaution and utilize the $\svec (\cdot)$ characterization of the Newton system \eqref{e:inexactNewtII}. Let
\begin{align*} 
R_{s_k}^p &= \Acal_s \svec (X_k) - b,  \\
R_{s_k}^d &= \tran{\Acal}_s y_k + \svec (S_k) - \svec (C). 
\end{align*}
Then, in each iteration of the II-QIPM we use the HKM direction with $P = S^{1/2}$, and solve 
\begin{align}\label{e:inexactNewt}
         \begin{pmatrix}
         0 & \tran{\cal A}_s & \Ical \\
         {\cal A}_s & 0 & 0 \\
         \Ecal_{s_k} & 0 & \Fcal_{s_k}
         \end{pmatrix}
         \begin{pmatrix}
         \svec (\Delta X_k) \\
          \Delta y_k \\
         \svec (\Delta S_k)
         \end{pmatrix}
         &= \begin{pmatrix}
         - \eta_1 (R_{s_k}^d + r^d_k)\\
         - \eta_1 (R_{s_k}^p + r^p_k) \\
          R_{s_k}^c + \svec(r_k^c)
         \end{pmatrix}, 
\end{align}
where $\Ecal_{s_k}$ and $\Fcal_{s_k}$ are defined according to \eqref{e:HKM_E}-\eqref{e:HKM_F}, and 
$$ R_{s_k}^c = \svec \left(2 (1 - \eta_2) \tau_k \mu_0 I -  S_k^{1/2} X_k S_k^{1/2} \right).$$
 
We can now give a full description of the classical Inexact-Infeasible IPM, as given in
\cite{zhou2004polynomiality}, in Algorithm \ref{alg:InfeaseIPM}. This algorithm converges in polynomial time provided that \eqref{e:resbound} and \eqref{e:inexactNewt} hold; this will be discussed in Section \ref{s:IFQIPM}.

\begin{algorithm}[H]
\SetAlgoLined 
\KwIn{$\tau_0, \theta_0 = 1$; $\epsilon > 0$; $\eta_1 \in (0,1]$; $\gamma_1, \gamma_2, \eta_2 \in (0,1)$ \\
\noindent Pick a sequence $\{ \vartheta_i \}_{i =0}^{\infty}$ in $(0, 1]$, such that $ \bar{\vartheta} = \sum_{i = 0}^{\infty} \vartheta_i < \infty$ \\
\noindent Choose a nonsingular matrix $P$ from Monteiro-Zhang family of search directions \\
\noindent Find an initial point $(\tau_0, \theta_0, X_0, y_0, S_0) \in \Ncal_I$\\
\noindent Store $(\tau_0, \theta_0, X_0, y_0, S_0)$ in QRAM\\}
    \While{$\tau > \epsilon$}{
      \begin{enumerate}
      \item Solve the system \eqref{e:inexactNewt} to obtain inexact search direction $(\Delta X_k, \Delta y_k, \Delta S_k)$ \label{QIPM-solve}
      \item Set \footnotesize $$\alpha_k = \max\left\{ \alpha ~:~\alpha \in \left[0 , \min \left(1, \frac{1}{\eta_1(1+ \bar{\vartheta})} \right) \right],~(\tau_k(\alpha), \theta_k(\alpha), X_k(\alpha), y_k(\alpha), S_k(\alpha)) \in \Ncal_I \right\}$$
        \normalsize
        \begin{align*}
          X_{k+1} &\leftarrow X_k + \alpha_k \Delta X_k,~ S_{k+1} \leftarrow S_k + \alpha_k \Delta S_k ~\text{and}~y_{k+1} \leftarrow y_k +  \alpha_k\Delta y_k \\
           \tau_{k+1} &\leftarrow (1 - \alpha \eta_1) \tau_k,~\text{and}~\theta_{k+1} \leftarrow (1 - \alpha \eta_1) \theta_k \\
           k &\leftarrow k+1
        \end{align*}
    \end{enumerate}
  } 
 \caption{Inexact-Infeasible Interior Point Method.}
 \label{alg:InfeaseIPM}
\end{algorithm}

\subsection{The Nullspace Representation of the Newton Linear System}
Analogous to Gondzio's \cite{gondzio2013convergence} IF-IPM for linearly-constrained quadratic optimization (LCQO), in our novel IF-QIPM approach the first two equations of \eqref{e:svecTodd2} will hold, while we assume that the third equation is only satisfied up to some residual error $R^r \in \Scal^{n}$, which in our setting is due to the use of quantum state tomography (subsequent discussion will show that we can guarantee the residual is symmetric):
\begin{equation}\label{e:svecTodd}
    \begin{pmatrix}
    0 & \tran{\cal A}_s & \Ical \\
         {\cal A}_s & 0 & 0 \\
         \Ecal_{s} & 0 & \Fcal_{s}
    \end{pmatrix}
        \begin{pmatrix}
     \svec(\Delta X) \\ \Delta y \\ \svec(\Delta S)
    \end{pmatrix}
    = 
            \begin{pmatrix}
    0 \\ 0\\ R^c_s + \svec(R^r)
    \end{pmatrix},
\end{equation}
where $R^c_s = \svec(\sigma \nu I - H_P (XS)) = \svec(R^c)$. As is standard in the inexact Newton method literature  \cite{dembo1982inexact, gondzio2013convergence, kelley1995iterative}, we assume that the residual term $R^r$, for some $\beta \in (0,1)$, satisfies 
\begin{equation}\label{e:RcAssumption1}
    \| R^r \|_F \leq \beta \| R^c  \|_F. \tag{AR1}
\end{equation}
In this way, we ensure that the residual norms are being driven towards zero as we approach optimality, due to the fact that $\| R^c  \|_F \to 0$ as $\nu \to 0$. Additionally, choosing the level of error in this manner is not very restrictive, as the algorithm terminates upon reaching a solution satisfying $\nu \leq \epsilon$.

The ensuing result from \cite{todd1998nesterov} establishes the uniqueness of the solution to the system \eqref{e:svecTodd}. 
\begin{theorem}[Theorem 3.1 in \cite{todd1998nesterov}] \label{t:uniqueSol}
Suppose $X$ and $S$ are positive definite. Then the system of equations \eqref{e:svecTodd} has a unique solution $(\Delta X, \Delta y, \Delta S) \in   \Scal^n \times \R{m} \times \Scal^n$ if $\Ecal^{-1}_s \Fcal_s$ is positive definite (not necessarily symmetric). In particular, this holds when $X$, $S$ and $H_P (XS)$ are positive definite.
\end{theorem}

\begin{proof}
The proof is the same as the one given in \cite{todd1998nesterov}, but we repeat it here for completeness. In what follows, positive definiteness does not imply symmetry. We seek to show that the system 
\begin{equation}\label{e:solTrivial}
    \begin{pmatrix}
    0 & \Acal_s & 0 \\
     \Acal^{\top}_s & 0 & \Ical \\
    0 & \Ecal_s & \Fcal_s 
    \end{pmatrix}
        \begin{pmatrix}
    \Delta y \\ \svec(\Delta X) \\ \svec(\Delta S)
    \end{pmatrix}
    = 0
\end{equation}
is satisfied only by the all-zero solution. Note that for convenience, we have reordered the block-rows of the Newton system in order to utilize block-Gaussian elimination for the purposes of this proof. From \eqref{e:EkInv} it follows that $\Ecal_s$ is invertible. Hence, applying block-Gaussian elimination to \eqref{e:solTrivial} yields the Schur complement equation:
\begin{equation*} 
    (\Acal_s \Ecal_s^{-1} \Fcal_s \Acal_{s}^{\top} ) \Delta y = 0.
\end{equation*}
By assumption, $\Ecal_s^{-1} \Fcal_s \succ 0$ and $\Acal_s$ has full row rank, implying that the matrix $\Acal_s \Ecal_s^{-1} \Fcal_s \Acal_{s}^{\top}$ is positive definite, therefore $\Delta y = 0$. Thus the second equation in \eqref{e:solTrivial} gives $\Delta S = - \smat (\Acal_s^{\top} \Delta y) = 0$ and from the third equation it follows that $\Delta X = - \Ecal_s^{-1} \Fcal_s \Delta S = 0$. This completes the proof. 
\end{proof}

If the Newton linear system as given by \eqref{e:svecTodd} is solved using QLSA and tomography, the first two block equations in this system will not be satisfied exactly, either. To design our IF-QIPM, we re-write the optimality conditions \eqref{e:svecTodd2} as:
\begin{align}\label{e37b}
   \Delta X &\in \Null (\Acal_s) \nonumber\\
    \Delta S &\in \Rcal (\Acal_s) \\
    (P \otimes_s P^{- \top} S) \svec(\Delta X) + (P X \otimes_s P^{-\top}) \svec (\Delta S)   &= \svec(\sigma \nu I - H_P (XS)) \nonumber
\end{align}
where $\Null (\Acal_s)$ denotes the nullspace of $\Acal_s$ and $\Rcal (\Acal_s)$ denotes the rowspace of $\Acal_s$. Crucially, $\Null (\Acal_s)$ and $\Rcal (\Acal_s)$ are orthogonal subspaces to one another; we always have $\Delta X \bullet \Delta S = 0$ for any $\Delta X \in \Null (\Acal_s)$ and $\Delta S \in \Rcal (\Acal_s)$.

For $\Rcal (\Acal_s)$, we choose  $$\{\textbf{\textup{svec}}(A^{(1)}), \textbf{\textup{svec}}(A^{(2)}), \dots, \textbf{\textup{svec}}(A^{(m)})\}$$ as a basis. If the primal-dual pair \eqref{e:SDO}-\eqref{e:SDD} is given in canonical form (inequalities instead of equalities), we can trivially obtain bases for the nullspace and rowspace directly from the coefficient matrix. Alternatively, one can recast the SDOP as a slightly larger one using the Self-Dual embedding model, for which the nullspace basis matrix can be constructed trivially from the coefficient matrix (i.e., without needing factorization or Gaussian elimination). We can also calculate a basis of $\Null (\Acal_s)$ using either Gaussian elimination or the QR-Factorization. Recall that $\Acal_s \in \R{m \times \frac{n(n+1)}{2}}$, so $\Acal_s^{\top} \in \R{\frac{n(n+1)}{2} \times m}$. The QR factorization of $\Acal^{\top}_s$ is defined as:
$$ \Acal_s^{\top} = \begin{bmatrix} Q_1 & Q_2 \end{bmatrix} \begin{bmatrix} R \\ 0 \end{bmatrix}, $$
where $Q_1 \in \R{\frac{n(n+1)}{2} \times m}$ and $Q_2 \in \R{\frac{n(n+1)}{2} \times \left(\frac{n(n+1)}{2} - m \right) }$. Then, it follows that the columns of $Q_2$ form a basis for the null space of $\Acal_s.$ Hence, introducing a variable $\Delta z \in \R{ \left(\frac{n(n+1)}{2} - m \right)}$, $\Delta X$ and $\Delta S$ can be written as
\begin{subequations}
\begin{align}
    \textbf{svec}(\Delta X) &= Q_2 \Delta z, \label{e:pBasis}\\
    \textbf{svec}(\Delta S) &= - \Acal_s^{\top} \Delta y. \label{e:dBasis}
\end{align}
\end{subequations}

The following result defines the Newton linear system to be solved at each iteration of the IF-QIPM.
\begin{proposition}
 Let $P$ be a nonsingular scaling matrix from the Monteiro-Zhang family of search directions, and define 
 \begin{align*}
    \Ecal_s &= (P \otimes_s P^{- \top} S),\\
    \Fcal_s &= (P X \otimes_s P^{-\top}), \\
     R^c_s  &= \svec( \sigma \nu I - H_P (XS)).
 \end{align*}
The Newton linear system for the IF-QIPM is given by
\begin{equation}\label{e:IF-Newton}
    \begin{bmatrix} 
    \Ecal_s Q_2 & \Fcal_s (-\Acal_s^{\top}) 
    \end{bmatrix}
        \begin{bmatrix} \Delta z \\ \Delta y
    \end{bmatrix} =
    R_s^c.
\end{equation}
\end{proposition}

\begin{proof}
Substituting equations \eqref{e:pBasis} and \eqref{e:dBasis} into the left-hand side of the third equation in system \eqref{e37b} yields
 \begin{align*}  
     (P \otimes_s P^{- \top} S) \svec (\Delta X) + (P X \otimes_s P^{-\top}) \svec (\Delta S) &= (P \otimes_s P^{- \top} S)( Q_2 \Delta z)  + (P X \otimes_s P^{-\top}) (- \Acal_s^{\top} \Delta y) \\
     &= \begin{bmatrix} \Ecal_s Q_2 & \Fcal_s (-\Acal_s^{\top}) 
    \end{bmatrix}
        \begin{bmatrix} \Delta z \\ \Delta y
    \end{bmatrix}.
 \end{align*} 
\end{proof}
Observe that the coefficient matrix of equation \eqref{e:IF-Newton} is full rank, as the columns of $A^{\top}_s$ and $Q_2$ are orthogonal, \textit{and} the columns of $A^{\top}_s$ and $Q_2$ are linearly independent. Thus, equation \eqref{e:IF-Newton} has a unique solution. Upon using QLSA to solve the system \eqref{e:IF-Newton} for $\ket{\Delta z \circ \Delta y}$, we will need to map classical estimates $(\overbar{\Delta z}, \overbar{\Delta y})$ of the solution to a classical estimate $(\overbar{\Delta X}, \overbar{\Delta y}, \overbar{\Delta S})$ in order to proceed to the next iteration. Define $\Delta X$ and $\Delta S$ by \eqref{e:pBasis} and \eqref{e:dBasis}, respectively. As $(\Delta z, \Delta y)$ is the unique solution of \eqref{e:IF-Newton}, consequently $\Delta X$ and $\Delta S$ are uniquely determined as well.

Our use of inexact tomography causes $\Delta z$ and $\Delta y$ to be calculated with some error. Yet, by using the bases of the nullspace and the rowspace to calculate $\Delta X$ and $\Delta S$, we ensure that these search directions are always coming from orthogonal subspaces, regardless of the errors introduced by tomography or the length of the stepsize. We formalize this result in the following proposition. 
 
\begin{proposition} For a given $X \in \Pcal^0$ and $(y,S) \in \Dcal^0$, let $(\overbar{\Delta z}, \overbar{\Delta y})$ be an estimate of the solution to \eqref{e:IF-Newton}. Then, 
$$
    \overbar{\Delta z} = \Delta z + \xi_{z} ~\text{and}~
    \overbar{\Delta y} = \Delta y+ \xi_{y},
$$
where $\xi_{z}$ and $\xi_y$ denote the error terms of system \eqref{e:IF-Newton}. 
Let $\overbar{\Delta X}$ and $\overbar{\Delta S}$ be given by \eqref{e:pBasis} and \eqref{e:dBasis}, respectively. Then for any step size $\alpha$, we have
    \begin{align*}
        \Acal_s (\svec(X) + \alpha \svec(\overbar{\Delta X})) &= b,~\text{and}~
        \Acal_s^{\top} (y + \alpha \overbar{\Delta y}) + (\svec(S) + \alpha \svec (\overbar{\Delta S})) = \svec(C).
    \end{align*}
\end{proposition}

\begin{proof}
  First, note that for any $X \in \Pcal$ and $(y, S) \in \Dcal$, we have
  \begin{align*}
      \Acal_s \svec(X)  &= b \\
      \Acal_s^{\top} y + \svec(S) &= \svec(C).
\end{align*}
  Therefore,
  \begin{align*}
      \Acal_s (\svec(X) + \alpha \svec( \overbar{\Delta X})) &= \Acal_s \svec( X) + \alpha \Acal_s \svec( \overbar{\Delta X}) \\
      &= b + \alpha \Acal_s [Q_2 (\overbar{\Delta z})] \\
      &= b + \alpha [\Acal_s Q_2 (\Delta z + \xi_z)] \\
      &= b + \alpha [0 \cdot (\Delta z + \xi_z)] \\
      &= b. 
  \end{align*}
  Similarly, 
\begin{align*}
    \Acal_s^{\top} (y + \alpha \overbar{\Delta y}) + (\svec(S) + \alpha \svec (\overbar{\Delta S})) &=  [\Acal_s^{\top} y + \svec(S)] + [\Acal_s^{\top} \alpha \overbar{\Delta y} + \alpha \svec (\overbar{\Delta S})] \\
    &= \svec (C)  + \alpha [\Acal_s^{\top} \overbar{\Delta y} +  \svec (\overbar{\Delta S})] \\
    &= \svec (C)  + \alpha [\Acal_s^{\top} (\overbar{\Delta y} - \overbar{\Delta y} )] \\
    &= \svec (C). 
\end{align*}
\end{proof}

Having established the fact that our search directions computed from inexact estimates of the solution to \eqref{e:IF-Newton} preserve primal and dual feasibility, we next certify the validity of using the system \eqref{e:IF-Newton} to obtain a solution to the system \eqref{e:svecTodd}. Further, we demonstrate that if the system \eqref{e:svecTodd} has a unique solution, than the system \eqref{e:IF-Newton} does as well. The proof of the next result is straightforward using the definition of the quantities involved.

\begin{proposition}\label{p:QuantToClassic}
Suppose $X$ and $S$ are positive definite. Then, any solution $(\Delta z, \Delta y)$ of the system \eqref{e:IF-Newton}, provides a solution $(\Delta X, \Delta y, \Delta S)$ to the system \eqref{e:svecTodd}. 
\end{proposition}

The following result from \cite{monteiro1998polynomial} plays a crucial role in the analysis required to establish polynomial convergence of our IF-QIPM.

\begin{lemma}[Lemma 2.1 in \cite{monteiro1998polynomial}]
Suppose that $(X, S) \in \Scal_{++}^n \times \Scal_{++}^n$, $Q \in \R{n \times n}$ is a nonsingular matrix and $\nu = \frac{X \bullet S}{n}$. Then: 
\begin{itemize}
    \item[(a)] $d (\widetilde{X}, \widetilde{S}) = d(X,S)$ where $\widetilde{X} = QXQ^{\top}$ and $\widetilde{S} = Q^{-\top} S Q^{-1}$.\label{lem2.1a}
    \item[(b)] $d(X, S) \leq \| H_Q (XS - \nu I) \|_F$ with equality holding if $QXSQ^{-1} \in \Scal^n$. \label{lem2.1b}
\end{itemize}
\end{lemma}

  

\section{Technical Results}\label{s:tech}
In this section, we present the technical results required to prove that our IF-QIPM for SDO exhibits a polynomial iteration complexity. Our convergence analysis can be viewed as a generalization of Gondzio's \cite{gondzio2013convergence} analysis from LCQO to SDO. To accomplish this, we adapt the work of Monteiro \cite{monteiro1998polynomial} to account for the inexactness of the complementarity condition in order to prove polynomial convergence of the IF-QIPM presented in Section \ref{s:IFQIPM}. Though in this work we chose the Nesterov-Todd direction $P = W^{-1/2}$, the convergence analysis we perform is presented in general terms of $P$, and holds for any member of the Monteiro-Zhang family of search directions \cite{monteiro1998unified}, which includes the NT, AHO and HKM directions. We leave the results used to establish convergence of the II-QIPM to Appendix \ref{app:zhou_toh}. \\

Following \cite{monteiro1998polynomial}, in this section we assume that $(X, y, S) \in \Scal_{++}^n \times \R{m} \times \Scal_{++}^n$ and that $P \in \R{n \times n}$ is a nonsingular matrix. Additionally, we assume $(\Delta X, \Delta y, \Delta S)$ is a solution of the system 
\begin{subequations}\label{e:6}
\begin{align}
 A^{(i)} \bullet \Delta X &= b_i -  A^{(i)} \bullet X, ~~\forall i \in [m], \label{e:6c}  \\
 \sum_{i = 1}^m \Delta y_i A^{(i)}+ \Delta S &= C - S - \sum_{i = 1}^m y_i A^{(i)}, \label{e:6b}\\
   H_P (\Delta X S + X \Delta S) &= \sigma \nu I - H_P (XS) + R^r, \label{e:6a}
\end{align}
\end{subequations}
for some $\sigma \in [0, 1]$. Recall that $(\Delta X, \Delta y, \Delta S)$ is a solution to the system \eqref{e:svecTodd} obtained from our classical estimate $(\Delta z, \Delta y)$ of the solution $\ket{ \Delta z \circ \Delta y}$ to the Newton linear system \eqref{e:IF-Newton}. That is,
\begin{equation}\label{e:t40}
     (\Delta X, \Delta y, \Delta S) = \left(\smat (Q_2 \Delta z), \Delta y, - \smat (\Acal_s^{\top} \Delta y) \right),
\end{equation}
where $\Delta X$ and $\Delta S$ are computed according to \eqref{e:pBasis} and \eqref{e:dBasis}, respectively. Further, for $\alpha \in \R{}$ we define the quantities:
\begin{subequations}
\begin{align}
    \wtX &= PXP^{\top}, &\wtS &= P^{-\top} S P^{-1}, \label{e8} \\
    \wtDX &= P \Delta XP^{\top}, &\wtDS &= P^{-\top} \Delta S P^{-1},  \label{e9} \\
    X(\alpha) &= X + \alpha \Delta X, &S(\alpha) &= S + \alpha \Delta S,~ &y(\alpha) = y + \alpha \Delta y,  \label{e10}\\
    \wtX (\alpha) &= P X(\alpha) P^{\top} = \wtX + \alpha \wtDX, \label{e11} \\
    \wtS (\alpha) &= P^{-\top} S(\alpha) P^{-1} = \wtS + \alpha \wtDS, \label{e12}\\
    \Wcal_X &= \wtX^{-1/2} \left[ \wtDX \wtS + \wtX \wtDS + \wtX \wtS - \sigma \nu I - \tr^r\right] \wtX^{1/2}, \label{e13}
\end{align}
\end{subequations}

and
\begin{equation} \label{e14}
      \nu = \frac{X \bullet S}{n} = \frac{\wtX \bullet \wtS}{n}, ~~~~\nu (\alpha)= \frac{X (\alpha)\bullet S(\alpha)}{n} = \frac{\wtX (\alpha) \bullet \wtS(\alpha)}{n}.  
\end{equation}

\begin{lemma}[Lemma 3.1 in \cite{monteiro1998polynomial}]
$(\Delta X,  \Delta y, \Delta S)$ is a solution of the system \eqref{e:6} if and only if $(\wtDX, \Delta y, \wtDS)$ is a solution of the system 
\begin{subequations}\label{e:15}
\begin{align}
 \widetilde{A}^{(i)} \bullet \Delta \wtX &= b_i -  \widetilde{A}^{(i)} \bullet \wtX, ~~\forall i \in [m], \label{e:15c}  \\
\sum_{i = 1}^m \Delta y_i \widetilde{A}^{(i)}+ \wtDS &= \widetilde{C} - \wtS - \sum_{i = 1}^m y_i \widetilde{A}^{(i)} \label{e:15b}\\
   H_I (\wtDX \wtS + \wtX \wtDS) &= \sigma \nu I - H_I (\wtX \wtS) + \tr^r \label{e:15a}
\end{align}
where $\tr^r = P^{-\top} R^r P^{-1}$, $\widetilde{C} = P^{-\top} C P^{-1}$ and $\widetilde{A}^{(i)} = P^{-\top} \widetilde{A}^{(i)} P^{-1}$ for $i \in [m]$.
\end{subequations}
\end{lemma}

\begin{proof}
  The proof follows as a consequence of equations \eqref{e:6}, \eqref{e8} and \eqref{e9}.
\end{proof}

For the rest of this section, we work under the assumption that
$$ (X, y, S) \in \Pcal^0 \times \Dcal^0.$$
We remark that in the scaled setting our assumption for the residual term in \eqref{e:RcAssumption1} is given by:
\begin{equation}\label{e:scaledBoundr}
     \| \tr^r \|_F \leq \beta \| \wtR^c  \|_F, \tag{AR2}
\end{equation}
where $\beta \in (0,1)$ and $\wtR^c = \sigma \nu I - H_I (\wtX \wtS).$ For the analysis in this section we assume that the error tolerances are chosen so that the residual term $\tr^r$for the scaled problem \eqref{e:15} in each iteration satisfies
\begin{equation}\label{e:scaledBound}
    \| \wtX^{-1/2} \tr^r \wtX^{1/2} \|_F \leq \beta \|\wtX^{-1/2} \wtR^c  \wtX^{1/2} \|_F \tag{AR3}
\end{equation}
for some $\beta \in (0,1)$. Note that these assumptions that we make regarding the residual are mild. In light of Lemma \ref{lem2.1b}(a), $R^r$ and $\tr^r$ satisfy the same Frobenius norm bound, and the right hand sides of the inequalities in \eqref{e:RcAssumption1}, \eqref{e:scaledBoundr} and \eqref{e:scaledBound} are all invariant under the stated scaling. 

\begin{lemma}[Lemma 3.3 in \cite{monteiro1998polynomial}]
The following relations hold: 
\begin{align}
    H_{\wtX^{1/2}} (\Wcal_X) &= 0, \label{e16} \\
    \wtDX \bullet \wtDS &= \Delta X \bullet \Delta S = 0. \label{e17}
\end{align}
\end{lemma}

\begin{proof}
Equation \eqref{e16} is a direct consequence of \eqref{e:15a}. Further from our assumption 
$$ (X, y, S) \in \Pcal^0 \times \Dcal^0,$$
and applying equations \eqref{e:6a}, \eqref{e:6b} and \eqref{e9} yields the two identities in \eqref{e17}. 
\end{proof}

The following result can be viewed as the inexact analogue of Lemma 3.4 in \cite{monteiro1998polynomial}. 
\begin{lemma}
For all $\alpha \in \R{}$ we have
\begin{align}
    \nu (\alpha) &= (1 - \alpha + \sigma \alpha) \nu + \alpha \frac{\trace{(R^r)}}{n}, \label{e18} \\
    \wtX^{-1/2} \left[ \wtX (\alpha) \wtS (\alpha) - \nu (\alpha) I \right] \wtX^{1/2} &= (1-\alpha)( \wtX^{1/2}  \wtS \wtX^{1/2}  - \nu I) + \alpha \Wcal_X + \alpha^2 \wtX^{-1/2}\wtDX \wtDS\wtX^{1/2}  \nonumber \\ 
    &~~~+ \alpha \wtX^{-1/2}  \left[ \tr^r - \frac{\trace{(R^r)}}{n} I \right] \wtX^{1/2}. \label{e20}
\end{align}
\end{lemma}

\begin{proof}
From \eqref{e10} we have
$$ X(\alpha) S(\alpha) = (X + \alpha \Delta X)(S + \alpha \Delta S) = XS + \alpha (X \Delta S + \Delta X S) + \alpha^2 \Delta X \Delta S.$$
Hence, by the linearity of $H_P(\cdot)$ and equation \eqref{e:6a}, it follows
\begin{align*}
     X(\alpha) S(\alpha) &= (X + \alpha \Delta X)(S + \alpha \Delta S) = H_P(XS) + \alpha H_P (X \Delta S + \Delta X S) + \alpha^2 H_P (\Delta X \Delta S) \\
    &=  H_P(XS) + \alpha [H_P ( \sigma \nu I - XS) + R^r] + \alpha^2 H_P (\Delta X \Delta S) \\
    &=  (1-\alpha) H_P(XS) + \alpha \sigma \nu I + \alpha R^r + \alpha^2 H_P (\Delta X \Delta S). 
\end{align*}
Now, for $M \in \R{n \times n}$ we have\footnote{This is due to the fact that the eigenvalues of $M$ are invariant under the transformation $H_P (M)$.} $\trace{(H_P (M))} = \trace{(M)}$, and thus
\begin{align*}
      X(\alpha) \bullet S(\alpha) &= \trace{[X(\alpha) S(\alpha) ]} = \trace{[H_P (X(\alpha) S(\alpha))]} \\
      &= \trace{[ (1-\alpha) H_P(XS) + \alpha \sigma \nu I + \alpha R^r + \alpha^2 H_P (\Delta X \Delta S)  ]}\\
      &= (1-\alpha) \trace{[H_P (XS)]} + \alpha \sigma \nu \trace{(I)} + \alpha \trace{(R^r)} + \alpha^2 \trace{[H_P (\Delta X \Delta S)]} \\
      &= (1-\alpha) X \bullet S + \alpha \sigma \nu n + \alpha \trace{(R^r)} + \alpha^2 \Delta X \bullet \Delta S \\
      &= (1-\alpha) X \bullet  S + \alpha \sigma \nu n + \alpha \trace{(R^r)} ,
\end{align*}
where the final equality follows from \eqref{e17}. Dividing the above expression by $n$ and applying \eqref{e14} yields \eqref{e18}. \\

\noindent From here, we apply equations \eqref{e11}, \eqref{e12}, \eqref{e18} and \eqref{e13}, and it follows:
\begin{align*}
    \wtX (\alpha) \wtS (\alpha) - \nu (\alpha) I &= (\wtX + \alpha \wtDX ) (\wtS + \alpha \wtDS) - \nu (\alpha) I \\
        &= \wtX \wtS + \alpha (\wtDX \wtS + \wtX \wtDS) + \alpha^2 \wtDX \wtDS -  \nu (\alpha)  I \\
        &= (1-\alpha) \wtX \wtS + \alpha \wtX \wtS + \alpha (\wtDX \wtS + \wtX \wtDS) + \alpha^2 \wtDX \wtDS -  \nu (\alpha)  I \\ 
        &= (1-\alpha) \wtX \wtS + \alpha (\wtDX \wtS + \wtX \wtDS + \wtX \wtS) + \alpha^2 \wtDX \wtDS -  \nu (\alpha)  I \\ 
        &= (1-\alpha) \wtX \wtS + \alpha (\wtDX \wtS + \wtX \wtDS + \wtX \wtS) + \alpha^2 \wtDX \wtDS -  \nu (\alpha)  I + \underbrace{(\alpha \tr^r - \alpha \tr^r)}_{=0}\\ 
        &= (1-\alpha) \wtX \wtS + \alpha (\wtDX \wtS + \wtX \wtDS + \wtX \wtS) + \alpha^2 \wtDX \wtDS -  \left[ (1 - \alpha + \sigma \alpha) \nu + \frac{\alpha \trace{(R^r)}}{n} \right]  I \\
        &~~~~~~+ (\alpha \tr^r - \alpha \tr^r)\\  
        &= (1-\alpha) (\wtX \wtS - \nu I) + \alpha (\wtDX \wtS + \wtX \wtDS + \wtX \wtS -  \sigma \nu I - \tr^r) + \alpha^2 \wtDX \wtDS \\
        &~~~~~~+ \alpha \left[ \tr^r - \frac{\trace{(R^r)}}{n} I \right]    \\  
    &= (1-\alpha)( \wtX \wtS - \nu I) + \alpha \wtX^{1/2} \Wcal_X \wtX^{-1/2} + \alpha^2 \wtDX \wtDS + \alpha \left[ \tr^r - \frac{ \trace{(R^r)}}{n} I \right],
\end{align*}
which holds for any $\alpha \in \R{}$. Multiplying on the left by $\wtX^{-1/2}$ and on the right by $\wtX^{1/2}$ gives \eqref{e20}. 
\end{proof}

The following result from \cite{monteiro1998polynomial} plays a crucial role in the analysis. 

\begin{lemma}[Lemma 3.5 in \cite{monteiro1998polynomial}]\label{lem3.5}
Let $M \in \R{n \times n}$ be such that $H_Q (M) = 0$ for some nonsingular $Q \in \R{n \times n}$. Then, 
\begin{align}
    \| H_I (M) \|_F &\leq \frac{1}{2} \| M - M^{\top} \|_F, \label{e21} \\
    \| M \|_F &\leq \frac{\sqrt{2}}{2} \| M - M^{\top} \|_F  \label{e22}.
\end{align}
In particular, if $M = U_1 + U_2$ for some $U_1 \in \Scal^n$ and $U_2 \in \R{n \times n}$, then
$$ \| M \|_F \leq \sqrt{2} \| U_2 \|_F.$$
\end{lemma}


We can now use the above result in order to prove the following lemma, which is adapted from Lemma 3.6 in \cite{monteiro1998polynomial}. 

\begin{lemma}\label{lem3.6}
For every $\theta \in \R{}$, we have
\begin{align}
      \| \Wcal_X \|_F &\leq \sqrt{2} \delta_x \left\| \wtS^{1/2} \wtX^{1/2} -  \theta \nu \wtS^{-1/2} \wtX^{-1/2} \right\| + \sqrt{2} \left\| \wtX^{-1/2} \tr^r \wtX^{1/2} \right\|_F,\label{e23}
\end{align}
and
\begin{align}
        \left\|  \wtX^{-1/2} \left[ \wtX (\alpha) \wtS (\alpha) - \nu (\alpha) I \right] \wtX^{1/2} \right\|_F &\leq (1-\alpha) d (\wtX, \wtS) + \alpha^2 \delta_x \delta_s +  \alpha \left\| \wtX^{-1/2}  \left[ \tr^r - \frac{\trace{(R^r)}}{n} I \right] \wtX^{1/2}   \right\| \nonumber \\
        &~~~+ \alpha \sqrt{2} \left\| \wtS^{1/2} \wtX^{1/2} - \theta \nu \wtS^{-1/2} \wtX^{-1/2}  \right\| + \sqrt{2} \| \wtX^{-1/2} \tr^r \wtX^{1/2} \| ,\label{e24}
\end{align}
for all $\alpha \in [0,1]$, where
\begin{equation}\label{e25}
    \delta_x = \left\| \wtX^{-1/2} \wtDX \wtS^{1/2} \right\|_F,~~~ \delta_s = \left\| \wtS^{-1/2} \wtDS \wtX^{1/2} \right\|_F.
\end{equation}
\end{lemma}

\begin{proof}
We follow a similar strategy to \cite{monteiro1998polynomial}, but modify the decomposition (namely the term $U_2$), to account for the inexactness residual term. Let $\theta \in \R{}$ be given. By \eqref{e13}, we have $\Wcal_X = U_1 + U_2$, where
\begin{align*}
    U_1 &= \wtX^{1/2} \wtDS \wtX^{1/2} + \theta \nu \wtX^{-1/2} \wtDX \wtX^{-1/2} + \wtX^{1/2} \wtS \wtX^{1/2} - \sigma \nu I \in \Scal^n,~\text{and} \\
    U_2 &= \wtX^{-1/2} \wtDX \wtS^{1/2} \left(\wtS^{1/2} \wtX^{1/2} - \theta \nu \wtS^{-1/2} \wtX^{-1/2}   \right) - \wtX^{-1/2} \tr^r \wtX^{1/2}.
\end{align*}
Then, applying \eqref{e16}, it follows $H_Q (\Wcal_X) = 0$ for $Q = \wtX^{1/2}$. This fact, combined with Lemma \ref{lem3.5} and the fact that $\| AB \|_F \leq \| A \|_F \| B \|$ for every $A, B \in \R{n \times n}$, from the definition of $U_2$ and $\delta_x$ it follows
\begin{align*}
    \| \Wcal_X \|_F \leq \sqrt{2} \| U_2 \|_F \leq \sqrt{2} \delta_x \left\| \wtS^{1/2} \wtX^{1/2} -  \theta \nu \wtS^{-1/2} \wtX^{-1/2} \right\| + \sqrt{2} \left\| \wtX^{-1/2} \tr^r \wtX^{1/2} \right\|_F,
\end{align*} 
i.e., \eqref{e23} holds. From here, we apply \eqref{e:7} and \eqref{e20}, and note that for $\alpha \in [0,1]$, by the definition of $\delta_x$ and $\delta_s$ we have:
\small
\begin{align*}
     &\left\|  \wtX^{-1/2} \left[ \wtX (\alpha) \wtS (\alpha) - \nu (\alpha) I \right] \wtX^{1/2} \right\|_F \\
     &\quad \quad\ \leq (1-\alpha) \left\| \wtX^{1/2} \wtS \wtX^{1/2} - \nu I \right\|_F + \alpha \| \Wcal_X \|_F + \alpha \delta_x \delta_s 
     + \alpha \left\| \wtX^{-1/2}  \left[ \tr^r - \frac{\trace{(R^r)}}{n} I \right] \wtX^{1/2}   \right\|\\
     &\quad \quad = (1-\alpha) d (\wtX, \wtS) + \alpha \| \Wcal_X \|_F + \alpha^2 \delta_x \delta_s + \alpha \left\| \wtX^{-1/2}  \left[ \tr^r - \frac{\trace{(R^r)}}{n} I \right] \wtX^{1/2}   \right\|.
\end{align*}
\normalsize
Combining the above result with \eqref{e23} implies \eqref{e24} and the proof is complete. 
\end{proof}

For ease of notation, we follow \cite{monteiro1998polynomial} in defining the quantity
 \begin{equation}\label{e26}
     \Phi_{\theta} (A, B) = \left\| A^{1/2} B^{1/2} - \theta \frac{A \bullet B}{n} A^{-1/2} B^{-1/2} \right\|_F
 \end{equation}
 for every $(A, B) \in \Scal_{++}^n \times \Scal_{++}^n$ and $\theta \in \R{}$. 
 
\begin{lemma}[Lemma 3.7 in \cite{monteiro1998polynomial}] If $d(X, S) \leq \gamma \nu$ for some $\gamma \in (0,1)$, then
\begin{subequations}
\begin{align}
    \left\| \wtX^{-1/2} \wtS^{-1/2}  \right\|^2 &\leq \frac{1}{(1-\gamma) \nu}, \label{e27}\\
    \left[ \Phi_{\theta} (\wtX, \wtS)  \right]^2 &\leq \frac{\gamma^2 + (1-\theta)^2 n}{1- \gamma} \nu. \label{e28}
\end{align}
\end{subequations}
\end{lemma}

The following result establishes upper bounds on the Frobenius norms of terms involving the residuals. 

 \begin{lemma}\label{lemmaRbounds}
Suppose that the error tolerances are chosen such that the assumptions \eqref{e:RcAssumption1},  \eqref{e:scaledBoundr} and \eqref{e:scaledBound} hold. For $\beta \in (0, 1)$ we have
\begin{subequations}
\begin{align}
\| R^r \|_F \leq \beta \gamma \sigma \nu,~~\| \tr^r \|_F &\leq \beta \gamma \sigma \nu, \label{eRcBound} \\
\left\| \wtX^{-1/2} \tr^r \wtX^{1/2} \right\|_F &\leq \beta \gamma \sigma \nu, \label{eXRcBound} \\
\left\| \wtX^{-1/2} \left[ \tr^r - \frac{ \trace{(R^r)}}{n} I \right] \wtX^{1/2}   \right\|_F &\leq 2 \beta \gamma \sigma \nu.\label{eRcBoundTrace}
\end{align}
\end{subequations}
 \end{lemma}
 
 \begin{proof}
\noindent From equation \eqref{e:15a}, we have  $\wtR^c = \sigma \nu I - H_I(\wtX, \wtS)$. We first note that 
\begin{align*}
     \wtX^{-1/2}  \wtR^c \wtX^{1/2} =  \wtX^{-1/2} H_I(\sigma \nu I - \wtX \wtS) \wtX^{1/2} = \wtX^{-1/2}  [\sigma \nu I - \wtX \wtS] \wtX^{1/2} 
    &= \sigma \nu I - \wtX^{-1/2}  \wtX \wtS \wtX^{1/2}\\
    &= \sigma \nu I - \wtX^{1/2} \wtS \wtX^{1/2}.
\end{align*}
Thus, it follows
$$ \| \wtX^{-1/2}  \wtR^c \wtX^{1/2} \|_F = \| \sigma \nu I - \wtX^{1/2} \wtS \wtX^{1/2} \|_F = \|\wtX^{1/2} \wtS \wtX^{1/2} - \sigma \nu I \|_F = d_{\sigma} (\wtX, \wtS).$$
Further, applying Lemma \ref{lem2.1b}(a) one can observe
$$\| \wtR^c \|_F =  \|H_I (\wtX, \wtS)  - \sigma \nu I \|_F = d_{\sigma} (\wtX, \wtS) = d_{\sigma} (X, S) \leq \sigma \gamma \nu.$$
Combining this fact with assumptions \eqref{e:RcAssumption1} and \eqref{e:scaledBoundr} we have
\begin{align*}
\| R^r \|_F &\leq \beta \| R^c \|_F \leq \beta \sigma \gamma \nu,~\text{and}~
    \| \tr^r \|_F \leq \beta \| \wtR^c \|_F \leq \beta \sigma \gamma \nu,
\end{align*}
and hence \eqref{eRcBound} holds. By assumption \eqref{e:scaledBound}, it follows that
$$ \| \wtX^{-1/2} \tr^r \wtX^{1/2} \|_F \leq \beta \|\wtX^{-1/2} \wtR^c  \wtX^{1/2} \|_F \leq \beta \sigma \gamma \nu.$$

Recall that for two matrices $A, B \in \R{n \times n}$, $\trace{(AB)} \leq \| A \|_F \| B \|_F$, thus
$$ \trace{(R^r)} = \trace{(R^r I)} \leq \|R^r \|_F \|I \|_F = \sqrt{n} \|R^r \|_F.$$
Then, noting that $\| R^r \|_F \leq \beta \| R^c \| \leq \beta \sigma \gamma \nu$, it follows
\begin{align*}
    \left\| \wtX^{-1/2} \left[ \tr^r - \frac{ \trace{(R^r)}}{n} I \right] \wtX^{1/2}   \right\|_F &\leq \left\| \wtX^{-1/2}  \tr^r  \wtX^{1/2} \right\|_F + \left\| \frac{ \trace{(R^r)}}{n} I  \right\|_F \\
    &\leq \beta \sigma \gamma \nu +   \frac{\left\vert \trace{(R^r)}  \right\vert }{n} \sqrt{n} \\
    &\leq \beta \sigma \gamma \nu + \left\|  R^r \right\|_F \leq 2 \beta \sigma \gamma \nu,
\end{align*}
which completes the proof. 
\end{proof}

In the next result, we establish an upper bound on the terms $\delta_x$ and $\delta_s$ by adapting Lemma 3.8 of \cite{monteiro1998polynomial} to the inexact setting.

\begin{lemma}\label{lemma3.8}
If $(X, y, S) \in \Ncal_F (\gamma)$ for some $\gamma > 0$ satisfying 
\begin{equation}\label{e29}
  2 \sqrt{2}  \frac{\gamma}{1-\gamma} \leq 1,
\end{equation}
and $\beta \in (0,1)$ such that 
$$ \beta \sigma  \leq \sqrt{ \frac{\gamma^2 + (1-\sigma)^2 n}{1- \gamma}},$$
then
$$ \max \{ \delta_x , \delta_s \} \leq 2 \left( \Phi_{\sigma} (\wtX , \wtS) + \sqrt{ \frac{\gamma^2 + (1-\sigma)^2 n}{1- \gamma} \nu} \right),$$
where $\delta_x$, $\delta_s$ and $\Phi_{\sigma} (\cdot, \cdot)$ are defined by \eqref{e25} and \eqref{e26}.
\end{lemma}

\begin{proof}
From \eqref{e13}, one can see 
\begin{align}\label{e30}
      \wtX^{1/2} \wtDS \wtS^{-1/2} +  \wtX^{-1/2} \wtDX \wtS^{1/2} &= \Wcal_{X} \wtX^{-1/2} \wtS^{-1/2} + \sigma \nu \wtX^{-1/2} \wtS^{-1/2} - \wtX^{1/2} \wtS^{1/2} \nonumber \\
      &+ (\wtX^{-1/2} \tr^r \wtX^{1/2}) \wtX^{-1/2} \wtS^{-1/2}. 
\end{align} 
Since $\gamma \in (0,1)$, it follows that $(1-\gamma) \in (0,1)$ so we have $(1-\gamma)  < (1 - \gamma)^{1/2}$ and hence $$ \frac{1}{(1 - \gamma)^{1/2}} <  \frac{1}{1 - \gamma}.$$

From \eqref{e17}, observe that the two terms on the left hand side of \eqref{e30} are orthogonal. Combining this fact with \eqref{e25}, \eqref{e30}, \eqref{e26}, \eqref{e14}, \eqref{e23} with $\theta = 1$; and \eqref{e27}, \eqref{e28} with $\theta = 1$ and \eqref{e29}, it follows that:
\begingroup
\allowdisplaybreaks
\begin{align*}
    \max &\{ \delta_x, \delta_s \} \leq \left( \delta_x^2 + \delta_s^2 \right)^{1/2} = \left( \left\| \wtX^{-1/2} \wtDX \wtS^{1/2} \right\|^2_F + \left\| \wtS^{-1/2} \wtDS \wtX^{1/2} \right\|^2_F \right)^{1/2} \\
    &= \left\| \wtX^{-1/2} \wtDX \wtS^{1/2} + \wtX^{1/2} \wtDS \wtS^{-1/2}  \right\|_F  \\
    &\leq \left\| \Wcal_X \right\|_F \left\| \wtX^{-1/2} \wtS^{-1/2} \right\| + \left\| \wtX^{1/2} \wtS^{1/2} - \sigma \nu \wtX^{-1/2} \wtS^{-1/2} \right\| + \left\| (\wtX^{-1/2} \tr^r \wtX^{1/2}) \wtX^{-1/2} \wtS^{-1/2} \right\|_F~~~(\text{by}~\eqref{e30})\\
      &=\left\| \Wcal_X \right\|_F \left\| \wtX^{-1/2} \wtS^{-1/2} \right\| + \left\| \wtX^{1/2} \wtS^{1/2} - \sigma \nu \wtX^{-1/2} \wtS^{-1/2} \right\| + \left\| \wtX^{-1/2} \tr^r  \wtS^{-1/2} \right\|_F\\
    &\leq \sqrt{2} \delta_x \left\| \wtX^{1/2} \wtS^{1/2} - \nu \wtX^{-1/2} \wtS^{-1/2} \right\|_F \left\| \wtX^{-1/2} \wtS^{-1/2} \right\| + \Phi_{\sigma} (\wtX, \wtS) + (1 + \sqrt{2}) \left\| \wtX^{-1/2} \tr^r \wtX^{1/2} \right\|_F \left\| \wtX^{-1/2} \wtS^{-1/2} \right\|  \\
    &\leq \sqrt{2} \delta_x \frac{\gamma \nu^{1/2}}{(1-\gamma)^{1/2}} \frac{1}{(1-\gamma)^{1/2} \nu^{1/2}} + \Phi_{\sigma} (\wtX, \wtS) + (1 + \sqrt{2}) \beta \sigma \gamma \nu \frac{1}{(1-\gamma)^{1/2} \nu^{1/2}} \\
    &= \sqrt{2} \delta_x \frac{\gamma \nu^{1/2}}{(1-\gamma)^{1/2}} \frac{1}{(1-\gamma)^{1/2} \nu^{1/2}} + \Phi_{\sigma} (\wtX, \wtS) +  \left[(1 + \sqrt{2}) \frac{\gamma}{(1-\gamma)^{1/2}} \right] \beta \sigma \nu^{1/2} \\
    &< \sqrt{2} \delta_x \frac{\gamma \nu^{1/2}}{(1-\gamma)^{1/2}} \frac{1}{(1-\gamma)^{1/2} \nu^{1/2}} + \Phi_{\sigma} (\wtX, \wtS) +  \left[(1 + \sqrt{2}) \frac{\gamma}{(1-\gamma)} \right] \beta \sigma \nu^{1/2} \\
    &\leq \sqrt{2} \delta_x \frac{\gamma \nu^{1/2}}{(1-\gamma)^{1/2}} \frac{1}{(1-\gamma)^{1/2} \nu^{1/2}} + \Phi_{\sigma} (\wtX, \wtS) +  \left[ 2\sqrt{2} \frac{\gamma}{(1-\gamma)} \right] \beta \sigma \nu^{1/2} \\
    &\leq \sqrt{2} \delta_x \frac{\gamma \nu^{1/2}}{(1-\gamma)^{1/2}} \frac{1}{(1-\gamma)^{1/2} \nu^{1/2}} + \Phi_{\sigma} (\wtX, \wtS) +  \beta \sigma \nu^{1/2} \\
    &\leq \sqrt{2} \delta_x \frac{\gamma \nu^{1/2}}{(1-\gamma)^{1/2}} \frac{1}{(1-\gamma)^{1/2} \nu^{1/2}} + \Phi_{\sigma} (\wtX, \wtS) + \sqrt{ \frac{\gamma^2 + (1-\sigma)^2 n}{1- \gamma} \nu}   \\
    &\leq \frac{\delta_x}{2} + \Phi_{\sigma} (\wtX, \wtS) + \sqrt{ \frac{\gamma^2 + (1-\sigma)^2 n}{1- \gamma} \nu} \leq \frac{1}{2} \max \{ \delta_x, \delta_s \} + \Phi_{\sigma} (\wtX, \wtS) + \sqrt{ \frac{\gamma^2 + (1-\sigma)^2 n}{1- \gamma} \nu} ,
\end{align*}
\endgroup
and the proof is complete. 
\end{proof}

The final result of this section, analogous to Lemma 3.9 of \cite{monteiro1998polynomial}, follows from the above lemmas. 
\begin{lemma}\label{lem:3.9}
If $(X, y, S) \in \Ncal_F (\gamma)$ for some $\gamma > 0$ satisfying \eqref{e29}, then, for every $\alpha \in [0, 1]$, $\beta \in (0,1)$, and $\eta > 0$, we have  
\begin{align*}
         &\left\|  \wtX^{-1/2} \left[ \wtX (\alpha) \wtS (\alpha) - \nu (\alpha) I \right] \wtX^{1/2} \right\|_F \\
         &~~~\leq \left( (1-\alpha) \gamma + 4 \sqrt{2} \alpha \frac{\gamma \left[ \gamma^2 + (1-\sigma)^2 n \right]^{1/2}}{1-\gamma} + 16 \alpha^2 \frac{\gamma^2 + (1-\sigma)^2 n}{1-\gamma} +\alpha (2+ \sqrt{2}) \beta  \sigma  \gamma   \right) \nu.
\end{align*}
\end{lemma}

\begin{proof}
We apply \eqref{e24} with $\theta = 1$, Lemma \ref{lem2.1a} (a), \eqref{e8}, the assumption $d(X,S) \leq \gamma \nu$, Lemma \ref{lemmaRbounds}, Lemma \ref{lemma3.8} and \eqref{e28} where we choose $\theta = 1$ and $\theta = \sigma$. We have:
\small
\begin{align*}
    &\left\|  \wtX^{-1/2} \left[ \wtX (\alpha) \wtS (\alpha) - \nu (\alpha) I \right] \wtX^{1/2} \right\|_F \\
    &~~\leq (1- \alpha) d(X,S) + \sqrt{2} \alpha \delta_x \Phi_1 (\wtX, \wtS) + \alpha^2 \delta_x \delta_s + \alpha  (2+ \sqrt{2}) \beta  d_{\sigma} (X,S) \\
    &~~\leq (1- \alpha) \gamma \nu + 2 \sqrt{2} \alpha  \left( \Phi_{\sigma} (\wtX , \wtS) + \sqrt{ \frac{\gamma^2 + (1-\sigma)^2 n}{1- \gamma} \nu} \right) \Phi_1 (\wtX, \wtS) + \alpha^2 \left[ 2 \Phi_{\sigma} (\wtX , \wtS) + 2 \sqrt{ \frac{\gamma^2 + (1-\sigma)^2 n}{1- \gamma} \nu} \right]^2\\ &~~~~~~+ \alpha (2+ \sqrt{2}) \beta  \sigma  \gamma \nu \\
    &~~\leq (1- \alpha) \gamma \nu + 2 \sqrt{2} \alpha  \left(2 \sqrt{ \frac{\gamma^2 + (1-\sigma)^2 n}{1- \gamma} \nu} \right) \Phi_1 (\wtX, \wtS) + \alpha^2 \left[ 4 \sqrt{ \frac{\gamma^2 + (1-\sigma)^2 n}{1- \gamma} \nu} \right]^2 + \alpha (2+ \sqrt{2}) \beta  \sigma  \gamma \nu \\
     &~~= (1- \alpha) \gamma \nu + 4 \sqrt{2} \alpha  \sqrt{ \frac{\gamma^2 + (1-\sigma)^2 n}{1- \gamma} \nu} \cdot \Phi_1 (\wtX, \wtS) + 16 \alpha^2 \left[  \sqrt{ \frac{\gamma^2 + (1-\sigma)^2 n}{1- \gamma} \nu} \right]^2 + \alpha (2+ \sqrt{2}) \beta  \sigma  \gamma \nu \\
     &~~\leq \left( (1-\alpha) \gamma + 4 \sqrt{2} \alpha \frac{\gamma \left[ \gamma^2 + (1-\sigma)^2 n \right]^{1/2}}{1-\gamma} + 16 \alpha^2 \frac{\gamma^2 + (1-\sigma)^2 n}{1-\gamma} +\alpha (2+ \sqrt{2}) \beta  \sigma  \gamma   \right) \nu.
\end{align*}
\normalsize
The proof is complete. 
\end{proof}

\section{Quantum Interior Point Methods for SDO} \label{s:IFQIPM}
In this section we present our Inexact-Feasible and Inexact-Infeasible QIPMs. For the IF-QIPM, we establish polynomial convergence, then discuss how to implement the block-encoding of the Newton linear system, and conduct the overall running time analysis. Though we present the II-QIPM in its entirety, we leave the discussion on convergence, and implementing the block-encodings for the corresponding Newton linear system to the Appendix. 

\subsection{An Inexact-Feasible QIPM for SDO}
We quantize a general IF-IPM and the resulting scheme is described in Algorithm \ref{alg:IF-QIPM}. There is no classical counterpart in the literature for SDO, as the nullspace representation of the Newton system destroys the symmetry of the Newton system, and is thus impractical for classical computing. The algorithm has several parameters, the most
important of which is the optimality gap tolerance $\epsilon$. Details
for all steps of Algorithm \ref{alg:IF-QIPM} are discussed
subsequently in this paper.

In our initialization steps, we must calculate bases of the nullspace and rowspace of $\Acal_s$\label{p:tlsPage}. For this step, one can either use Gaussian elimination, or the QR-factorization of $\Acal_s$, which is typically more stable from a numerical point of view. Theoretically this can be accomplished using $\Ocal(n^{2 \omega})$ classical arithmetic operations
where $\omega = 2.37$ is the matrix multiplication exponent \cite{burgisser2013algebraic, strassen1969gaussian}, although this cost can be potentially reduced much closer to $\Ocal(n^{4})$ if the initial data is sparse and we use sparse matrix-vector multiplication, see \cite{yuster2005fast}. It is also possible to trivially obtain a basis for the nullspace and the range space for ``free" if the SDO problem is given in canonical form (inequalities instead of equalities): in that case, using the Self-Dual embedding model for that form, the nullspace basis matrix can be constructed trivially too from the coefficient matrix (without needing a QR-factorization or Gaussian elimination). Crucially, the computation of these bases need only be carried out once. That is, we calculate $\Acal_s$ and $Q_2$ one time, and store these matrices in QRAM before the algorithm begins, as they remain unchanged for the duration of the algorithm. This step is akin to the preprocessing steps (e.g., Cholesky factorization, ordering of basis) typical of many optimization algorithms.

In steps \ref{IF-QIPM-inv} and \ref{IF-QIPM-mult} of Algorithm
\ref{alg:IF-QIPM} we require classical matrix multiplication and
inversion of $n \times n$ matrices. These steps take $\Ocal (n^{\omega})$ classical arithmetic operations. A detailed analysis for step \ref{IF-QIPM-mult} can be found in Proposition \ref{prop:NTcompact}. In steps \ref{IF-QIPM-solve} and \ref{IF-QIPM-tom}, we employ
a QLSA and the tomography routine which
reconstructs, with high probability, a classical description of the
solution to the Newton linear system. Following
our discussion on tomography in Section \ref{s:tomography}, this can be accomplished in $\widetilde{\Ocal}_{n, \frac{1}{\xi_k}}
\left( \frac{n^2}{\xi_k} T_{\text{LS}}\right)$ time, where, for
the NT direction, $T_{LS}$ is the
time needed to prepare and solve the NT linear system \label{p:tlsPage}, and $\xi_k$ is
the error to which we perform our tomography steps to extract a classical estimate of the solution to the Newton linear system. A discussion on how to set $\xi_k$ at each iterate is provided at the beginning of Section \ref{ss:runT}. Then, in step \ref{IF-QIPM-base}, we use our classical estimates of $\overbar{\Delta z}$ and $\overbar{\Delta y}$ to classically calculate $\overbar{\Delta X}$ and $\overbar{\Delta S}$. This step amounts to matrix-vector multiplication of dimension $n^2$, and requires $\Ocal(n^4)$ arithmetic operations. Finally, we classically update the current solutions to the SDOP $X_k$, $y_k$ and $S_k$, and
the central path parameter $\nu_k$.

\begin{algorithm}[tb]
\SetAlgoLined
\KwIn{$\epsilon, \delta > 0$;  $\sigma = 1 - \delta/\sqrt{n}$; $\beta, \gamma \in (0,1)$ \\
Choose a nonsingular matrix $P$ from Monteiro-Zhang family of search directions \\
Choose $(X_0, y_0, S_0) \in \Ncal_F (\gamma)$\\
Set $\nu_0 = \frac{X_0 \bullet S_0}{n}$ \\
Store $(X_0, y_0, S_0)$ in QRAM \\
\noindent Compute bases for $\Rcal (\Acal_s) = \Acal_s^{\top}$ and $\Null (\Acal_s) = Q_2$ \\
\noindent Store $\Acal_s$ and $Q_2$ in QRAM }
    \While{$\nu > \epsilon$}{ 
    \begin{enumerate}
    \item Update central path parameter $\nu_k \leftarrow \frac{X_k \bullet S_k}{n} $
    \item Compute matrices $X_k^{-1}$ and $S_k^{-1}$ classically, and store these matrices in QRAM  \label{IF-QIPM-inv}
    \item Compute scaling matrix $P_k$, right hand side matrix, and necessary matrix products for Newton system \eqref{e:IF-Newton} and store in QRAM   \label{IF-QIPM-mult}
        \item Using block-encodings, solve \eqref{e:IF-Newton} to construct inexact search direction $ | \Delta z_k  \circ \Delta y_k \rangle$   \label{IF-QIPM-solve}
        \item Obtain classical estimate $\overbar{\Delta z}_k, \overbar{\Delta y}_k$ of $\Delta z_k,  \Delta y_k$ using vector state tomography \label{IF-QIPM-tom}
        \item Use classical estimate $\overbar{\Delta z}_k, \overbar{\Delta y}_k$ to obtain classical estimate  $\overbar{\Delta X}_k, \overbar{\Delta S}_k$ of $\Delta X_k, \Delta S_k$ \label{IF-QIPM-base}
        \normalsize
        \item Update current solution  
        \begin{align*}
           X_{k+1} &\leftarrow X_k + \Delta X_k,~ S_{k+1} \leftarrow S_k + \Delta S_k ~\text{and}~y_{k+1} \leftarrow y_k + \Delta y_k \\
           k &\leftarrow  k+1
        \end{align*} \label{IF-QIPM-update}
    \end{enumerate}
  } 
 \caption{Inexact-Feasible Quantum Interior Point Method}
\label{alg:IF-QIPM}
\end{algorithm}

Here we establish that Algorithm \ref{alg:IF-QIPM} converges for a specific choice of the error
parameters, discuss how to construct the Newton linear
system for the Nesterov-Todd direction, and
provide a complexity analysis for the chosen parameters.

\subsubsection{Polynomial Convergence of the IF-QIPM}\label{ss:runT}
In order to have a convergent algorithm, we need to ensure that the errors $\xi_k$ introduced by tomography are properly controlled. As we discussed earlier, the tomography error bound is
relative (it is stated in additive form, but assumes that the vector
to be extracted has unit norm): to ensure that we satisfy the absolute
error bound for the unscaled vector, we need to divide $\xi_k$ by the
maximum norm of a solution vector. Now, by definition $\| R_k^c \|_F \leq \sigma \gamma \nu_k$ holds at every iteration, and we terminate once $\nu_k \leq \epsilon$. Hence, for every iteration of the algorithm before termination with an $\epsilon$-optimal solution, the inequality
    $$\| R_k^c \|_F > \sigma \gamma \epsilon,$$
    holds, where $\gamma = \frac{1}{20}$ and $\sigma \approx 1$ for large $n$. In other words, the requirement on the residual bound never becomes much smaller than $\epsilon$, if we aim at $\epsilon$-optimality. Thus, we can set
\begin{equation} \label{e:xik}
  \xi_k = \frac{\beta}{\varrho} \cdot \max  \left\{  \| R_k^c \|_F, \frac{1}{25} \epsilon \right\}. 
\end{equation}
where $\varrho$ is the maximum norm of a solution to the Newton linear
system. 

\begin{lemma}\label{l:symmetrization}
  Choosing $\xi_k$ according to \eqref{e:xik}, the inexact search
  direction $(\Delta z, \Delta y)$ computed at step \ref{IF-QIPM-tom} of
  Algorithm~\ref{alg:IF-QIPM} satisfies \eqref{e:RcAssumption1}.
\end{lemma}

\begin{proof}
For feasible IPMs, the complementarity gap is updated as $\nu_{k+1} = (1 - \frac{\delta}{\sqrt{n}}) \nu_k$, which implies that  $\| R^c_{k+1} \|_F \leq (1 - \frac{\delta}{\sqrt{n}})\| R^c_k \|_F$, due to the fact that $\| R^c_{k+1}\|_F \leq \sigma \gamma \nu_{k+1}$.
\end{proof}

Next, we provide the inexact analogue of Theorem 4.1 from \cite{monteiro1998polynomial}, which yields the analysis of one iteration of Algorithm \ref{alg:IF-QIPM} for suitable choices of $\gamma$, $\delta$ and $\beta$.  

\begin{theorem}\label{t:Converge}
Let $\gamma, \beta \in (0,1)$ and $\delta \in [0, n^{1/2})$ be constants satisfying 
\begin{equation}\label{e:31}
   \frac{2 \sqrt{2} \gamma}{1 - \gamma} \leq 1,~~ \beta \sigma  \leq \sqrt{ \frac{\gamma^2 + (1-\sigma)^2 n}{1- \gamma}},~~\beta \leq  1 -   \frac{\gamma}{\sqrt{n}} - \frac{ 21.7 (\gamma^2 + \delta^2)}{ (2 + \sqrt{2}) \left( 1 - \frac{\delta}{\sqrt{n}}\right) \gamma (1- \gamma)}.
\end{equation}
Suppose that $(X, y, S) \in \Ncal_F (\gamma)$ and let $(\Delta X, \Delta y, \Delta S)$ denote the solution that we obtain from solving system \eqref{e:IF-Newton}, where $\sigma = 1 - \delta/\sqrt{n}$, $\nu = (X \bullet S)/n$ and $P \in \R{n \times n}$ is a nonsingular matrix. Then,
\begin{itemize}
    \item[(a)] $(\widehat{X}, \hat{y}, \widehat{S}) = (X + \Delta X,  y + \Delta y, S + \Delta S) \in \Ncal_F (\gamma);$
    \item[(b)] $\widehat{X} \bullet \widehat{S} = \left( 1 - \frac{\delta}{\sqrt{n}} \right) (X \bullet S).$
\end{itemize}
\end{theorem}

\begin{proof}
In the following proof, we seek to show that a new iterate, after a step $\alpha$ in this direction remains in the neighborhood of the central path. That is, the new iterate satisfies
$$ \left\|  \wtX^{-1/2} \left[ \wtX (\alpha) \wtS (\alpha) - \nu (\alpha) I \right] \wtX^{1/2} \right\|_F \leq \gamma \nu (\alpha) =\gamma (1 - \alpha + \sigma \alpha) \nu + \alpha \gamma \frac{\trace{(R^r)}}{n}.$$
 Now, $|\trace{(R^r)}| \leq \sqrt{n} \| R^r \|_F \leq \sqrt{n} \beta \sigma \gamma \nu < \sqrt{n} \sigma \gamma \nu$. However, since we cannot make any conclusions regarding the sign of the quantity $\trace{(R^r)}$, it suffices to show 
$$ \left\|  \wtX^{-1/2} \left[ \wtX (\alpha) \wtS (\alpha) - \nu (\alpha) I \right] \wtX^{1/2} \right\|_F \leq \gamma (1 - \alpha +  \sigma \alpha) \nu - \alpha \gamma \frac{\sigma \gamma \nu}{\sqrt{n}}. ~~(*)$$ 

First, we establish an upper bound on the quantity 
$$ \left\|  \wtX^{-1/2} \left[ \wtX (\alpha) \wtS (\alpha) - \nu (\alpha) I \right] \wtX^{1/2} \right\|_F,$$
and show that this upper bound satisfies the inequality $(*)$ for appropriately chosen $\gamma, \beta$ and $\delta$. Applying Lemma \ref{lem:3.9}, along with the definition of $\sigma$ and \eqref{e:31}, for all $\alpha \in [0,1]$ we have
  \begin{align}
         &\left\|  \wtX^{-1/2} \left[ \wtX (\alpha) \wtS (\alpha) - \nu (\alpha) I \right] \wtX^{1/2} \right\|_F \nonumber \\
         &~~~\leq \left( (1-\alpha) \gamma + 4 \sqrt{2} \alpha \frac{\gamma \left[ \gamma^2 + (1-\sigma)^2 n \right]^{1/2}}{1-\gamma} + 16 \alpha^2 \frac{\gamma^2 + (1-\sigma)^2 n}{1-\gamma} +\alpha (2+ \sqrt{2}) \beta  \sigma  \gamma   \right) \nu \nonumber \\
         &~~~\leq \left( (1-\alpha) \gamma  + 21.7 \alpha \frac{\gamma^2 + (1-\sigma)^2 n}{1-\gamma} + \alpha (2 + \sqrt{2}) \beta  \sigma  \gamma  \right) \nu \nonumber \\
         &~~~=\left( (1-\alpha) \gamma  +  \alpha \frac{21.7 (\gamma^2 + (1-\sigma)^2 n)}{1-\gamma} + \alpha (2 + \sqrt{2}) \beta  \sigma  \gamma   \right) \nu \nonumber \\
        &~~~=\left( (1-\alpha) \gamma  +  \alpha \frac{21.7 (\gamma^2 + \delta^2)}{1- \gamma}  + \alpha (2 + \sqrt{2})\beta  \sigma  \gamma   \right) \nu \nonumber.~~(**)
\end{align}
Substituting the upper bound $(**)$ into the inequality $(*)$, one can see that we must have
$$ \left( \alpha \frac{ 21.7 (\gamma^2 + \delta^2)}{1- \gamma}  + \alpha (2 + \sqrt{2}) \beta  \sigma  \gamma \right) \nu \leq \alpha \left( 1 - \frac{\delta}{\sqrt{n}}\right) \gamma \nu - \alpha  \gamma  \frac{\sigma \gamma \nu}{\sqrt{n}}.$$
Dropping like terms on both sides of the inequality, this implies
$$ \frac{ 21.7 (\gamma^2 + \delta^2)}{1- \gamma}  + (2 + \sqrt{2}) \beta  \sigma  \gamma \leq \left( 1 - \frac{\delta}{\sqrt{n}}\right) \gamma -   \frac{\sigma \gamma^2}{\sqrt{n}}.$$
Using the fact that $\sigma = 1 - \delta/\sqrt{n}$, we can further simplify this expression to 
\begin{equation}\label{e:beta}
    \beta \leq  1 -   \frac{\gamma}{\sqrt{n}} - \frac{ 21.7 (\gamma^2 + \delta^2)}{ (2 + \sqrt{2}) \left( 1 - \frac{\delta}{\sqrt{n}}\right) \gamma (1- \gamma)}.
\end{equation}
Choosing $\beta \in (0,1)$ such that \eqref{e:beta} is satisfied, it follows
  \begin{align}
         \left\|  \wtX^{-1/2} \left[ \wtX (\alpha) \wtS (\alpha) - \nu (\alpha) I \right] \wtX^{1/2} \right\|_F \nonumber 
        &\leq \left( (1-\alpha) \gamma  +  \alpha \left( 1 - \frac{\delta}{\sqrt{n}}\right) \gamma  \right) \nu + \alpha \gamma \frac{\trace{(R^r)}}{n} \nonumber \\  
         &= \left( (1-\alpha) \gamma  +  \alpha \gamma \sigma  \right) \nu + \alpha \gamma \frac{\trace{(R^r)}}{n}. \label{e32}  
\end{align}

By equation \eqref{e18}
$$  \nu (\alpha) = (1 - \alpha + \sigma \alpha) \nu + \alpha \frac{\trace{(R^r)}}{n}.$$
Hence, 
$$\left\|  \wtX^{-1/2} \left[ \wtX (\alpha) \wtS (\alpha) - \nu (\alpha) I \right] \wtX^{1/2} \right\|_F \leq \gamma (1 - \alpha + \alpha \sigma ) \nu +  \gamma \frac{ \trace{(R^r)}}{n} = \gamma \nu(\alpha).$$ 
Therefore, defining 
$$ G (\alpha) = \left( \wtX^{-1/2} \wtX(\alpha) \wtS (\alpha) \wtX^{1/2} \right) / \nu (\alpha),$$
it follows that
$G(\alpha)$ is nonsingular for all $\alpha \in [0,1]$ as $\| G(\alpha) - I \| \leq \gamma < 1$. As a consequence, $\wtX (\alpha)$ and $\wtS (\alpha)$ are also nonsingular for $\alpha \in [0,1]$. 

From here, we follow \cite{monteiro1998polynomial} by combining \eqref{e11} and \eqref{e12} with a simple continuity argument which demonstrates that for each $\alpha \in [0,1]$, we have that $(\wtX (\alpha), \wtS (\alpha))$ and $( X (\alpha) , S (\alpha))$ are in $\Scal_{++}^n \times \Scal^n_{++}$. Noting our assumption $(X, y, S) \in \Pcal^0 \times \Dcal^0$, along with equations \eqref{e:6b} and \eqref{e:6c}, one can observe that $\left(X (\alpha), y (\alpha), S (\alpha) \right)$ satisfies
\begin{align*}
     A^{(i)} \bullet X - b_i &= 0,~~i = 1, \dots, m, \\
    \sum_{i \in [m]} y_i A^{(i)} + S - C &= 0,
\end{align*}
and thus $\left( X(\alpha), y(\alpha), S(\alpha) \right) \in \Pcal^0 \times \Dcal^0$ for every $\alpha \in [0,1]$. Finally, we apply Lemma \ref{lem2.1a}(a) with $(X, S) = (X(1), S(1))$ and $Q=P$, along with Lemma \ref{lem2.1b}(b) with $(X, S) = (\wtX(1), \wtS(1))$ and $Q=\wtX^{-1/2}$, and equations \eqref{e11}, \eqref{e12}, \eqref{e18} with $\alpha = 1$, and \eqref{e32} with $\alpha = 1$, which yields 
\begin{align*}
    d \left( X(1) , S(1) \right) &= d \left( \wtX(1) , \wtS(1) \right) \leq \left\| H_{\wtX^{-1/2}} \left[ \wtX (1) \wtS(1) - \nu(1) I \right] \right\|_F \\ 
    &\leq \left\| \wtX^{-1/2} \left[ \wtX (1) \wtS(1) - \nu(1) I \right] \wtX^{1/2}  \right\|_F \leq \gamma \nu (1).
\end{align*}
Therefore,  $(\widehat{X}, \hat{y}, \widehat{S}) = (X + \Delta X,  y + \Delta y, S + \Delta S) \in \Ncal_F (\gamma)$, and hence (a) holds. Further, (b) follows immediately from the definition of $\sigma$ and equation \eqref{e18} with $\alpha =1$.
\end{proof}

Following from Theorem \ref{t:Converge}, we obtain the convergence result for Algorithm \ref{alg:IF-QIPM}. 

\begin{corollary}\label{c:Converge}
Let $\gamma$, $\delta$ and $\beta$ be constants in $(0,1)$ and satisfying 
$$   \frac{2 \sqrt{2} \gamma}{1 - \gamma} \leq 1,~~ \beta \sigma  \leq \sqrt{ \frac{\gamma^2 + (1-\sigma)^2 n}{1- \gamma}},~~ \beta \leq 1 -   \frac{\gamma}{\sqrt{n}} - \frac{ 21.7 (\gamma^2 + \delta^2)}{ (2 + \sqrt{2}) \left( 1 - \frac{\delta}{\sqrt{n}}\right) \gamma (1- \gamma)}.$$
Then, every iterate $(X_k, y_k, S_k)$ generated by the IF-QIPM in Algorithm \ref{alg:IF-QIPM} is an element of the neighborhood $\Ncal_F (\gamma)$, and satisfies the relation $X_k \bullet S_k = (1- \delta / \sqrt{n})^k (X_0 \bullet S_0).$ Further, the IF-QIPM terminates in at most $\Ocal(\sqrt{n} \log (1/\epsilon))$ iterations. 
\end{corollary}

Examples of constants $\gamma$, $\delta$ and $\beta$ satisfying the conditions of Corollary \ref{c:Converge} are $\gamma = \delta = 1/20$ and $\beta = 1/4$.

\subsubsection{Implementing the Nesterov-Todd linear system}
As in \cite{kerenidis2018quantum}, we factorize the NT Newton linear system given by \eqref{e:IF-Newton}, and construct block-encodings for its factors. In this paper we use block-encodings and QSLA within an IPM by pre-computing some matrix inverses and products \textit{classically}, and storing the results in QRAM so that the block-encodings of the factors of the Newton linear system can be constructed in time $\widetilde{\Ocal}_{\frac{n}{\xi}} (1)$, where $\xi$ is some error tolerance. This is more efficient than other schemes that we explored. Note that the largest matrices appearing in the Newton system, namely $Q_2 \in \R{\frac{n(n+1)}{2} \times (\frac{n(n+1)}{2} - m)}$ and $\Acal_s^{\top} \in \R{\frac{n(n+1)}{2} \times m}$, are only calculated once before the algorithm begins so their size is not prohibitive. Hence, compared to solving the Newton linear system, classical computation of matrix powers and products has an insignificant cost, and we circumvent increased dependence on the condition number. 
 
\begin{proposition}\label{prop:ntbe1}
  Let 
  $$ M_1 = \begin{bmatrix} (P \otimes_s P^{- \top} S) Q_2 & 0
\end{bmatrix}.$$
  Let $P$, $Q_2$, and $P^{-\top} S$ be stored in QRAM. Then a $( \| M_1 \|_F, \Ocal(\log n), \xi_{M_1})$-block-encoding of $M_1$ can be constructed in time $\widetilde{\Ocal}_{\frac{n}{\xi_{M_1}}}(1).$
\end{proposition}

\begin{proof}
First, given that $Q_2$ is stored in QRAM, we can construct a $( \| Q_2 \|_F , \Ocal (\log n), \xi_{1})$-block-encoding of $Q_2$ in time $\widetilde{\Ocal}_{\frac{n}{\xi_{1}}}(1)$. Next, with $P$ and $P^{-\top} S$ stored in QRAM, we can apply Proposition \ref{prop:blockencSymKronecker} so that we can construct a $( \| P \otimes_s P^{-\top} S \|_F , \Ocal (\log n), \xi_{2})$-block-encoding of $P \otimes_s P^{-\top} S$ in time $\widetilde{\Ocal}_{\frac{n}{\xi_{2}}} (1)$. 

From here we apply Proposition \ref{prop:product} with 
$$
    \xi_{1} = \frac{\xi_{M_1}}{2 \|P \otimes_s P^{-\top} S\|_F}~\text{and}~
    \xi_{2} =\frac{\xi_{M_1}}{2 \|Q_2\|_F},
$$
yielding a $( \| M_1 \|_F, \Ocal(\log n), \xi_{M_1})$-block-encoding of $M_1$ in time $\widetilde{\Ocal}_{\frac{n}{\xi_{M_1}}}(1)$, as desired. 
\end{proof}

\begin{proposition}\label{prop:ntbe2}
  Let 
  $$ M_2 = \begin{bmatrix} 0 & (P X \otimes_s P^{-\top}) (-\Acal_s^{\top})
\end{bmatrix}.$$
Let $P^{- \top}$, $\Acal_s$, $PX$ and $P^{-\top}S$ be stored in QRAM. Then a $( \| M_2 \|_F, \Ocal(\log n), \xi_{M_2})$-block-encoding of $M_2$ can be constructed in time $\widetilde{\Ocal}\frac{n}{\xi_{M_2}}(1).$
\end{proposition}

\begin{proof}
The steps of the proof follow that of the previous result. With $- \Acal_s^{\top}$ stored in QRAM, we can construct a $( \| \Acal_s^{\top} \|_F , \Ocal (\log n), \xi_1)$-block-encoding of $-\Acal_s^{\top}$ in time $\widetilde{\Ocal}\frac{n}{\xi_1}(1)$. Further, with $PX$ and $P^{-\top}$ stored in QRAM, we apply Proposition \ref{prop:blockencSymKronecker} such that we can construct a $( \| PX \otimes_s P^{-\top} \|_F , \Ocal (\log n), \xi_{2})$-block-encoding of $PX \otimes_s P^{-\top}$ in time $\widetilde{\Ocal}_{\frac{n}{\xi_2}} (1)$. 

Applying Proposition \ref{prop:product} with 
$$
    \xi_{1} = \frac{\xi_{M_1}}{2 \|PX \otimes_s P^{-\top}\|_F}~\text{and}~
    \xi_2 =\frac{\xi_{M_1}}{2 \| \Acal \|_F},
$$
we obtain a $( \|M_2 \|_F, \Ocal(\log n), \xi_{M_2})$-block-encoding of $M_2$, and the proof is complete. 
\end{proof}

\begin{proposition}\label{prop:NTcompact}
  The Nesterov-Todd linear system matrix \eqref{e:IF-Newton} can be written compactly as:
  \begin{equation}\label{eq:ntcompactIF}
      M_{NT} = \begin{bmatrix} (P \otimes_s P^{- \top} S) Q_2 & (P X \otimes_s P^{-\top}) (-\Acal_s^{\top})
\end{bmatrix}.
  \end{equation}
 If  $P$, $P^{-\top}$, $-\Acal_s^{\top}$, $Q_2$, $PX$ and $P^{-\top}S$ are stored in QRAM. Then a $$( \| M_{NT} \|_F, \Ocal(\log n), \xi/(\| M_{NT} \|_F \kappa^2 \log^2 \frac{\kappa}{\xi}))$$
 block-encoding of $M_{NT}$ can be constructed in time $\widetilde{\Ocal}_{n, \kappa, \frac{1}{\xi}}(1).$
\end{proposition}

\begin{proof}
 Carrying out the calculations shows that \eqref{eq:ntcompactIF} corresponds to the Nesterov-Todd linear system. We construct the two following block-encodings:
 \begin{align*}
     M_1 &= \begin{bmatrix} (P \otimes_s P^{- \top} S) Q_2 & 0
\end{bmatrix}\\
    M_2 &= \begin{bmatrix} 0 & (P X \otimes_s P^{-\top}) (-\Acal_s^{\top}),
\end{bmatrix}
 \end{align*}
 using Propositions~\ref{prop:ntbe1}-\ref{prop:ntbe2}. We choose the
  precision of this step so that we obtain $$(\|M_1\|_F, \Ocal(\log n),
  \xi/(\|M_1\|_F \kappa^2 \log^2 \frac{\kappa}{\xi}))$$ and
  $(\|M_2\|_F, \Ocal(\log n), \xi/(\|M_2\|_F \kappa^2 \log^2
  \frac{\kappa}{\xi}))$-block-encodings of $M_1$ and $M_2$, respectively, where $\kappa$
  refers to the condition number of \eqref{e:IF-Newton}, here and
 in the sequel. We add these two block-encodings together using
  Proposition~\ref{prop:lincombblock}, obtaining a $(\max\{\|M_1\|_F,
  \|M_2\|_F\}, \Ocal(\log n), \xi/(\kappa^2 \log^2
  \frac{\kappa}{\xi}))$-block-encoding of \eqref{eq:ntcompactIF}, in
  time $\widetilde{\Ocal}_{n, \kappa, \frac{1}{\xi}}(1)$. Since $\max\{\|M_1\|_F, \|M_2\|_F\} =
  \Ocal(\|M_{\text{NT}}\|_F)$, we obtain the claimed result.
\end{proof}

One should note that the construction of the Newton linear system provided in Proposition~\ref{prop:NTcompact} is only one of the many possible ways to obtain a block-encoding of our Newton linear system while utilizing the Nesterov-Todd direction. Further, though the Nesterov-Todd direction is the search direction we study in this manuscript, note that our construction is written in general terms of $P$, thus this construction can be used for any suitable choice of $P$, such as the AHO ($P=I$) and HKM ($P = S^{1/2}$) directions. Note that by employing a factorization of the Newton linear system, we are able to theoretically improve the running time of implementing the block-encodings of the Newton linear system; which in turn allows us to improve the efficiency with which we solve the system. Next, we use this factorization to solve the Newton system. 

\begin{theorem}
Using the Nesterov-Todd direction $P = W^{-1/2}$, there is a quantum algorithm, which given $\ket{\svec (R^c)}$ and access to QRAM data structures encoding $P$, $P^{-\top}$, $-\Acal_s^{\top}$, $Q_2$, $PX$ and $P^{-\top}S$ , outputs a state $\xi$-close to $\ket{\Delta z \circ \Delta y }$ in time
$$ \widetilde{\Ocal}_{n, \kappa, \frac{1}{\xi}} \left(\kappa \max \{ \| M_1 \|_F, \| M_2 \|_F \} \right).$$
We can also output an estimate of $\|\Delta z \circ \Delta y  \|$ with relative error $\delta$ by increasing the running time by a factor of $\frac{1}{\delta}$. 
\end{theorem}

\begin{proof}
  This is a direct consequence of Proposition~\ref{prop:NTcompact} and
  Theorem~\ref{thm:qlsa}.
\end{proof}


\begin{proposition}\label{prop:solSize}
  The norm $\varrho$ of the solution of the quantum Newton linear system \eqref{e:IF-Newton} satisfies
  $$ \varrho = \Ocal \left( \| M_{NT}^{-1} \| \right) .$$
\end{proposition}
\begin{proof}
We have 
$$ \begin{bmatrix}
    \Delta z\\
    \Delta y
\end{bmatrix} = (M_{NT})^{-1} \svec (R^c + R^r).$$

By Lemma \ref{lemmaRbounds}, we have 
\begin{align}
    \| R^c + R^r \|_F \leq   \| R^c \|_F + \| R^r \|_F \leq (1 + \beta) \sigma \gamma \nu.
\end{align}
Choosing a constant $c_0 \geq (1 + \beta) \sigma \gamma$, it follows that $\| R^c + R^r \|_F \leq c_0 \nu$, and therefore  
$$\| R^c + R^r \|_F = \Ocal (\nu) = \Ocal (1),$$
where the final equality follows from the definition of condition number and $\nu = \Ocal(1)$. 
 To see this, note that across every iteration of the QIPM, we have $\max_k \{ \nu_k \} = \nu_0 = \frac{X_0 \bullet S_0}{n} = \Ocal(1)$, as in each iteration we strictly decrease the quantity $X \bullet S$ (which consequently decreases $\nu)$.  
 
The bound on the norm of the solution is thus given by 
$$
   \| (\Delta z, \Delta y)\|_F  = \| M_{NT}^{-1} \svec (R^c + R^r) \|_F  
   \leq \| M_{NT}^{-1} \| \cdot \| (R^c + R^r) \|_F  = \Ocal \left(\| M_{NT}^{-1} \| \right).
$$
The proof is complete. 
\end{proof}

Recall (see page \pageref{p:tlsPage}) that $T_{LS}$ denotes the time needed to prepare and solve the Newton linear system \eqref{e:IF-Newton} with the Nesterov-Todd direction. 

\begin{proposition} 
Let $\xi$ be the precision to which we prepare the quantum state proportional to the solution of the Newton linear system \eqref{e:IF-Newton} using a QLSA. Then, we have $ \varrho T_{LS} = \widetilde{\Ocal}_{n, \kappa, \frac{1}{\xi}} \left(n \kappa^2  \right)$.
\end{proposition} 
\begin{proof}
The result follows trivially from Proposition \ref{prop:solSize} and noting the fact that $\max\{\|M_1\|_F, \|M_2\|_F\} =
  \Ocal(\|M_{\text{NT}}\|_F)$. First, observe
$$
     T_{LS} = \widetilde{\Ocal}_{n, \kappa, \frac{1}{\xi}} (\kappa \max\{\|M_1\|_F, \|M_2\|_F\} ) = \widetilde{\Ocal}_{n, \kappa, \frac{1}{\xi}} (\kappa \| M_{NT} \|_F) = \widetilde{\Ocal}_{n, \kappa, \frac{1}{\xi}} \left(n \kappa \cdot \| M_{NT} \| \right).
$$
Then,
$$\varrho T_{LS} = \widetilde{\Ocal}_{n, \kappa, \frac{1}{\xi}} \left(n  \kappa \cdot \| M_{NT} \| \cdot \| M_{NT}^{-1} \|  \right) = \widetilde{\Ocal}_{n, \kappa, \frac{1}{\xi}} \left(n \kappa^2  \right),$$
where the final equality follows from the definition of condition number. 
 \end{proof}
 
 \subsection{An Inexact-Infeasible QIPM for SDO}
The II-QIPM follows a similar algorithmic structure to that of the IF-QIPM, with the key distinction being the Newton linear system to be solved at each iteration, and that we do not need to compute the bases of the nullspace and rowspace of $\Acal$. 

\begin{algorithm}[tb]
\SetAlgoLined
\KwIn{$\nu_0, \theta_0 = 1$; $\epsilon, \delta > 0$; $\eta_1 \in (0,1]$; $\gamma_1, \gamma_2, \eta_2 \in (0,1)$ \\
Pick a sequence $\{ \vartheta \}_{i =0}^{\infty}$ in $(0, 1]$, such that 
$ \bar{\vartheta} = \sum_{i = 0}^{\infty} \vartheta_i < \infty$ \\
Choose a nonsingular matrix $P$ from Monteiro-Zhang family of search directions \\
\noindent Find an initial point $(\nu_0, \theta_0, X_0, y_0, S_0) \in \Ncal_I$ \\
\noindent Store $(\nu_0, \theta_0, X_0, y_0, S_0)$ in QRAM\\} 
\While{$\nu > \epsilon$}{ 
    \begin{enumerate}
    \item Compute matrix $X_k^{-1}$ and $S_k^{-1}$ classically, and store these matrices in QRAM \label{IIQIPM-inv}
    \item Compute scaling matrix $P_k$, right hand side matrix and necessary matrix products for Newton system factorization and store in QRAM  \label{IIQIPM-mult}
    \item Using quantum linear systems algorithm, solve the system \eqref{e:inexactNewt} to construct inexact search direction $ | \svec(\Delta X_k) \circ \Delta y_k \circ \svec(\Delta S_k) \rangle$  \label{IIQIPM-solve}
    \item Obtain classical estimate of $\svec(\Delta X_k), \Delta y_k, \svec(\Delta S_k)$ using tomography  \label{IIQIPM-tom}
    \item Set \footnotesize $$\alpha_k \leftarrow \max\left\{ \alpha ~:~\alpha \in \left[0 , \min \left(1, \frac{1}{\eta_1(1+ \bar{\vartheta})} \right) \right],~(\nu_k(\alpha), \theta_k(\alpha), X_k(\alpha), y_k(\alpha), S_k(\alpha)) \in \Ncal \right\}$$ \label{IIQIPM-search}
        \normalsize
    \item Update solution
        \begin{align*}
           X_{k+1} &\leftarrow X_k + \alpha_k \Delta X_k,~ y_{k+1} \leftarrow y_k + \alpha_k \Delta y_k ~\text{and}~S_{k+1} \leftarrow S_k +  \alpha_k\Delta S_k \\
           \nu_{k+1} &\leftarrow (1 - \alpha \eta_1) \nu_k,~\text{and}~\theta_{k+1} \leftarrow (1 - \alpha \eta_1) \theta_k \label{IIQIPM-update} \\
           k &\leftarrow k+ 1
        \end{align*}
    \end{enumerate}
  } 
 \caption{Inexact-Infeasible Quantum Interior Point Method}
\label{alg:QII-IPM}
\end{algorithm}

We now discuss convergence of the algorithm, using the main result from the work of Zhou and Toh \cite{zhou2004polynomiality}, adapted to our case. Note that the proof in \cite{zhou2004polynomiality} is only
explicitly given for the HKM direction ($P = S^{1/2}$), but
convergence for the other cases can be shown analogously. For a detailed discussion on the analysis required to prove convergence, and implementing the block-encodings for the II-QIPM Newton linear system, we refer the reader to the Appendix. 

To prove convergence, we first need to show that
Algorithm~\ref{alg:QII-IPM} satisfies several properties. More
specifically, we need to show that the inexact search direction
computed at the end of step \ref{IIQIPM-tom} satisfies
\eqref{e:resbound}-\eqref{e:inexactNewt}. This requires choosing
error parameters for the QLSA and for the tomography step. Recall the
required bounds on the inexactness residuals, reported for convenience
from \eqref{e:resbound}:
\begin{equation*}
  \|r_k^d\| \le \gamma_1 \rho \tau_k \vartheta_k,~
  \|r_k^p\| \le \hat{\vartheta} \gamma_1 \rho \tau_k \vartheta_k,~
  \|r_k^c\| \le \frac{1}{2}(1 - \eta_2) \gamma_2 \theta_k \rho^2.
\end{equation*}
Note that to obtain a classical estimate in step \ref{IIQIPM-tom}, we
must determine the norm of the inexact search direction. Recall that QLSA also introduces some error in the
inexact search direction. In line with our convergence discussion for the IF-QIPM we assume that
all of the error comes from tomography.

To determine appropriate error parameters, we first discuss the
sequence $\{ \vartheta_k \}^\infty_{k=0}$, as $\vartheta_k$ determines
the rate at which the residual norms must be decreased with respect to
$\tau_k$, where $\tau \in (0, 1]$ is defined in \eqref{e:infcp}. We choose the sequence:
\begin{equation}\label{eq:vartheta}
\vartheta_k = \begin{cases}
\min \left( 1, \frac{1}{\log n \log (1 / \epsilon)} \right) & \text{if}~k < N:= (\log n) n^2 \lfloor \log (1 / \epsilon) \rfloor \\
\frac{1}{(k+1 - N)^{1.1}} &\text{if}~ k \geq N.\end{cases} 
\end{equation}
This sequence, which is admissible according to
\cite{zhou2004polynomiality}, has the advantage of a relatively slow
decrease of the error terms, and it greatly simplifies our running
time analysis. Then $\bar{\vartheta} \leq n^2 + 10$ and for a
sufficiently large $n$,
$$ \vartheta_k \approx \frac{1}{\log n \log (1/\epsilon)}$$ for all
iterations $k$ before reaching $\epsilon$-optimality. Therefore,
choosing the sequence as defined by \eqref{eq:vartheta}, the residual
norms will decrease $1/(\log n \log (1/\epsilon))$ faster than
$\tau_k$.

At iteration $k$, we then choose $\xi_k$ satisfying:
\begin{equation*}
  \xi_k \leq \frac{1}{\varrho} \min \{\|r_k^d\|, \|r_k^p\|, \|r_k^c\|\},
\end{equation*}
where $\varrho$ is the maximum norm of a solution to the Newton linear
system \eqref{e:inexactNewt} corresponding to the II-QIPM. Setting $\vartheta_k$ according to
\eqref{eq:vartheta}, for sufficiently large
$n$, we simply have $\vartheta_k = \frac{1}{\log n
  \log 1/\epsilon}$ before convergence. Hence, we can then take the error
$\xi_k$ to be:
\begin{equation}
  \label{e:xikII}
  \xi_k =  \frac{\epsilon }{\varrho\left(\log n \log 1/\epsilon \right)}.
\end{equation}

\begin{lemma}\label{l:symmetrization}
  Choosing $\xi_k$ according to \eqref{e:xikII}, the inexact search
  direction computed at the end of step \ref{IIQIPM-tom} of
  Algorithm~\ref{alg:QII-IPM} satisfies \eqref{e:resbound}.
\end{lemma}

\begin{proof}
For infeasible IPMs, the termination condition requires the feasibility gap be less than or equal to $\epsilon$, and we scale down by $\varrho$ due to normalization. 
\end{proof}

\begin{theorem}\label{t:t1zhou}
Let $\epsilon > 0$ be a small constant. Suppose that the conditions
\eqref{e5}, \eqref{e6} and \eqref{e7} hold. Suppose further that
$\xi_k$ is chosen according to \eqref{e:xik}. Then, after $k = \Ocal(n^2
\ln (1/\epsilon))$ iterations of Algorithm~\ref{alg:QII-IPM}, we have
$\theta_k \leq \epsilon$.
\end{theorem}
\begin{proof}
  Follows from \cite[Theorem.~1]{zhou2004polynomiality}, after noting
  that, as a consequence of Lemma~\ref{l:symmetrization}, the inexact
  search direction used in Algorithm~\ref{alg:QII-IPM} satisfies the
  conditions of \cite{zhou2004polynomiality}.
\end{proof}
By the definition of the neighborhood of the infeasible central path, $\Ncal_I$, we have $\tau_k \leq \theta_k$ for every iteration $k$. In view of Theorem \ref{t:t1zhou}, after $k = \Ocal(n^2
\ln (1/\epsilon))$ iterations we have $\theta_k \leq \epsilon$, which implies $\tau_k \leq \epsilon$. Accordingly, the infeasible central path parameter has been sufficiently reduced to obtain an $\epsilon$-optimal solution.

Having established convergence of both the IF- and II-QIPMs, we now turn our attention to their overall running times, this is discussed next. 
\section{Complexity}\label{s:complexity}
In this section, we formally analyze the worst case overall running times of the IF- and II-QIPMs. The per-iteration cost of either QIPM can be generally written as
$$T_{iter}  = T_{NT} + T_{progress},$$
where $T_{NT}$ is the quantum gate complexity of obtaining a classical solution to the Newton linear system \eqref{e:IF-Newton}, and $T_{progress}$ is the number of classical arithmetic operations required to prepare the next iteration. Thus, upon bounding $T_{iter}$ for the IF- and II-QIPMs one can obtain the worst-case overall complexity of each algorithm using the iteration bounds that were established in the previous section. 

We conclude this section by comparing the performance of our methods to known classical and quantum algorithms for SDO from the literature, before discussing the impact of dependence on the condition number bound $\kappa$ on the overall running time of QIPMs.

\subsection{Running time of the IF-QIPM}
The following result establishes per-iteration cost of the IF-QIPM.

\begin{proposition}\label{prop:costPerIterIF}
Each iteration $k$ of the IF-QIPM in Algorithm \ref{alg:IF-QIPM} has a quantum gate complexity of $$\Ocal \left( n^3  \frac{\kappa^2}{\epsilon}  \cdot \operatorname{polylog} \left(n, \kappa, \frac{1}{\epsilon} \right) \right)$$ and requires $\Ocal(n^4 \log (\kappa))$ classical arithmetic operations.
\end{proposition} 

\begin{proof}
In every iteration, we need to prepare and solve the NT Newton linear system \eqref{e:IF-Newton}, and obtain a classical estimate of the quantum state encoding its solution. Furthermore, we apply the state tomography algorithm to obtain a classical description of the solution of the Newton linear system; we denote the corresponding running time by $T_{TO}(T_{LS}, \xi_k)$ for precision $\xi_k$. For simplicity and without loss of generality, we consider the tomography precision $\xi_k$ but neglect the QLSA precision, as only the tomography precision appears polynomially in the final expression for the iteration cost. Applying Theorem \ref{thm:euclidean_norm_tomo} with $d = n^2$ and $T_U = T_{LS}$, running the tomography step to precision $\xi_k$ chosen according to \eqref{e:xik} implies
$$
    T_{NT}^{IF} = T_{TO}(T_{LS}, \xi_k) = \widetilde{\Ocal}_{n, \kappa, \frac{1}{\xi_k}} \left(\frac{n^2}{\xi_k} T_{LS}   \right) = \widetilde{\Ocal}_{n, \kappa, \frac{1}{\xi_k}} \left(n^2 \frac{\varrho}{\epsilon} T_{LS}  \right)
    = \widetilde{\Ocal}_{n, \kappa, \frac{1}{\epsilon}}\left(  n^3 \frac{\kappa^2}{\epsilon}  \right).
$$
QRAM accesses are required to prepare the Newton linear system and obtain a classical estimate its the solution at iteration $k$.

To progress to the iteration, we must compute the primal and dual search directions $\Delta X$ and $\Delta S$ from the nullspace representation provided in equations \eqref{e:pBasis}-\eqref{e:dBasis}, and update the solution $(X,y,S)$ to the SDO before proceeding to the next iteration. We must also compute our scaling matrix $P$ and its inverse. Computing $\Delta X$ and $\Delta S$ using \eqref{e:pBasis}-\eqref{e:dBasis} amounts to classically computing matrix-vector products in which the dimension is $n^2$, and hence requires $\Ocal(n^4)$ arithmetic operations, while updating the solution is matrix addition involving $n \times n$ matrices, and can be accomplished using $\Ocal(n^2)$ arithmetic operations. Finally, noting that all of the matrices involved in the computation of $P$ have size $n \times n$, we can compute both $P$ and $P^{-1}$ using $\Ocal(n^{\omega})$ arithmetic operations. Summarizing, we have a total of $$T^{IF}_{progress} = \Ocal(n^4 + n^2 + n^{\omega}) = \Ocal(n^4),$$ 
arithmetic operations at each iterate, as computing $\Delta X$ and $\Delta S$ from our estimate of the solution to \eqref{e:IF-Newton} is the dominant operation. 

Note however, at each iterate the QLSA circuit requires the condition number of the Newton linear system coefficient matrix $M_{NT}$ (or, an upper bound on it) as part of the input. Computing the condition number classically is too costly, and so alternatively we use the following scheme to ensure correctness of steps \ref{IF-QIPM-solve} and \ref{IF-QIPM-tom} in Algorithm \ref{alg:IF-QIPM} (i.e., to ensure that $\tilde{\kappa}$ was chosen large enough). Use $\tilde{\kappa} = 1$ as a starting guess, and then 
\begin{enumerate}
    \item Execute steps \ref{IF-QIPM-solve} and \ref{IF-QIPM-tom} of Algorithm \ref{alg:IF-QIPM} with $\tilde{\kappa}$ as input
    \item Use classical estimate of the solution to compute 
    $$ \varepsilon = \left\| M_{NT} \begin{bmatrix} \overbar{\Delta z}_k \\\overbar{\Delta y}_k\end{bmatrix} - \svec (R^c_k) \right\|_2$$
    \item If $\varepsilon \leq \varrho \xi_k$, then \textbf{accept} $(\overbar{\Delta z}_k,\overbar{\Delta y}_k)$ and terminate loop. Otherwise, set $\tilde{\kappa} = 2\cdot \tilde{\kappa}$ and \textbf{go to Step 1}.
\end{enumerate}

One can observe that $\varepsilon$ can be computed using $\Ocal (n^4)$ arithmetic operations, and at most $\log (\kappa)$ such verification steps will be necessary in the worst case. Summarizing, we arrive at a per-iteration cost of $ \log (\kappa) T_{NT}^{IF} =  \widetilde{\Ocal}_{n, \kappa, \frac{1}{\epsilon}}\left(  n^3 \frac{\kappa^2}{\epsilon}  \right)$ QRAM accesses and 
$\Ocal(n^4 \log (\kappa))$ classical arithmetic operations, which completes the proof.
\end{proof}

We are now in a position to bound the overall running time of our IF-QIPM given in Algorithm \ref{alg:IF-QIPM}, which is formalized in the following result. 
\begin{corollary}\label{corr:runtimeIF}
A quantum implementation of Algorithm \ref{alg:IF-QIPM} with access to QRAM outputs an $\epsilon$-optimal solution $(X^*, y^*, S^*)$ to the primal-dual SDO pair \eqref{e:SDO}-\eqref{e:SDD} using at most
$$ \Ocal \left( n^{3.5} \frac{\kappa^2}{\epsilon} \cdot  \operatorname{polylog} \left(n, \kappa, \frac{1}{\epsilon} \right) \right)  $$
QRAM accesses and $ \Ocal \left( n^{4.5} \cdot  \operatorname{polylog} \left( \kappa, \frac{1}{\epsilon} \right) \right) $ arithmetic operations.
\end{corollary}

\begin{proof}
By Corollary \ref{c:Converge}, the IF-QIPM requires $ \Ocal (\sqrt{n} \log (1/\epsilon))$ iterations to converge to an $\epsilon$-optimal solution. Hence, letting $\chi$ be a nonnegative constant and applying Proposition \ref{prop:costPerIterIF}, it follows that Algorithm \ref{alg:IF-QIPM} requires at most 
$$
\widetilde{\Ocal}_{n, \kappa, \frac{1}{\epsilon}} \left( \sum_{k=1}^{\chi \sqrt{n} \log (1/\epsilon)} \left(n^3 \frac{\kappa^2}{\epsilon} \right) \right)  = \widetilde{\Ocal}_{n, \kappa, \frac{1}{\epsilon}} \left(\sqrt{n} \left(n^3 \frac{\kappa^2}{\epsilon}  \right) \right),
$$
QRAM accesses, and $\Ocal \left( \sum_{k=1}^{\chi \sqrt{n} \log (1/\epsilon)} n^4 \log(\kappa)  \right) =  \widetilde{\Ocal}_{\kappa, \frac{1}{\epsilon}} \left( n^{4.5} \right) $ arithmetic operations, which completes the proof.
\end{proof}

Recall that LOPs are simply a special instance of SDO in which all of the matrices are diagonal. Hence, our algorithm can also be applied to these problems as well. 
\begin{corollary}\label{corr:rtLO}
A quantum implementation of Algorithm \ref{alg:IF-QIPM} with access to QRAM outputs an $\epsilon$-optimal solution $(X^*, y^*, S^*)= (\diag (x^*), y^*, \diag (s^*))$ to a Linear Optimization problem using at most
$$ \Ocal \left( n^{2} \frac{\kappa^2}{\epsilon} \cdot  \operatorname{polylog} \left(n, \kappa, \frac{1}{\epsilon} \right) \right)  $$
QRAM accesses and $ \Ocal \left( n^{2.5}  \cdot \operatorname{polylog} \left( \kappa, \frac{1}{\epsilon} \right) \right) $ arithmetic operations.
\end{corollary}

\begin{proof}
The quantum gate complexity can be obtained upon repeating the analysis we used in the proof of Corollary \ref{corr:runtimeIF}, noting that for LO the dimension of the problem is $d = n$ instead of $d = n^2$. Then, recall that in the case of LO, all of the matrices are diagonal, and therefore no symmetrization is required. Hence, the only classical progression steps in this case are computing the primal and dual search directions from their nullspace representation, and subsequently updating the solution, and each of these steps require $\Ocal (n^2)$ arithmetic operations. 
\end{proof}

Since our IF-QIPM utilizes a hybrid quantum-classical scheme, one can trivially obtain a fully classical IF-IPM by replacing the quantum subroutine for solving the Newton linear system with an inexact classical linear systems algorithm. Note that the system we obtain in \eqref{e:IF-Newton} using the nullspace representation is not guaranteed to be positive semidefinite, so we cannot directly apply the Conjugate Gradient Method (CG) to solve it. For $M \in \R{d \times d}$ and $v \in \R{d}$ with $M \not\in \Scal^n_+$, one could first symmetrize the system $M u = v$ as
$$ M^{\top} M u = M^{\top} v,$$
and subsequently apply CG, in which case the cost of this step would be 
$$\Ocal \left(d^{\omega} + d^2 \sqrt{\kappa_M^2} \log (1/\xi) \right)= \Ocal \left(d^{\omega} + d^2 \kappa_M \log (1/\xi) \right),$$ 
because we need to carry out matrix-matrix multiplication, and the condition number of the resulting coefficient matrix is $\kappa_M^2$. Alternatively, we could employ one of the generalizations of CG such as the Generalized Minimum Residual Method (GMRES), however, not only is the coefficient matrix we employ indefinite, but it is nonsymmetric and dense: applying GMRES\footnote{We do not review CG or GMRES in detail, and instead refer the reader to Chapter 6 of \cite{saad2003iterative}.} to solve a dense, nonsymmetric and indefinite system of dimension $d$ to precision $\xi$ leads to a complexity of $\Ocal \left( d^3 \log(1/\xi) \right)$. In our case $d = n^2$, so classically performing this step will require time $\Ocal(n^{2\omega} + n^4 \kappa \log (1/\xi))$ if we symmetrize and use CG, or time $\Ocal \left( n^6 \log(1/\xi) \right)$ using GMRES.

This outcome is undesirable because it would lead to a slower algorithm than what already exists in the classical literature. Alternatively, one can use a polynomial approximation $q$, of $1/u$ on the interval $[1/\kappa_M^2,1]$. The solution would then be $q(M^2)Mv$, which is approximately $M^{-2}M v = M^{-1}v$. Such a polynomial $q$ with degree roughly $\kappa_M \log(1/\xi)$ exists (see, e.g., \cite[Section 6.11]{saad2003iterative}), and the approximate solution would be $q(M^2)Mv$, which is close to $M^{-1}v$ whenever the linear system $Mu=v$ has a solution. Crucially, $q(M^2)M v$ can be obtained using $2 \cdot \text{deg}(q)+1$ matrix-vector products, and does not require computing (or even writing down) the matrix $M^2$, thus avoiding the $n^{2\omega}$ present term in the complexity associated with symmetrization and subsequently using CG. Rather, in our case the complexity becomes
$$\Ocal(d^2 \sqrt{  \kappa_M  } \log (1/\xi)) = \Ocal \left( \left( n^2 \right)^2 \sqrt{  \kappa^2  } \log (1/\xi) \right)  = \Ocal \left(n^4  \kappa   \log (1/\xi) \right).$$
For an explicit discussion of this method, we refer the reader to Section 16.5 of the monograph \cite{vishnoi2013lx}.

This argument is the basis for the next result, which bounds the running time of the classical IF-IPM for SDO.
\begin{corollary}
A classical implementation of Algorithm \ref{alg:IF-QIPM} obtains an $\epsilon$-optimal solution $(X^*, y^*, S^*)$ to the primal-dual SDO pair \eqref{e:SDO}-\eqref{e:SDD} with overall complexity
\begin{equation*}
 \Ocal \left( n^{4.5} \kappa \cdot \textup{polylog}\left(\kappa, \frac{1}{\epsilon} \right) \right).
\end{equation*}
\end{corollary}

\begin{proof}
Simply note that applying \cite[Theorem 16.6]{vishnoi2013lx} with $d = n^2$ one has $$T_{NT}^{IF} = \Ocal \left( n^4 \kappa \cdot \textup{polylog}\left(\kappa, \frac{1}{\epsilon} \right) \right),$$ and we still have $T^{IF}_{progress} = \Ocal(n^4)$. Just as in the case of the IF-QIPM, the linear system solver we use for our IF-IPM takes the condition number of the Newton linear system as input. Using the verification scheme described in the proof of Proposition \ref{prop:costPerIterIF} implies that we might need to repeat these steps $\log (\kappa)$ times to ensure the solution we obtain is accepted. The stated result then follows upon accounting for the iteration bound provided by Corollary \ref{c:Converge}. 
\end{proof}

\begin{corollary}
A classical implementation of Algorithm \ref{alg:IF-QIPM} obtains an $\epsilon$-optimal solution $(X^*, y^*, S^*)= (\diag (x^*), y^*, \diag (s^*))$ to a Linear Optimization problem with overall complexity
\begin{equation*}
  \Ocal \left( n^{2.5} \kappa \cdot \textup{polylog}\left(\kappa, \frac{1}{\epsilon} \right) \right).
\end{equation*}
\end{corollary}

\subsection{Running time of the II-QIPM}
Similar to our analysis for the IF-QIPM, we begin by formally stating the per-iteration cost of the II-QIPM.

\begin{proposition}\label{prop:costPerIterII}
Every iteration $k$ of the II-QIPM in Algorithm \ref{alg:QII-IPM} has a quantum gate complexity of $$\Ocal \left(  n^{3.5} \kappa \kappa_{\cal A} \rho
  \left(  \|{\cal A}\|_F + \|S\|_F \right)   \frac{1}{\epsilon} \cdot \operatorname{polylog} \left(n, \kappa, \frac{1}{\epsilon} \right) \right)$$ and requires $\Ocal(n^{\omega} \log (\kappa))$ classical arithmetic operations. 
\end{proposition}

\begin{proof}
In each iteration of the II-QIPM, we prepare and solve the HKM
linear system; and let $T_{LS}^{II}$ denote the time required to do so. As in the IF-QIPM, we the utilize vector state tomography algorithm from \cite{van2021tomography} to obtain a classical description of the solution of the Newton linear system \eqref{e:inexactNewt}; and similar to the proof of Proposition \ref{prop:costPerIterIF}, denote the corresponding running time by
$T_{TO}(T_{LS}^{II} , \xi_k)$ for precision $\xi_k$. Applying Proposition~\ref{prop:hkmls}, we have $T_{LS}^{II} = \widetilde{\Ocal}_{n, \kappa} (\kappa (\|{\cal A}\|_F + \|S\|_F))$, and hence, we have: 
\begin{align*}
    T_{NT}^{II} = T_{TO}(T_{LS}^{II}, \xi_k) = \widetilde{\Ocal}_{n, \kappa, \frac{1}{\xi_k}} \left(\frac{n^2}{\xi_k} T_{LS}^{II} \right) &= \widetilde{\Ocal}_{n, \kappa, \frac{1}{\epsilon}} \left(\frac{n^2}{\epsilon}  \varrho \left( \kappa (\|{\cal A}\|_F + \rho n^{1.5} \right) \right).
\end{align*} 
From here, we point out that for II-(Q)IPMs, $\varrho = \Ocal (\kappa_{\cal A} \rho
n^{1.5})$; this is a consequence of the known bounds $\|\Delta X\|,
\|\Delta S\| = \Ocal(\rho n)$ (see, e.g., \cite{zhou2004polynomiality}),
the relationship between spectral norm and Frobenius norm, and the
definition of the neighborhood of the central path
\eqref{e:neighborhood}. We therefore have
$$
    T_{NT}^{II}  = \widetilde{\Ocal}_{n, \kappa, \frac{1}{\epsilon}} \left(\frac{n^2}{\epsilon}  \kappa_{\cal A} \rho
n^{1.5} \left( \kappa (\|{\cal A}\|_F + \|S\|_F \right)\right) = \widetilde{\Ocal}_{n, \kappa, \frac{1}{\epsilon}} \left( n^{3.5} \kappa \kappa_{\cal A} \rho
  \left(  \|{\cal A}\|_F + \rho n^{1.5} \right)  \frac{1}{\epsilon} \right),
$$
as $\| S \|_F = \Ocal (\rho n^{1.5})$ by \eqref{e:II_norm_bound}.

We must also update the solution to the SDOP, and compute our scaling matrix $P$ and its inverse. Updating the solution is again standard matrix addition involving $n \times n$ matrices, and can be accomplished using $\Ocal(n^2)$ arithmetic operations, and the computation of $P = S^{1/2}$ and its inverse can be performed using $\Ocal(n^{\omega})$ arithmetic operations, as all of the involved matrices are of size $n \times n$. 

Just as in the case of the IF-QIPM, the QLSA we use for our II-QIPM takes the condition number of the Newton linear system as input. Using the scheme described in the proof of Proposition \ref{prop:costPerIterIF} implies that we might need to repeat these steps (both quantum and classical) $\log (\kappa)$ times to ensure the solution we obtain is accepted. Thus, each iteration of Algorithm \ref{alg:QII-IPM} has a quantum gate complexity of $$\Ocal \left( n^{3.5} \kappa \kappa_{\cal A} \rho
  \left(  \|{\cal A}\|_F + \rho n^{1.5} \right)  \frac{1}{\epsilon} \cdot \operatorname{polylog} \left(n, \kappa, \frac{1}{\epsilon} \right) \right)$$ and requires $\Ocal(n^{\omega} \log(\kappa))$ classical arithmetic operations. The proof is complete. 
\end{proof}

Combining the preceding result with the iteration bound in Theorem~\ref{t:t1zhou} immediately gives the overall running time of the II-QIPM, which is formalized next. 
\begin{corollary}
  A quantum implementation of Algorithm \ref{alg:QII-IPM} with access to QRAM converges to an $\epsilon$-optimal solution to the primal-dual SDO pair \eqref{e:SDO}-\eqref{e:SDD} in time
$$\Ocal \left( n^{5.5} \kappa \kappa_{\cal A} \rho
  \left(  \|{\cal A}\|_F + \rho n^{1.5} \right) \frac{1}{\epsilon} \cdot \operatorname{polylog} \left(n, \kappa, \frac{1}{\epsilon} \right) \right).$$
\end{corollary}
 
\begin{proof}
We can begin by expanding the expression for the per-iteration cost of the algorithm provided in Proposition \ref{prop:costPerIterII}. Then, noting that by Theorem~\ref{t:t1zhou}, the II-QIPM requires $ \Ocal(n^2 \ln 1/\epsilon)$ iterations, we can express the total running time of the algorithm as
$$
  \sum_{k=1}^{\chi n^2 \log (1/\epsilon)}   n^{3.5} \kappa \kappa_{\cal A} \rho 
  \left(  \|{\cal A}\|_F + \rho n^{1.5} \right)   \frac{1}{\epsilon} =\Ocal \left(  n^{5.5} \kappa \kappa_{\cal A} \rho
  \left(  \|{\cal A}\|_F + \rho n^{1.5} \right)   \frac{1}{\epsilon} \cdot \operatorname{polylog} \left(n, \kappa, \frac{1}{\epsilon} \right) \right),
$$
and the proof is complete. 
\end{proof}

\subsection{Comparison with existing algorithms for SDO}\label{ss:compare}
Before proceeding further, we remind the reader of a number of key differences between (Q)IPMs and (Q)MWU algorithms. The output of (Q)IPMs is a classical primal-dual pair $(X,y,S)$ such that 
        \begin{align*} 
        A^{(i)} \bullet X &= b_i \quad \forall i \in [m], \quad X \succ 0 \nonumber \\
        \sum_{i \in [m]} y_i A^{(i)} + S &= C   \quad  S \succeq 0 
        \end{align*}
         with 
         $$\frac{X \bullet S}{n} \leq \epsilon.$$
         Conversely, primal-dual (Q)MWUs output a dual solution $y \in \R{m}$ such that setting 
    $$ X = R \cdot \frac{\exp\left( - \sum_{i = 1}^m y_i A^{(i)} \right)}{\trace \left( \exp\left( - \sum_{i = 1}^m y_i A^{(i)} \right) \right)}$$
    the primal-dual pair $(X,y)$ satisfies
    \begin{align*} 
    \left| A^{(i)} \bullet X - b_i \right| &\leq  \epsilon_{\text{abs}} \quad \forall i \in [m], \quad X \succ 0 \nonumber \\
    S = C - \sum_{i \in [m]} y_i A^{(i)}   &\succeq - \epsilon_{\text{abs}} I  \
    \end{align*}
    and the objective value attained by this solution is $\text{OPT} \in [\varsigma - {\cal O}(\epsilon_{\text{abs}}), \varsigma + {\cal O}(\epsilon_{\text{abs}})]$, where $\varsigma$ is a bound on the optimal objective value determined using binary search.
    Observe that this is another distinction between (Q)IPMs and (Q)MWUs; the output of our QIPMs always satisfies primal and dual feasibility exactly, whereas the solution obtained by (Q)MWU methods satisfies primal and dual feasibility approximately.
    
    The QMMWUs (specifically, all papers by Brandao et al. \cite{brandao2017exponential, brandao2019faster, brandao2017quantum}, as well as the papers by van Apeldoorn et al.~\cite{van2018improvements, van2020quantum} also assume that the SDO problem data is normalized: 
        $$\| A^{(1)} \|, \dots, \|A^{(m)} \|, \| C \| \leq 1.$$
    Our QIPMs do not require this assumption. For the QIPMs presented in this work, only the Newton linear system must be normalized in each iteration (so we can apply a QLSA to solve it), and this normalization is accounted for in our accuracy requirements and the subsequent running time analysis. (For classical MMWUs, the choice of the normalization can vary depending on the papers; the analysis of the classical algorithm in \cite{van2020quantum} assumes the above normalization.)

    Assuming the data is normalized with respect to the operator norm means that for more general problems (i.e., without normalized data) the error scales as $\theta \epsilon_{\text{abs}}$, where 
    $$ \theta = \max \left\{ \| A^{(1)} \|, \dots, \|A^{(m)} \|, \| C \| \right\}.$$
    Thus, for MMWUs one would lose a factor $\theta$ in the precision. Note that being able to compute $\theta$ to normalize the data as $\theta^{-1} A^{(1)}, \dots, \theta^{-1} A^{(m)},  \theta^{-1} C$ could require classical access to these matrices, and an additional $\Ocal(mns)$ cost in the running time to compute $\theta$.

Table \ref{tab:runTnsq} presents the running time of our algorithms alongside well-known algorithms for solving SDOPs. For ease of comparison, we assume that $m = \Ocal(n^2)$ and that the problem is fully dense (i.e., $s = n$). It can be readily seen that the II-QIPM is worse in every parameter when compared to the other algorithms. Conversely, the IF-QIPM obtains a speedup with respect to the dimension $n$ over each of the listed classical solution methodologies. Yet, the IF-QIPM exhibits a quadratic dependence on the condition number bound $\kappa$, and linear dependence on the inverse precision, implying an asymptotically worse running time overall. 

\begin{table}[H]
    \centering
    \begin{tabular}{lll}
    \toprule 
    \hline
    \textbf{References}     & \textbf{Method}  & \textbf{Runtime} \\
    \hline  
    \cite{monteiro1998polynomial, nesterov1997self, nesterov1998primal}     &  IPM   & $\widetilde{\Ocal}_{n, \frac{1}{\epsilon}} \left( n^{6.5} \right)$   \\
    \cite{jiang2020faster}     &  IPM   & $\widetilde{\Ocal}_{\frac{1}{\epsilon}} \left( n^{5.246} \right)$  \\
    \cite{jiang2020improved, lee2015faster} & CPM    & $\widetilde{\Ocal}_{n, R, \frac{1}{\epsilon}} \left(  n^6 \right)$ \\
    \cite{van2018improvements} & QMWU   &$ \widetilde{\Ocal}_{n, \frac{1}{\epsilon_{\text{abs}}}} \left( \left(n^2 + n^{1.5}
\frac{Rr}{\epsilon_{\text{abs}}}\right) \left(\frac{Rr}{\epsilon_{\text{abs}}}\right)^4
\right)$     \\
    This work (Quantum) & IF-QIPM  & $\widetilde{\Ocal}_{n, \kappa, \frac{1}{\epsilon}}  \left(\sqrt{n} \left(n^3 \frac{\kappa^2}{\epsilon} + n^4 \right)  \right)$  \\
    This work (Classical) & IF-IPM  & $\widetilde{\Ocal}_{\kappa, \frac{1}{\epsilon}}  \left(n^{4.5} \kappa \right)$  \\
This work (Quantum)  & II-QIPM  &$\widetilde{\Ocal}_{n, \kappa, \frac{1}{\epsilon}}\left( n^{5.5} \kappa \kappa_{\cal A} \rho
  \left(  \|{\cal A}\|_F + \rho n^{1.5} \right) \frac{1}{\epsilon} \right) $   \\
    \bottomrule
    \end{tabular}
    \caption{Total running times for classical and quantum algorithms to solve \eqref{e:SDO}-\eqref{e:SDD} with $m = \Ocal (n^2)$.}
    \label{tab:runTnsq}
\end{table}

When compared to the QMWU algorithm, it can be observed that the QIPM achieves a favorable dependence on the error. Although the QMWU is faster with respect to dimension, the QMWU has polynomial dependence on the trace bound $R$, and for certain applications this may negate any speedup, such as the SDO approximation of Quadratic Unconstrained Binary Optimization (QUBO) problems, i.e., problems of the form 
\begin{equation}\label{e:QUBO}
    \max_{x \in \{-1, 1\}^n} x^{\top} C x.
\end{equation}
QUBOs play an important role in quantum information sciences, and the resulting SDO approximation is given by
\begin{equation}\label{e:QUBO-SDO}
    \max_{X \in \Scal^n} \left\{ C \bullet X : \operatorname{diag} (X) = e, X \succeq 0 \right\},
\end{equation}
where $e$ is the all ones vector of dimension $n$. Hence, any optimal solution of \eqref{e:QUBO-SDO} satisfies $\trace(X) = n$, suggesting that for these instances the worst case running time of the QMWU is
$$ \widetilde{\Ocal}_{n, \frac{1}{\epsilon_{\text{abs}}}}\left( n^{6.5} \left(\frac{r}{\epsilon_{\text{abs}}}\right)^5
\right),$$
as $R = n$ and $m = n$ for \eqref{e:QUBO-SDO}. Thus, the IF-QIPM could outperform the QMWU for this class of SDOPs, provided that $\kappa = \Ocal \left( n r^{2.5} \left(\frac{1}{\epsilon}\right)^2 \right)$. 

For SDOPs with diagonal constraints such as \eqref{e:QUBO-SDO}, the Arora-Kale method can be specialized (see, e.g., \cite{lee2020}) to obtain an $\tilde{\Ocal}(ns \epsilon^{-3.5})$ algorithm. Note that to achieve this running time, the algorithm avoids explicitly computing or reporting the primal solution $X$, and instead makes use of techniques to estimate its diagonal entries and trace inner products of the form $A \bullet X$. As an alternative, these methods report a ``gradient" $G \in \Scal^n$ such that $X = W \exp(G) W$ for a diagonal matrix $W$. While the algorithm of Lee and Padmanabhan \cite{lee2020} offers speedups with respect to the dimension, its running time is exponentially slower than classical IPMs with respect to the inverse precision. 

To obtain an additive $\epsilon$-convergence using the algorithm found in \cite{lee2020}, the error parameter would have to be set to  
    $$\tilde{\epsilon} = \frac{\epsilon}{\| C \|_{\ell_1}}.$$ 
    When $C$ is dense and $\| C \|_{\ell_1} = \Ocal(n)$, the running time of their algorithm therefore scales as $$\Ocal(n^2 (\tilde{\epsilon})^{-3.5})) = \Ocal(n^2 (n/\epsilon)^{3.5})) = \Ocal (n^{5.5} \epsilon^{-{3.5}}).$$
    For our QIPM to reach duality gap $\epsilon$ would require complexity 
    $$\Ocal\left(\sqrt{n} \left(n^3 \kappa^2 \frac{1}{\epsilon/n}  + n^4 \right) \textup{polylog} (n, \epsilon^{-1}, \kappa) \right) = \widetilde{\Ocal}_{n,\kappa, \frac{1}{\epsilon}} \left(n^{4.5}  \frac{\kappa^2}{\epsilon}  \right).$$
    It can then be seen that in this case, our QIPM scales better than the classical MWU in $n$ and $\epsilon^{-1}$ in this case, but depends linearly on $\kappa$. However, one could solve the QUBO problem to $\epsilon$-optimality using the IPM from Jiang et al.~\cite{jiang2020faster} in time $\Ocal (n^{\omega + 0.5} \log(n/\epsilon))$, or even our classical IF-IPM with overall running time $\widetilde{\Ocal}_{\kappa, \frac{n}{\epsilon}}  \left(n^{4.5} \kappa  \right)$.

One potential use case of our quantum and classical IF-IPMs would be as a limited precision oracle in an iterative refinement method. Iterative refinement (IR) methods are used to solve optimization problems to extended precision by using (possibly fixed) low-precision calls to an optimization subroutine (such as an IPM) to solve a sequence of simpler problems \cite{gleixner2020linear, gleixner2012improving, gleixner2016iterative}. In this setting, low-precision solutions are completely acceptable, and since both the classical and quantum IF-IPMs we propose classically report the solution, the IF-QIPM would seem to be the preferable choice of subroutine.

However, the behavior of $\kappa$ as we traverse the central path towards optimality poses a significant hurdle to the performance of QIPMs. This is highlighted in detail next. 

\subsection{The relationship between $\kappa$ and optimality in IPMs}
\label{s:comparingrt}
Noting the presence of the condition number $\kappa$ in our running time results for the QIPMs, it is of crucial importance to determine whether or not bounds on $\kappa$ can be established. This is of particular importance as the coefficient matrices $M_{\text{NT}}$ change at each iteration of the QIPMs. It has been well studied that without any numerical safeguards, in theory the condition number $\kappa$ will go to $\infty$ as we approach the optimal set to the optimization problem at hand. In what follows, we highlight how the condition number $\kappa$ is related to the optimality gap $\epsilon$, and also note that $\xi$, the precision of the quantum linear solver has to be also in the order of $\epsilon$ or less. 

To keep our discussion simple, let us first consider the case of LO. Following Roos, Terlaky and Vial \cite{roos2005interior} for the LO case, let 
\begin{align*}
    B &= \{ i :x_i > 0~\text{for some optimal solutions} \}, \\
    N &= \{ i : s_i > 0 ~\text{for some optimal solutions}\},
\end{align*}
where $x$ and $s$ denote the diagonals of $X$ and $S$. 
Then, $(B, N)$ forms the so called \textit{optimal partition} of the LO case of \eqref{e:SDO}. It is known that 
$$B \cup N = \{1, \dots, n \}.$$
Further, letting $ \Scal \Pcal^*$ denote the primal-dual optimal set,
we define the quantities 
\begin{align*}
     \sigma_{SP}^x &:= \min_{i \in B} \max_{x \in \Scal \Pcal^*} \{x_i\},\\
    \sigma_{SP}^s &:= \min_{i \in N} \max_{x \in \Scal \Pcal^*} \{s_i \}, 
\end{align*}
where $\sigma_{SP}^x  = \infty$ if $B = \emptyset$ and $\sigma_{SP}^s  = \infty$ if $N = \emptyset$. By the boundedness of the optimal set $\Scal \Pcal^*$, it follows that $\sigma_{SP}^x$ is finite if $B$ is nonempty, and similarly $\sigma_{SP}^s$ is finite if $N$ is nonempty. As Roos, Terlaky and Vial \cite{roos2005interior} note, by the definition of $B$ and $N$, at least one must be nonempty. Therefore, taking 
$$\sigma_{\text{SP}} = \min \{ \sigma_{SP}^x, \sigma_{SP}^s  \},$$
it follows that the condition number of the LO problem must be positive and finite.  As we follow the central path, for the case of LO we have $s_i x_i = \nu$ for all $i$. As we approach optimality we have
\begin{lemma}[Lemma I.43 in \cite{roos2005interior}]\label{lem:terlaky}
For any $\nu >0$, 
\begin{align*}
x_i (\nu) &\geq \frac{\sigma_{\text{SP}} }{n},~~~\forall i \in B,  &x_i (\nu) &\leq \frac{n \nu}{\sigma_{\text{SP}} },~\forall i \in N,\\
s_i (\nu) &\leq \frac{n \nu}{\sigma_{\text{SP}}},~\forall i \in B,  &s_i (\nu) &\geq \frac{\sigma_{\text{SP}}}{n},~~\forall i \in N.
\end{align*}
\end{lemma}
As in the Newton system of LO, at each iteration we need to calculate the coefficient matrix $A^{\top} XS^{-1}A$ of the Newton system in compact form (see e.g., \cite{roos2005interior}). The condition number of this matrix will be in the order of the largest element of the diagonal matrix $XS^{-1}$ divided by the smallest diagonal element. Thus, considering the bounds of Lemma \ref{lem:terlaky}, we have that this ratio is at least 
\begin{align*}
\left( \frac{\frac{\sigma_{\text{SP}}}{n}}{\frac{n \nu}{\sigma_{\text{SP}} }}  \right) \left/ \left( \frac{\frac{n \nu}{\sigma_{\text{SP}} }}{\frac{\sigma_{\text{SP}}}{n}}  \right)  \right. = \left(\frac{\sigma_{\text{SP}}^2}{n^2 \nu} \right)  \left/ \left(\frac{n^2 \nu} {\sigma_{\text{SP}}^2}\right) \right. = \frac{\sigma_{\text{SP}}^4}{n^4 \nu^2},
\end{align*}
which clearly goes to infinity as $\nu \to 0$, as the QIPM approaches the optimal set. It is also known that the condition number $\sigma$ of the LO problem can be as small as $2^L$ if the problem data is integer, see, e.g., \cite{roos2005interior}.
 
Now, let us consider, as customary in the classical setting \cite{roos2005interior}, that all of the data of the LO problem is integer. Then, the binary encoding of the LO problem requires at most
$$ L = \sum_i \sum_j (1 + \log_2 (|a_{ij}|+ 1)) + \sum_i (1 + \log_2 (|b_{i}| + 1))  +\sum_j (1 + \log_2 (|c_{j}| + 1))   $$
 bits. It is also known that $\delta \geq 2^{-L}$ and by \cite[Theorem I.53]{roos2005interior}, that if $\epsilon < \frac{\sigma_{\text{SP}}^5}{4 n^{1.5}}$ then an exact optimal solution can be identified by solving a linear system of equations. For simplicity, the bound given here is not the tightest possible bound, some quantities in the bound of \cite[Theorem I.53]{roos2005interior} are estimated by $\sigma_{\text{SP}}.$ 
 
Another way to look at it, is to estimate the total running time to reach a near-optimal primal-dual feasible solution pair, with optimality gap less than $\epsilon$. To reach this precision, we need to have 
 $$n \nu < \epsilon < \frac{\sigma_{\text{SP}}^5}{4 n^{1.5}},$$
i.e. $\nu = \Ocal \left( \frac{\sigma_{\text{SP}}^5}{4 n^{1.5}} \right)$. Thus we can identify an exact optimal solution of the LO problem if
$$ \nu = \Ocal \left( \frac{2^{-5L}}{n^{2.5}} \right).$$
Then, we have
$$ \kappa \geq \frac{\sigma_{\text{SP}}^4}{n^4 \nu^2} \geq \frac{\sigma_{\text{SP}}^4}{n^4 \frac{\sigma_{\text{SP}}^{10}}{n^5}} = n \sigma_{\text{SP}}^{-6} = \Omega (n 2^{6L}).$$
We also need 
$$\xi < \epsilon < c_0,$$
where $c_0 = \Ocal \left( \frac{2^{-5L}}{n^{2.5}}  \right).$

By Corollary \ref{corr:rtLO}, when applied to LOPs, the IF-QIPM has a worst case overall running time of 
$$ \Ocal \left(\sqrt{n} \left(n^{1.5} \frac{\kappa^2}{\epsilon} + n^2 \right)   \operatorname{polylog} \left(n, \kappa, \frac{1}{\epsilon} \right) \right). $$
Substituting all of these in the total running time bound of our IF-QIPM yields
$$ \widetilde{\Ocal}_{n, \kappa, \frac{1}{\epsilon}} \left(  n^{1.5} \frac{n^2 2^{12L}}{\frac{2^{-5L}}{n^{2.5}}} + n^{2.5} \right) = \widetilde{\Ocal}_{n, \kappa} \left(   n^{6.5} 2^{17L} L  \right)  $$
as the time complexity of our QIPM to identify an exact solution, compared to the $\Ocal(n^3 L)$ arithmetic operations needed for the classical IPM to do the same. Carrying out similar calculations for the II-QIPM yields a worst case running time bound of 
$$\widetilde{\Ocal}_{n, \kappa, \frac{1}{\epsilon}}\left(  n^{4.5} \kappa \kappa_{\cal A} \rho
  \left(  \|{\cal A}\|_F + \rho n^{1.5} \right) \right) = \widetilde{\Ocal}_{n, \kappa}\left(  n^{5.5}   \kappa_{\cal A} \rho
  \left(  \|{\cal A}\|_F + \rho n^{1.5} \right) 2^{6L} L\right) $$
when applied to LOPs. 

As LO is a special case of SDO, the general bound for the condition number of the Newton linear systems \eqref{e:inexactNewt} and \eqref{e:IF-Newton} are bounded below by the condition number bound of LO problems. For the condition number of SDO problems, we refer to Mohammad-Nezhad and Terlaky \cite{aliMohamad}. Let 
\begin{align*}
    \Bcal &= \Rcal (X^*), \\
    \Ncal &= \Rcal (S^*),
\end{align*}
where $\Rcal (\cdot)$ denotes the range space and $(X^*, y^*, S^*)$ denotes a maximally complementary solution. Define $n_{\Bcal} = \text{dim}(\Bcal)$ and $n_{\Ncal} = \text{dim}(\Ncal)$. Then, letting $Q = [Q_{\Bcal}, Q_{\Tcal}, Q_{\Ncal}]$ be orthonormal bases for $\Bcal$, $\Ncal$ and $\Tcal$, respectively, it follows from \cite{aliMohamad} that for every primal-dual optimal solution $(X, y, S) \in \Pcal^* \times \Dcal^*$, one can write $X$ and $S$ as:
\begin{align*}
    X &= Q_{\Bcal} U_X Q^{\top}_{\Bcal} \\
    S &= Q_{\Ncal} U_S Q^{\top}_{\Ncal},
\end{align*}
where $U_X$ and $U_S$ are symmetric $n_{\Bcal} \times n_{\Bcal}$ and $n_{\Ncal} \times n_{\Ncal}$ matrices, respectively. Then, similar to the result for linear optimization, defining
\begin{subequations}
\begin{align}
    \sigma_{\Bcal} &= \max_{(X \in \Pcal^*} \lambda_{\min} (Q^{\top}_{\Bcal} X Q_{\Bcal}) = \max_{\bar{Q_{\Bcal}} \in \Gamma_{\Bcal}} \max_{X \in \Pcal^*} \lambda_{\min} (\bar{Q}^{\top}_{\Bcal} X \bar{Q}_{\Bcal}) \\
     \sigma_{\Ncal} &= \max_{(y,S) \in \Dcal^*} \lambda_{\min} (Q^{\top}_{\Ncal} S Q_{\Ncal}) = \max_{\bar{Q_{\Ncal}} \in \Gamma_{\Ncal}} \max_{(y,S) \in \Dcal^*} \lambda_{\min} (\bar{Q}^{\top}_{\Ncal} S \bar{Q}_{\Ncal}) \\
     \sigma &= \min \{\sigma_{\Bcal}, \sigma_{\Ncal} \}
\end{align}
\end{subequations}
where $\Gamma_{\Bcal}$ and $\Gamma_{\Ncal}$ denote the sets of all orthonormal bases for $\Bcal$ and $\Ncal$, respectively. Final, we arrive at the following results from \cite{aliMohamad}. 

\begin{lemma}[Lemma 2 in \cite{aliMohamad}]
  The condition number $\sigma$ is positive. 
\end{lemma}

\begin{lemma}[Lemma 8 in \cite{aliMohamad}]
  Let the SDO problems $(P)$ and $(D)$ be given by integer data, $L$ denote the binary input length of the largest absolute value of the entries in $b$, $C$ and $A^{(i)}$ for $i = 1, \dots, m$ and $\| \cdot \|_F$ be the Frobenius norm. Then for the condition number $\sigma$, we have
  $$ \sigma \geq \min \left\{ \frac{1}{r_{\Pcal^*} \sum_{i = 1}^m \| A^{(i)} \|_F}, \frac{1}{r_{\Dcal^*}} \right\}. $$
\end{lemma}
\noindent This implies that $\sigma$ is positive for SDO as well, therefore $\kappa \to \infty$.

\section{Conclusion}\label{s:con}
In this work we present two provably convergent quantum interior point methods for semidefinite optimization problems. They are building on recent advances in quantum linear system solvers.
The quantization of classical interior point methods is the subject of several recent papers in the literature. We compare the theoretical performance of classical and quantum interior point methods with respect to various input parameters, concluding that our Inexact-Feasible QIPM provides a speedup in terms of the size of the problem. Yet, due to the dependence on the condition number of the QLSA, as well as a $1/\epsilon$ dependence picked up from repeated tomography, it is not faster than the classical small neighborhood IPM in the worst case. The dependence on the condition number in particular presents a significant obstacle toward a theoretical quantum speedup. In practice, the QIPM may still have an advantage if low-precision solutions are acceptable. 
 
\section*{Acknowledgements}
The authors would like to thank the anonymous referees of Quantum for very helpful comments that improved the presentation, and for pointing us to \cite{vishnoi2013lx} to improve the running time of our classical IF-IPM.

This project has been carried out thanks to funding by the Defense 
Advanced Research Projects Agency (DARPA),
ONISQ grant W911NF2010022, titled The Quantum Computing Revolution and 
Optimization: Challenges and Opportunities. This research also used resources of the Oak Ridge Leadership Computing Facility, which is a DOE Office of Science User Facility supported under Contract DE-AC05-00OR22725. G.~Nannicini is supported by ONR award \# N000142312585. 

\appendix
\section{Block-encoding the Inexact-Infeasible Newton Systems}\label{app:II_newton}

\subsection{Implementing the Inexact-Infeasible Newton linear system for Members of the Monteiro-Zhang Family of Search Directions}

For the version of our QIPM that uses the NT direction, we factorize
the corresponding linear system matrix and construct block-encodings
for its factors. Note that in principle we could construct a $(2,
\widetilde{\Ocal}(1), \xi)$-block-encoding for $W = S^{-1/2}(S^{1/2} X
S^{1/2})^{1/2} S^{-1/2}$ in time $\widetilde{\Ocal}(\|S\|^4_F \kappa_S^5
\kappa_X)$, starting from the individual factors $S, X$ and with
repeated applications of Propositions~\ref{prop:ampproduct} and \cite[Lemma 9 and 10]{chakraborty2018power}.  However,
due to the heavy dependence on the condition number, we choose to
classically compute all the factors (including matrix powers), and
store them in QRAM. This is more efficient: all the involved matrices
are $n \times n$, so classical computation of matrix powers has
negligible cost compared to the solution of the Newton linear system,
which is of size $\Ocal(n^2 \times n^2)$, and we can avoid increasing the
dependence on the condition number. The following two propositions
establish how one can construct block-encodings of these factors, and
bound the corresponding running time.

\begin{proposition}
  \label{prop:ntbe1II}
  Let 
  $$ 
        M_1 =   \begin{pmatrix}
         0 & \tran{\cal A}_s & \Ical \\
         {\cal A}_s & 0 & 0 \\
         0 & 0 & 0
         \end{pmatrix}.
     $$
 Let ${\cal A}_s$ be stored in QRAM. Then, a $(\|M_1\|_F,
  \Ocal(\log n), \xi_{M_1})$-block-encoding
  of $M_1$ can be constructed in time $\widetilde{\Ocal}_{ \frac{n}{\xi_{M_1}}}(1)$.
\end{proposition}

\begin{proof}
The result follows directly from Proposition \ref{prop:qramblockenc}.
\end{proof}

\begin{proposition}
  \label{prop:ntbe2II}
  Let $P = W^{-1/2}$, where $W$ is the Nesterov-Todd scaling matrix and 
  $$M_2 = \begin{pmatrix}
         0 & 0 &0 \\
        0 & 0 & 0 \\
         P \otimes_s P^{-\top} S & 0 & 0
         \end{pmatrix}.$$
    Let $P$  and $P^{-\top} S$, be stored in QRAM. Then, a $(\|M_2\|_F,
  \Ocal(\log n), \xi_{M_2})$-block-encoding
  of $M_2$ can be constructed in time $\widetilde{\Ocal}_{ \frac{n}{\xi_{M_2}}}(1)$.
\end{proposition}

\begin{proof}
With $P$ and $P^{-\top} S$ stored in QRAM, it follows from Proposition \ref{prop:blockencSymKronecker} that we can construct a $( \| P \otimes_s P^{-\top} S \|_F , \Ocal (\log n), \xi_{M_2})$-block-encoding of $P \otimes_s P^{-\top}$ in time $\widetilde{\Ocal}_{\frac{n}{\xi_{M_2}}} (1)$. $M_2$ can then be obtained using, e.g., Lemma~\ref{lem:sparseBE} and product of block-encoded matrices to place $P \otimes_s P^{-\top}$ in the right position.
\end{proof}

\begin{proposition}
  \label{prop:ntbe3II}
  Let $P = W^{-1/2}$, where $W$ is the Nesterov-Todd scaling matrix and
  $$M_3 = \begin{pmatrix}
         0 & 0 &0 \\
        0 & 0 & 0 \\
          0 & 0 & PX\otimes_s P^{-\top}  
         \end{pmatrix}.$$
    If $PX$, and
  $P^{-\top}$ are stored in QRAM. Then, a $(\|M_3\|_F,
  \Ocal(\log n), \xi_{M_3})$-block-encoding
  of $M_3$ can be constructed in time $\widetilde{\Ocal}_{ \frac{n}{\xi_{M_3}}}(1)$.
\end{proposition}

\begin{proof}
The proof follows the same argument as in the proof of Proposition~\ref{prop:ntbe2II}.
\end{proof}

\begin{proposition}\label{prop:NTcompactII}
  Let $P = W^{-1/2}$, where $W$ is the Nesterov-Todd scaling matrix. The Nesterov-Todd linear system matrix \eqref{e:NT2} can be written compactly as:
  \begin{equation}
    \label{eq:ntcompact}
    M_{\text{NT}} = 
    \begin{pmatrix}
         0 & \tran{\cal A}_s & \Ical \\
         {\cal A}_s & 0 & 0 \\
         P \otimes_s P^{-\top} S & 0 & PX\otimes_s P^{-\top}
         \end{pmatrix},
  \end{equation}
  If  ${\cal A}_s$, $P$, $P^{-\top}$, $PX$, and
  $P^{-\top}S$ are stored in QRAM, then a $(\|M_{\text{NT}}\|_F, \Ocal(\log n),
  \xi/(\kappa^2 \log^2 \frac{\kappa}{\xi}))$-block-encoding
  of \eqref{eq:ntcompact}, can be constructed in time $\widetilde{\Ocal}_{n, \kappa, \frac{1}{\xi}}(1)$.
\end{proposition}
\begin{proof}
  Carrying out the calculations $M_1 + M_2 + M_3$ shows that \eqref{eq:ntcompact}
  corresponds to the Nesterov-Todd linear system.  
  We construct the three following block-encodings:
$$
    M_1 = \begin{pmatrix}
         0 & \tran{\cal A}_s & \Ical \\
         {\cal A}_s & 0 & 0 \\
         0 & 0 & 0
         \end{pmatrix} \quad
    M_2 = 
         \begin{pmatrix}
         0 & 0 &0 \\
        0 & 0 & 0 \\
         P \otimes_s P^{-\top} S & 0 & 0
         \end{pmatrix} \quad
    M_3 = \begin{pmatrix}
         0 & 0 &0 \\
        0 & 0 & 0 \\
          0 & 0 & PX\otimes_s P^{-\top}  
         \end{pmatrix},
$$
  using Propositions~\ref{prop:ntbe1II}-\ref{prop:ntbe3II}. We choose the
  precision of this step so that we obtain $(\|M_1\|_F, \Ocal(\log n),
  \xi/(\|M_1\|_F \kappa^2 \log^2 \frac{\kappa}{\xi}))$,
  $(\|M_2\|_F, \Ocal(\log n), \xi/(\|M_2\|_F \kappa^2 \log^2
  \frac{\kappa}{\xi}))$ and $(\|M_3\|_F, \Ocal(\log n^2),
  \xi/(\|M_3\|_F \kappa^2 \log^2 \frac{\kappa}{\xi}))$-block-encodings, respectively, where $\kappa$
  refers to the condition number of \eqref{eq:ntcompact}, here and
  below. We add these two block-encodings together using
  Proposition~\ref{prop:lincombblock}, obtaining a $(\max\{\|M_1\|_F,
  \|M_2\|_F, \|M_3 \|_F\}, \Ocal(\log n), \xi/(\kappa^2 \log^2
  \frac{\kappa}{\xi}))$-block-encoding of \eqref{eq:ntcompact}, in
  time $\widetilde{\Ocal}_{n, \kappa, \frac{1}{\xi}} (1)$. Since $\max\{\|M_1\|_F, \|M_2\|_F, \|M_3 \|_F\} =
  \Ocal(\|M_{\text{NT}}\|_F)$, we obtain the claimed result.
\end{proof}
The construction described in Proposition~\ref{prop:NTcompactII} is only one
of the possible ways to obtain a block-encoding of the Newton linear
system corresponding to the Nesterov-Todd direction. We choose this
specific decomposition because among the ones that we analyzed, it
minimizes the amount of classical work while keeping the quantum gate complexity unchanged. We can then use this factorization to solve
the Newton system.
\begin{theorem}\label{t:NTtII}
  Let $P = W^{-1/2}$, where $W$ is the Nesterov-Todd scaling matrix. There is a quantum algorithm that given $$\ket{-\eta_1 R_{s_0}^d \circ -\eta_1 R_{s_0}^p  \circ \svec ( ( 1 - \eta_2) \tau \mu_0
    I - H_P(XS))}$$ and access to QRAM data structures encoding
  ${\cal A}_s$, $P$, $P^{-\top}$, $PX$, and
  $P^{-\top}S$, outputs a state $\xi$-close to
  $\ket{\Delta X \circ \Delta y \circ \Delta S}$ in time:
  \begin{equation*}
    \widetilde{\Ocal}_{n, \kappa, \frac{1}{\xi}} \left(\kappa \max\{\|M_1\|_F, \|M_2\|_F, \| M_3 \|_F\} \right),
  \end{equation*}
  using the NT direction. We can also output an estimate of $\|\Delta
  X \circ \Delta y \circ \Delta S\|$ with relative error $\delta$ by
  increasing the running time by a factor $\frac{1}{\delta}$.
\end{theorem}
\begin{proof}
  This is a direct consequence of Proposition~\ref{prop:NTcompactII} and
  Theorem~\ref{thm:qlsa}.
\end{proof}

\subsubsection{Implementing the HKM linear system}
We now analyze the version of our II-QIPM that uses the HKM
direction. As before, we could implement a $(2\kappa_S^{1/2},
\widetilde{\Ocal}(1), \xi_2)$-block-encoding for $S^{1/2} X S^{1/2}$ in
time $\widetilde{\Ocal}(\|S\|_F \kappa_S)$, using
Proposition~\ref{prop:qramblockenc} and \cite[Lemma 9 and 10]{chakraborty2018power}, but it is more efficient to perform these
computations classically --- see the discussion in the previous
subsection.

\begin{proposition}\label{prop:HKMcompactII}
  The HKM linear system matrix can be written compactly as:
  \begin{equation}
    \label{eq:hkmcompact}
    M_{\text{HKM}} = 
    \begin{pmatrix}
         0 & \tran{\cal A}_s & \Ical \\
         {\cal A}_s & 0 & 0 \\
         S^{1/2} \otimes_s S^{1/2} & 0 & S^{1/2} X \otimes_s S^{-1/2}
         \end{pmatrix},
  \end{equation}
  If  ${\cal A}_s, S^{1/2}, S^{-1/2}$, $S^{1/2} X$,  are stored in QRAM, then a $(\|M_{\text{HKM}}\|_F, \Ocal(\log n),
  \xi/(\kappa^2 \log^2 \frac{\kappa}{\xi}))$-block-encoding
  of \eqref{eq:hkmcompact}, can be constructed in time $\widetilde{\Ocal}_{n, \kappa, \frac{1}{\xi}}(1)$.
\end{proposition}

\begin{proof}
Follows the same proof as Proposition \ref{prop:NTcompactII}, as this is simply the case in which $P = S^{1/2}$.
\end{proof}


Having established a compact factorization of the Newton linear system
corresponding to the HKM direction, similar to our work for the NT
direction, we may now use the factorization to solve the Newton
system. The following two results establish the validity of our
factorization and bound the running time for solving this system in a
quantum setting, respectively.

\begin{proposition}
  \label{prop:hkmls}
Using the HKM direction $P = S^{1/2}$, there is a quantum algorithm, which given $$\ket{-\eta_1 R_{s_0}^d \circ -\eta_1 R_{s_0}^p  \circ \svec ( ( 1 - \eta_2) \tau \mu_0
    I - H_P(XS))}$$ and access to QRAM data structures encoding  ${\cal A}_s$, $P$, $P^{-\top}$, $PX$, and
  $P^{-\top}S$, outputs a state $\xi$-close to
  $\ket{\Delta X \circ \Delta y \circ \Delta S}$ in time
 $$\widetilde{\Ocal}_{n, \kappa, \frac{1}{\xi}}(\kappa (\|{\cal A}\|_F + \|S\|_F)).$$
We can also output an estimate of $\|\Delta
  X \circ \Delta y \circ \Delta S\|$ with relative error $\delta$ by increasing the running time by a factor of $\frac{1}{\delta}$. 
\end{proposition}
\begin{proof}
Observe, adding (i.e., a linear combination of) $M_1$, $M_2$ and $M_3$ yields 
$$
\begin{pmatrix}
         0 & \tran{\cal A}_s & \Ical \\
         {\cal A}_s & 0 & 0 \\
          S^{1/2} \otimes_s S^{1/2} & 0 & S^{1/2} X \otimes_s S^{-1/2}
         \end{pmatrix}
.$$
Then, the result follows from Theorem \ref{t:NTtII}.
\end{proof}

\subsubsection{Implementing the AHO linear system}
We now analyze the version of our II-QIPM that uses the AHO
direction.

\begin{proposition}\label{prop:AHOcompactII}
  The AHO linear system matrix can be written compactly as:
  \begin{equation}
    \label{eq:AHOcompact}
    M_{\text{AHO}} = 
    \begin{pmatrix}
         0 & \tran{\cal A}_s & \Ical \\
         {\cal A}_s & 0 & 0 \\
         I \otimes_s S^{1/2} & 0 &  X \otimes_s I
         \end{pmatrix},
  \end{equation}
  If  ${\cal A}, S^{1/2}, S^{-1/2}$, $XS^{1/2}$, $S^{1/2} X$,  are stored in QRAM, then a $(\|M_{\text{HKM}}\|_F, \Ocal(\log n),
  \xi/(\kappa^2 \log^2 \frac{\kappa}{\xi}))$-block-encoding
  of \eqref{eq:AHOcompact}, can be constructed in time $\widetilde{\Ocal}_{n, \kappa, \frac{1}{\xi}}(1)$.
\end{proposition}

\begin{proof}
The result follows the same proof as Proposition \ref{prop:NTcompactII}, as this is simply the special case in which $P = I$.
\end{proof}

\section{Full proofs of useful results from the literature}\label{app:zhou_toh}
We begin with a proof of Theorem \ref{t:t1zhou}.
\begin{proof}
Let 
\begin{equation}
\bar{\alpha}_k = \min \left(1, \frac{1}{\eta_1 (1 + \bar{\vartheta})}, \frac{0.5 \cdot (1 - \eta_2) \gamma_2 \theta_k \mu_0}{\| H_P  (\Delta X_k \Delta S_k )  \|} \right).
\end{equation}
Then, by Lemma \ref{l:l4zhou}, at the $k$-th iteration, there exists $\alpha_k$ satisfying
\begin{equation}
\alpha_k \geq \bar{\alpha}_k = \min \left(1, \frac{1}{\eta_1 (1 + \bar{\vartheta})}, \frac{c}{n^2} \right),
\end{equation}
where $c > 0$ is a constant which is independent of $\bar{\vartheta}$. Using the fact that
\begin{align*}
\theta_k &= \prod_{i =0}^{k -1} (1 - \alpha_i \eta_2) \leq (1 - \bar{\alpha} \eta_2)^k~, \text{and} \\
\frac{1}{\bar{\alpha}} &= \Ocal(n^2),
\end{align*}
it follows that $\theta_k \leq \epsilon$ for $k = \Ocal(n^2 \ln(1/\epsilon))$. 
\end{proof}

\begin{lemma}[Lemma 4 in \cite{zhou2004polynomiality}]\label{l:l4zhou}
Suppose that the conditions \eqref{e5}, \eqref{e6} and \eqref{e7} hold. Then
\begin{equation}
\| H_P (\Delta X_k \Delta S_k ) \| = \Ocal (n^2 \theta_k \mu_0).
\end{equation}
\end{lemma}

\subsection{Proof of Lemma \ref{l:l4zhou}}\label{ss:proofl}
Following the framework presented in \cite{zhou2004polynomiality}, we make use of the quantities that originally appeared in \cite{zhang1998extending}. Choosing $P = S^{1/2}$ (i.e., the HKM direction), define
\begin{subequations}
\begin{align}
Z_k &:= \Fcal_k \Ecal_k = \frac{1}{2} \left( S_k^{\frac{1}{2}} X_k S_k^{\frac{1}{2}} \otimes I + I \otimes S_k^{\frac{1}{2}} X_k S_k^{\frac{1}{2}} \right), \\
D_k &:= Z_k^{-\frac{1}{2}} \Fcal_k = Z_k^{\frac{1}{2}} \Ecal_k^{-1}.
\end{align}
\end{subequations}
Next, consider the eigenvalue decomposition of 
\begin{equation} \label{e:45}
S_k^{\frac{1}{2}} X_k S_k^{\frac{1}{2}} = Q_k \Lambda_k Q_k^{\top},
\end{equation}
where $Q_k^{\top} Q_k = I$, and $\Lambda_k = \textup{diag}(\lambda_1^k, \lambda_2^k, \dots, \lambda_n^k)$ with $\lambda_1^k \geq \lambda_2^k \geq \dots \geq \lambda_n^k > 0.$ Now, if $(\tau, \theta, X, y, X) \in \Ncal_I$, it can be shown that 
\begin{align}
(1-\gamma_2)\theta \mu_0 I &\preceq H_P (XS) = S_k^{\frac{1}{2}} X_k S_k^{\frac{1}{2}} \preceq (1+ \gamma_2) \theta \mu_0 I,
\end{align} 
and we have 
\begin{equation}\label{e:46}
(1-\gamma_2)\theta_k \mu_0 \leq \lambda_n^k \leq \cdots \leq \lambda_1^k \leq (1+ \gamma_2) \theta_k \mu_0.
\end{equation}
Then, as noted in \cite{zhou2004polynomiality}, from equations \eqref{e:45} and \eqref{e:46} the eigenvalues of $Z_k$ are defined to be
$$\left\{ \frac{1}{2} (\lambda_i^k + \lambda_j^k) : i,j = 1, 2, \dots, n \right\},$$
and thus
\begin{equation}\label{e48}
\| Z_k \|_2 \leq (1+\gamma_2) \theta_k \mu_0,~~~\| Z_k^{-1} \|_2 \leq \frac{1}{(1-\gamma_2) \theta_k \mu_0}.
\end{equation}
Next, define
\begin{equation}
G_k = \frac{1}{2} \left(  X_k^{-\frac{1}{2}} S_k^{-1} X_k^{-\frac{1}{2}} \otimes I + I \otimes X_k^{-\frac{1}{2}} S_k^{-1} X_k^{-\frac{1}{2}} \right).
\end{equation} 
Following \cite{zhou2004polynomiality}, $X_k^{\frac{1}{2}} S_k X_k^{\frac{1}{2}}$ and $S_k^{\frac{1}{2}} X_k S_k^{\frac{1}{2}}$ have the same spectrum, which implies that the eigenvalues of $G_k$ are defined to be 
$$\left\{ \frac{1}{2} \left(\frac{1}{\lambda_i^k} + \frac{1}{\lambda_j^k} \right) : i,j = 1, 2, \dots, n \right\}.$$
Therefore,
\begin{equation}
\frac{2}{(1+\gamma_2) \theta_k \mu_0} \leq \frac{1}{\lambda_i^k} + \frac{1}{\lambda_j^k}  \leq \frac{2}{(1-\gamma_2) \theta_k \mu_0} ,~ i,j = 1, 2, \dots, n .
\end{equation}

In what follows, we present lemmas from \cite{zhou2004polynomiality} which lead to a proof of Lemma \ref{l:l4zhou}. 

\begin{lemma}[Lemma 5 in \cite{zhou2004polynomiality}] \label{lemma5}
For any $\Mcal \in \R{n \times n}$
\begin{align*}
\| D_k^{-\top} \vec (\Mcal) \|^2 &\leq \frac{1}{(1-\gamma_2) \theta_k \mu_0} \|S_k^{\frac{1}{2}} \Mcal S_k^{\frac{1}{2}} \|^2 \\
\| D_k \vec (\Mcal) \|^2 &\leq \frac{1}{(1-\gamma_2) \theta_k \mu_0} \|X_k^{\frac{1}{2}} \Mcal X_k^{\frac{1}{2}} \|^2.
\end{align*}
\end{lemma}
\begin{proof}
From the definition of the matrices $\Ecal_k, \Fcal_k, G_k, Z_k$ and $D_k$, it follows
\begin{align*}
\| D_k^{-\top} \vec (\Mcal) \|^2 &= (\vec(\Mcal))^{\top} \Ecal_k Z^{-1}_k \Ecal_k \vec(\Mcal) \\
&= \vec \left( S_k^{\frac{1}{2}} \Mcal S_k^{\frac{1}{2}}  \right)^{\top}  Z^{-1}_k  \vec \left(S_k^{\frac{1}{2}} \Mcal S_k^{\frac{1}{2}}\right) \\
&\leq \frac{1}{(1-\gamma_2) \theta_k \mu_0} \left\|\vec \left( S_k^{\frac{1}{2}} \Mcal S_k^{\frac{1}{2}} \right) \right\|^2 \\
&= \frac{1}{(1-\gamma_2) \theta_k \mu_0} \left\| S_k^{\frac{1}{2}} \Mcal S_k^{\frac{1}{2}}  \right\|^2.
\end{align*}
Additionally, 
\begin{align*}
\| D_k \vec (\Mcal) \|^2 &= (\vec(\Mcal))^{\top} \Ecal_k^{-1} Z_k \Ecal_k^{-1} \vec(\Mcal) \\
&= (\vec(\Mcal))^{\top} \Ecal_k^{-1} \Fcal_k \vec(\Mcal) \\
&= \vec \left( X_k^{\frac{1}{2}} \Mcal X_k^{\frac{1}{2}} \right)^{\top} G_k \vec \left( X_k^{\frac{1}{2}} \Mcal X_k^{\frac{1}{2}} \right) \\
&\leq \frac{1}{(1-\gamma_2) \theta_k \mu_0} \left\|\vec \left( X_k^{\frac{1}{2}} \Mcal X_k^{\frac{1}{2}} \right) \right\|^2 \\
&= \frac{1}{(1-\gamma_2) \theta_k \mu_0} \left\|  X_k^{\frac{1}{2}} \Mcal X_k^{\frac{1}{2}}   \right\|^2.
\end{align*}
\end{proof}

\begin{lemma}[Lemma 6 in \cite{zhou2004polynomiality}]\label{lemma6}
Let $\kappa = \frac{\lambda_1^k}{\lambda_n^k} \leq \frac{1 + \gamma_2}{1 - \gamma_2}$. Then
\begin{align*}
&\| D_k^{-\top} \vec(\Delta X_k) \|^2 + \| D_k \vec(\Delta S_k) \|^2 + 2 \Delta X_k \cdot \Delta S_k = \| Z_k^{-\frac{1}{2}} (\vec (R_k^c) + \vec( r_k^c )) \|^2, \\
&\| H_P (\Delta X_k \Delta S_k ) \| \leq \frac{\sqrt{\kappa}}{2} \left( \| D_k^{-\top} \vec(\Delta X_k) \|^2 + \| D_k \vec(\Delta S_k) \|^2  \right).
\end{align*}
\end{lemma} 
\begin{proof}
See Lemmas 3.1 and 3.3 in \cite{zhang1998extending}.
\end{proof}

\begin{lemma}[Lemma 7 in \cite{zhou2004polynomiality}]\label{lemma7}
We have
$$ \| Z_k^{-\frac{1}{2}} (\vec (R_k^c )+ \vec( r_k^c ) ) \|^2 = \Ocal (n \theta_k \mu_0 ). $$
\end{lemma}
\begin{proof}
From \eqref{e:resbound} and \eqref{e48} it follows
\begin{align}
\| Z_k^{-\frac{1}{2}} ( \vec (r_k^c)) \|^2 &= ( \vec (r_k^c))^{\top} Z_k^{-1} \vec (r_k^c)\nonumber \\
&\leq \| Z_k^{-1} \|_2 \|r_k^c \|^2 \nonumber\\
&\leq \frac{0.25 [ (1 - \eta_2 ) \gamma_2 \theta_k \mu_0]^2}{(1 - \gamma_2) \theta_k \mu_0}\nonumber\\
 &= \frac{ [ (1 - \eta_2 ) \gamma_2]^2 \theta_k \mu_0}{4 (1 - \gamma_2)}. \label{e51}
\end{align}
Then, from equation \eqref{e:45} it follows
$$ \vec (R_k^c) = (Q_k \otimes Q_k) \vec((1-\eta_2)\theta_k \mu_0 I - \Lambda_k)$$
and hence,
\begin{align}
\| Z_k^{-\frac{1}{2}} ( \vec (R_k^c)) \|^2 &= ( \vec (R_k^c))^{\top} Z_k^{-1} \vec (R_k^c) \nonumber\\
&\leq \| Z_k^{-1} \|_2 \|\vec (R_k^c) \|^2 \nonumber\\
&\leq \frac{1}{  (1 -\gamma_2 ) \gamma_2 \theta_k \mu_0}  \sum_{i=1}^n ((1-\eta_2)\theta_k \mu_0 - \lambda^k_i)^2 \nonumber\\
&\leq \frac{1}{  (1 -\gamma_2 ) \gamma_2 \theta_k \mu_0}  \sum_{i=1}^n (|\theta_k \mu_0 - \lambda_i^k | + \eta_2 \theta_k \mu_0 )^2 \nonumber\\
&\leq \frac{n \theta_k \mu_0}{1 - \gamma_2} (\gamma_2 + \eta_2)^2. \label{e52}
\end{align}
Finally, from \eqref{e51} and \eqref{e52} it follows
\begin{align*}
\| Z_k^{-\frac{1}{2}} (\vec (R_k^c )+ \vec( r_k^c )) \|^2 &= 2 \| Z_k^{-\frac{1}{2}} ( \vec (R_k^c)) \|^2 +  2 \| Z_k^{-\frac{1}{2}} ( \vec (r_k^c)) \|^2 \\
&\leq \frac{2 n \theta_k \mu_0}{1 - \gamma_2} (\gamma_2 + \eta_2)^2 + \frac{ [ (1 - \eta_2 ) \gamma_2]^2 \theta_k \mu_0}{2 (1 - \gamma_2)}\\
&= \Ocal (n \theta_k \mu_0 ).
\end{align*}
The proof is complete.
\end{proof}

To complete the analysis, as in \cite{zhou2004polynomiality}, we consider a point $(\tilde{X}, \tilde{y}, \tilde{S})$. Note that it is guaranteed that such a point exists by Lemma \ref{lemma1}, and at iteration $k$, we face the following linear system:
\begin{subequations}
\begin{align}
\Acal^{\top} \tilde{y} + \vec( \tilde{S}) - \vec(X) &= R_0^d + r^d, \| \zeta^d_k \| \leq \gamma_1 \rho \label{e53}\\
 \Acal \vec (\tilde{X}) - b &= R_0^p + r^p,~~~~\| \Acal^{+} \zeta^p_k \| \leq \gamma_1 \rho.
\end{align}
\end{subequations}

Therefore, following \cite{zhou2004polynomiality}, by Lemma \ref{lemma1} there exists $(\tilde{X}_k, \tilde{y}_k, \tilde{S}_k)$ satisfying
\begin{subequations}
\begin{align}
\Acal^{\top} \tilde{y_k}- \vec( \tilde{S}_k) - \vec(C) &= R_0^d + \zeta^d_k \\
 \Acal  \vec( \tilde{X}_k) - b   &= R_0^p + \zeta^p_k \label{e56}\\
 (1-\gamma_1) \rho I &\preceq \tilde{X}_k \preceq (1+ \gamma_1) \rho I \label{e57}\\
  (1-\gamma_1) \rho I &\preceq \tilde{S}_k \preceq (1+ \gamma_1) \rho I.\label{e58}
\end{align}
\end{subequations}

\begin{lemma}[Lemma 8 in \cite{zhou2004polynomiality}]
Let $\Rcal_k^p, \Rcal_k^d \in \R{n \times n}$ be the unique matrices such that
\begin{subequations} \label{e59}
\begin{align}
\vec( \Rcal_k^p ) &= \Acal^+ r^p_k \\
\vec( \Rcal_k^d) &=  r^d_k.
\end{align}
\end{subequations}
Then, the following conditions hold
\begin{subequations}
\begin{align}
\left(X_k - X_* - \tau_k  (\tilde{X}_k - X_*) \right) \bullet \left(S_k - S_* - \tau_k  (\tilde{S}_k - S_*) \right) &= 0 \label{e60} \\
\left( \Delta X_k  + \eta_1 \tau_k  (\tilde{X}_k - X_*)  + \eta_1 \Rcal_k^p \right) \bullet \left( \Delta S_k  +\eta_1 \tau_k  (\tilde{S}_k - S_*) + \eta_1 \Rcal_k^d \right) &= 0. \label{e61} 
\end{align}
\end{subequations}
\end{lemma}

\begin{proof}
Given 
\begin{align*}
\Acal^{\top} y_* + \vec( S_*) &= \vec(C)\\
 \Acal \vec( X_*) &= b,
\end{align*}
it follows from \eqref{e53}-\eqref{e56} that 
\begin{align} \label{e60p}
\Acal^{\top} \left( \vec (X_k - X_* ) - \tau_k  \vec( \tilde{X}_k - X_*)   \right) &= 0 \\
\Acal \left( y_k - y_* - \tau_k  (\tilde{y}_k - y_*) \right)  + ( \vec (S_k - S_* ) - \tau_k  \vec( \tilde{S}_k - S_*)) &= 0.
\end{align}
As a consequence of \eqref{e60p} we have
\begin{equation}
\left( \vec (X_k - X_*) - \tau_k  \vec (\tilde{X}_k - X_*) \right)^{\top} \left( \vec (S_k - S_*) - \tau_k  \vec (\tilde{S}_k - S_*) \right) = 0
\end{equation}
and thus \eqref{e60} is proved. \\

Next, by \eqref{e:inexactNewt}, and \eqref{e53}-\eqref{e56} we have 
\begin{align}  
\Acal \left( \vec (\Delta X_k) + \eta_1  \tau_k \vec (\tilde{X}_k - X_*)  + \eta_1\Acal^+ r_k^p  \right) &= 0 \\ 
\Acal^{\top} \left( \Delta y_k  +\eta_1 \tau_k  (\tilde{y}_k - y_*) \right) + \left( \vec (\Delta S_k)  +\eta_1 \tau_k \vec  (\tilde{S}_k - S_*) +  \eta_1 r_k^d  \right) &= 0
\end{align}
from which \eqref{e61} follows. 
\end{proof}

Define
\begin{align}
t &= \left(  \| D_k^{-\top} \vec(\Delta X_k ) \|^2 + \| D_k \vec(\Delta S_k ) \|^2   \right)^{\frac{1}{2}} \\
\beta &= \left(  \| D_k^{-\top} \vec(\tilde{X}_k - X_*) \|^2 + \| D_k \vec(\tilde{S}_k - S_*) \|^2   \right)^{\frac{1}{2}} \\
\delta &= \left(  \| D_k^{-\top} \Acal^+ r^p_k\|^2 + \| D_k r^d_k ) \|^2   \right)^{\frac{1}{2}},
\end{align}
then from \cite{zhou2004polynomiality}, we have the subsequent lemma.

\begin{lemma}[Lemma 9 in \cite{zhou2004polynomiality}]\label{lemma9}
$$t \leq 2 \eta_1 (\tau_k \beta + \delta) + \sqrt{\varrho} $$
where $$\varrho = \| Z_k^{-\frac{1}{2}} \vec (R_k^c) + \vec( r_k^c )) \|^2 + 2 (\eta_1 \tau_k)^2  (\tilde{X}_k - X_*) \bullet (\tilde{S}_k - S_*) + 2 \eta_1^2 (\tau_k \beta + \delta) \delta. $$
\end{lemma} 

\begin{proof}
It follows from \eqref{e61}
\begin{align*}
\Delta X_k \bullet \Delta S_k &= -\eta_1 \tau_k \left( ( \Delta X_k \bullet (\tilde{S}_k - S_*) + (\tilde{X}_k - X_*) \bullet \Delta S_k \right)  \\
&= - (\eta_1 \tau_k)^2 (\tilde{X}_k - X_*)  \bullet  (\tilde{S}_k - S_*)  - \eta_1 (\Delta X_k \bullet \Rcal_k^p + \Rcal^d \Delta S_k ) \\
&= -\eta^2_1 \tau_k \left( (\tilde{X}_k - X_*)\bullet \Rcal_k^d + \Rcal^p \bullet  (\tilde{S}_k - S_*) \right) - \eta_1^2 \Rcal_k^p \bullet \Rcal_k^d \\
&\geq - \eta_1 \tau_k \beta t - \eta_1 \delta t - (\eta_1 \tau_k)^2  (\tilde{X}_k - X_*) \bullet (\tilde{S}_k - S_*) - \eta_1^2 \tau_k \beta \delta - \frac{1}{2} \eta_1^2 \delta^2 \\
&\geq - \eta_1 (\tau_k \beta  + \delta ) t - (\eta_1 \tau_k)^2  (\tilde{X}_k - X_*) \bullet (\tilde{S}_k - S_*) - \eta_1^2 \tau_k \beta \delta - \eta_1^2 (\tau_k \beta  + \delta ) \delta.
\end{align*}
Applying Lemma \ref{lemma6} in addition to the above inequality yields
$$t^2 - 2 \eta_1 (\tau_k \beta  + \delta )t - \varrho \leq 0.$$
Therefore, 
\begin{align*}
t &\leq \eta_1 (\tau_k \beta  + \delta ) + \sqrt{\eta_1^2 (\tau_k \beta  + \delta )^2 + \varrho} \\
&\leq 2 \eta_1 (\tau_k \beta  + \delta ) + \sqrt{\varrho}.
\end{align*}
\end{proof}

\begin{lemma}[Lemma 10 in \cite{zhou2004polynomiality}]\label{lemma10}
We have 
$$ \delta^2 = \Ocal (n^2 \theta_k \mu_0).$$
\end{lemma}

\begin{proof}
Recall, from \eqref{e:resbound}, we have
\begin{equation}
\| \Acal^+ r_k^d \| \leq \tau_k \gamma_1 \rho,~\text{and}~\|r_k^p\| \leq \tau_k \gamma_1 \rho.
\end{equation}
Applying Lemma \ref{lemma5}, \eqref{e59} and using the fact that $\| \Mcal \|_F \leq \trace{(\Mcal)}$ for $\Mcal \in \Scal_+^n$ it follows
\begin{align*}
 \| D_k^{-\top} \Acal^+ r^p_k\|^2 &\leq \frac{1}{(1 - \gamma_2) \theta_k \mu_0} \|S_k \|^2 \|\Acal^+ r^p_k \|^2 \\
 &\leq \frac{1}{(1 - \gamma_2) \theta_k \mu_0} (\trace{(S_k)})^2 \|\Acal^+ r^p_k \|^2 \\
  &\leq \frac{\gamma_1^2 \rho^2}{(1 - \gamma_2) \theta_k \mu_0} \tau_k^2(\trace{(S_k)})^2  \\
  &\leq \frac{1}{(1 - \gamma_2) \theta_k \mu_0} \Ocal(n^2 \theta_k^2 \rho^4) \\
  &= \Ocal(n^2 \theta_k \mu_0).
\end{align*}
Additionally, for $ \| D_k  r^d_k \|^2$, we have
\begin{align*}
 \| D_k  r^d_k \|^2 &\leq \frac{1}{(1 - \gamma_2) \theta_k \mu_0} \|X_k \|^2 \| r^d_k \|^2 \\
 &\leq \frac{1}{(1 - \gamma_2) \theta_k \mu_0} (\trace{(X_k)})^2 \| r^d_k \|^2 \\
  &\leq \frac{\gamma_1^2 \rho^2}{(1 - \gamma_2) \theta_k \mu_0} \tau_k^2(\trace{(X_k)})^2  \\
  &\leq \frac{1}{(1 - \gamma_2) \theta_k \mu_0} \Ocal(n^2 \theta_k^2 \rho^4) \\
  &= \Ocal(n^2 \theta_k \mu_0).
\end{align*}
Thus, $\delta^2 = \Ocal(n^2 \theta_k \mu_0)$.
\end{proof}

\begin{lemma}[Lemma 11 in \cite{zhou2004polynomiality}]\label{lemma11}
Under the conditions \eqref{e5}, \eqref{e6} and \eqref{e7},
$$ (\tilde{X}_k - X_*) \bullet (\tilde{S}_k - S_*) = \Ocal (n \mu_0).$$
\end{lemma}

\begin{proof}
Applying \eqref{e57} and \eqref{e58} and using the fact that $\tilde{X}_k, \tilde{S}_k, X_k, S_k \in \Scal_+^n$ we have
\begin{align*}
(\tilde{X}_k - X_*) \bullet (\tilde{S}_k - S_*)  &= \tilde{X}_k \bullet \tilde{S}_k - \tilde{X}_k \bullet S_* - X_* \bullet \tilde{S}_k + X_* \bullet S_* \\
&\leq \tilde{X}_k \bullet \tilde{S}_k \\
&\leq (1 + \gamma_2 )^2 n \rho^2 \\
&= (1 + \gamma)^2 n \mu_0 = \Ocal (n \mu_0).
\end{align*}
\end{proof}

\begin{lemma}[Lemma 12 in \cite{zhou2004polynomiality}]\label{lemma12}
Under the conditions \eqref{e5}, \eqref{e6} and \eqref{e7},
$$\tau_k^2 \beta^2 = \Ocal (n^2 \theta_k \mu_0).$$
\end{lemma}

\begin{proof}
Applying Lemma \ref{lemma5} and using the fact that $\tilde{X}_k - X_* \succeq 0$ it follows
\begin{align*}
 \| D_k^{-\top}  \vec(\tilde{X}_k - X_*) \|^2 &\leq \frac{1}{(1 - \gamma_2) \theta_k \mu_0} \| S_k^{\frac{1}{2}} (\tilde{X}_k - X_*) S_k^{\frac{1}{2}} \|^2 \\
 &\leq  \frac{1}{(1 - \gamma_2) \theta_k \mu_0} \trace{\left( S^{\frac{1}{2}} ( \tilde{X}_k - X_*) S_k^{\frac{1}{2}}  \right)^2} \\
 &= \frac{1}{(1 - \gamma_2) \theta_k \mu_0}  \left(  ( \tilde{X}_k - X_*) \bullet S_k \right)^2,
\end{align*}
and hence
\begin{align*}
 \| D_k \vec(\tilde{S}_k - S_*) \|^2 &\leq \frac{1}{(1 - \gamma_2) \theta_k \mu_0}  \left( X_k \bullet (\tilde{S}_k - S_*)  \right)^2.
\end{align*}
Therefore
$$ \tau_k \beta \leq \frac{1}{\sqrt{(1-\gamma_2) \theta_k \mu_0}} \tau_k \left( (\tilde{X}_k - X_*) \bullet S_k + X_k \bullet (\tilde{S}_k - S_* ) \right).$$
Next, following \cite{zhou2004polynomiality} we use \eqref{e60} and that $X_* \bullet S_* = 0$, so we have
\begin{align*}
\tau_k &\left( (\tilde{X}_k - X_*) \bullet S_k + X_k \bullet (\tilde{S}_k - S_* ) \right) \\
&= X_k \bullet S_k - X_k \bullet S_* - X_* \bullet S_k + \tau_k (\tilde{Y}_k \bullet S_* + X_* \bullet \tilde{S}_k)\\ 
&~~~+ \tau_k^2 (\tilde{X}_k \bullet \tilde{S}_k - \tilde{X}_k \bullet S_* - \tilde{X}_* \bullet \tilde{S}_k ) \\
&\leq X_k \bullet S_k + nu_k (\tilde{X}_k \bullet S_* + X_* \bullet \tilde{S}_k) + \tau_k^2 \tilde{X}_k \bullet \tilde{S}_k \\
&\leq (1 + \gamma_2) n \theta_k \mu_0 + \theta_k (1 + \gamma_1) \rho (I \bullet S_* + X_* \bullet I) + \theta_k^2 (1 + \gamma_1)^2 n \mu_0 \\
&\leq 6n \theta_k \mu_0 + 2 \tau_k \rho \left( \trace{(X_*)} + \trace{(S_*)} \right) \\
&\leq 6n \theta_k \mu_0 + 2 \tau_k \rho^2 \\
&= 8 n \theta_k \mu_0. 
\end{align*}
Thus, 
$$\theta_k^2 \beta^2 \leq \frac{64 n^2 \theta_k \mu_0}{1 - \gamma_2} = \Ocal(n^2 \theta_k \mu_0),$$
which completes the proof. 
\end{proof}

\begin{proof}{Proof of Lemma \ref{l:l4zhou}}
Applying Lemma \ref{lemma9} and using the fact that $(a + b)^2 \leq 2a^2 + 2b^2$, we have
\begin{align*}
t^2 &\leq  (2 (\tau_k \beta + \delta) + \sqrt{\varrho})^2 \leq 8 (\tau_k \beta + \delta)^2 + 2 \varrho \\
\varrho &\leq \| Z_k^{-\frac{1}{2}} (\vec (R_k^c) + \vec (r_k^c) ) \|^2 + 2 \tau_k^2 (\tilde{X}_k - X_*) \bullet (\tilde{S}_k - S_*) + 2  (\tau_k \beta + \delta) \delta.
\end{align*}
Hence,
\begin{align*}
t^2 &\leq 8 (\tau_k \beta + \delta)(\tau_k \beta +  2 \delta) + 2 \| Z_k^{-\frac{1}{2}} (\vec (R_k^c) + \vec (r_k^c) ) \|^2 + 4 \tau_k^2 (\tilde{X}_k - X_* ) \bullet (\tilde{S}_k - S_* ) \\
&\leq \Ocal(n^2 \theta_k \mu_0) + \Ocal(n \theta_k \mu_0)  + \Ocal(n \theta_k^2 \mu_0). 
\end{align*}
Note that, as in \cite{zhou2004polynomiality}, the last inequality results from applying Lemmas \ref{lemma7}, \ref{lemma10}, \ref{lemma11} and \ref{lemma12}. Finally, from Lemma \ref{lemma6}, it follows that
$$ \| H_P (\Delta X_k \Delta S_k ) \| \leq \frac{1}{2} \sqrt{\frac{1 + \gamma_2}{1 - \gamma_2}} t^2 = \Ocal(n^2 \theta_k \mu_0),$$
and the proof of Lemma \ref{l:l4zhou} is complete. 
\end{proof}

\bibliographystyle{quantum}
\bibliography{sdpbib}

\end{document}